\documentclass[11pt]{article}
\usepackage{xurl}
\usepackage{hyperref}
\hypersetup{
	colorlinks=true,
	linkcolor=blue,
	filecolor=blue,      
	urlcolor=red,
	citecolor=blue,
}

\usepackage{amsfonts}
\usepackage{amsgen,amsmath,amstext,amsbsy,amsopn,amssymb,subfigure, amsthm}
\usepackage[dvips]{graphicx}
\usepackage{comment}
\usepackage{enumitem}
\usepackage{natbib}
\usepackage{booktabs}
\usepackage{blkarray}
\usepackage{algorithm}
\usepackage{algpseudocode}
\usepackage{url}
\usepackage{longtable}
\usepackage{mleftright}

\usepackage{multirow}
\usepackage[usenames]{color}

\allowdisplaybreaks

\newtheorem{Theorem}{Theorem}
\newtheorem*{Theorem*}{Theorem}
\newtheorem{Lemma}{Lemma}
\newtheorem{Remark}{Remark}
\newtheorem{Corollary}{Corollary}

\newtheorem*{Assumption*}{Assumption}
\newtheorem{Proposition}{Proposition}

\newcommand{\be}{\begin{equation}}
	\newcommand{\ee}{\end{equation}}
 \newcommand{\bs}{\begin{split}}
	\newcommand{\es}{\end{split}}
\newcommand{\bea}{\begin{eqnarray}}
	\newcommand{\eea}{\end{eqnarray}}
\newcommand{\beas}{\begin{eqnarray*}}
	\newcommand{\eeas}{\end{eqnarray*}}

\newcommand{\EE}{\mathbb{E}}

\DeclareMathAlphabet\mathbfcal{OMS}{cmsy}{b}{n}

\newcommand{\bH}{{\mathbf{H}}}
\newcommand{\X}{{\mathbf{X}}}

\newcommand{\M}{{\mathbf{M}}}

\renewcommand{\P}{{\mathbf{P}}}
\newcommand{\Q}{{\mathbf{Q}}}
\newcommand{\R}{{\mathbf{R}}}
\renewcommand{\S}{{\mathbf{S}}}

\newcommand{\Var}{{\rm Var}}

\newcommand{\W}{{\mathbf{W}}}
\newcommand{\A}{{\mathbf{A}}}
\newcommand{\Y}{{\mathbf{Y}}}
\newcommand{\Z}{{\mathbf{Z}}}

\newcommand{\Span}{\operatorname{span}}

\newcommand{\rank}{{\rm rank}}

\newcommand{\diag}{{\rm diag}}

\newcommand{\SVD}{{\rm SVD}}

\newcommand{\argmin}{\mathop{\rm arg\min}}

\newcommand{\RR}{\mathbb{R}}
\newcommand{\PP}{\mathbb{P}}
\newcommand{\MM}{\mathcal{M}}
\newcommand{\KL}{\operatorname{KL}}

\allowdisplaybreaks

\begin{document}

\title{Optimal Estimation of Discrete Multiview Distributions under Heteroskedastic Multinomial Sampling}

\author{
Runshi Tang\thanks{Department of Statistics, University of Wisconsin-Madison, USA},
~
Julien Chhor$^1$\thanks{Toulouse School of Economics, France},
~
Olga Klopp$^1$\thanks{ESSEC Business School, France},
~
Alexandre B. Tsybakov$^1$\thanks{CREST, ENSAE, Institut Polytechnique de Paris, France},
\\
~
and
~
Anru R. Zhang$^1$\thanks{Department of Biostatistics \& Bioinformatics and Department of Computer Science, Duke University, USA}
}
\date{}
\maketitle
\footnotetext[1]{The marked authors are listed in alphabetical order.}

\begin{abstract}
Multiview latent-variable models provide a fundamental framework for discrete data analysis, with applications to latent structure models, topic models, and mixtures of product distributions. In the discrete setting, the joint distribution of the observed views can be represented as a nonnegative low-rank tensor, which we call a multiview density tensor. We study the problem of estimating this tensor from multinomial count data. A key challenge is that multinomial sampling induces heteroskedastic and dependent noise, so the difficulty of estimation depends not only on the ambient dimensions and rank, but also on how the probability mass is distributed across different locations of sample space.

We propose a general scaling framework for density tensor estimation under multinomial sampling. This framework leads to a spectral estimator for which we prove a Frobenius-norm upper bound that directly handles heteroskedasticity and negative dependence. For the original multiview model, we obtain fiber-mass-dependent Frobenius upper bounds and minimax lower bounds showing that this dependence is unavoidable. 
Under $\ell_1$ loss, we develop both oracle and feasible data-driven estimators based on the same scaling principle, establish minimax lower bounds, and show near-optimality for the oracle rule at fixed rank and for slice normalization under bounded slice-to-fiber imbalance.
Simulations support the theory and demonstrate the robustness of the proposed methods.
\end{abstract}

\begin{sloppypar}
\section{Introduction}

Multiview latent-variable models are a basic tool in modern statistics and machine learning. In these models, each observation consists of several conditionally independent views generated from an unobserved latent class. Such models are closely related to latent structure models, topic models, and other mixture-based representations of discrete data \citep{AllmanMatiasRhodes2009,hofmann1999,BleiNgJordan2003,anandkumar2014tensor,chhor2024generalized}. When the views are discrete, the joint distribution of the observed variables can be represented as a nonnegative tensor with a low-rank CP structure, linking the problem to the literature on tensor decompositions, identifiability, and latent-variable learning \citep{kolda2009tensor,kruskal1977three,anandkumar2014tensor,jain2014learning,chhor2024generalized}. This viewpoint turns multiview density estimation into a structured tensor estimation problem, in which both the low-rank geometry and the multinomial sampling mechanism play central roles.

In this paper, we study estimation of a \emph{multiview density tensor} $\P$ from multinomial count data. Concretely, $\P$ is a probability tensor on $[p_1]\times\cdots\times[p_d]$ admitting a low-rank decomposition $\P=\sum_{r=1}^R w_r\,a_r^{(1)}\circ\cdots\circ a_r^{(d)},$
where $w=(w_1,\ldots,w_R)$ is a probability vector and each factor $a_r^{(k)}$ is a probability vector. Given i.i.d.\ samples from this distribution, the natural observation is the count tensor $\Y\sim\operatorname{Multinomial}(n,\P)$. 
The tensor $\P$ has two complementary interpretations. First, it can be viewed as a mean tensor, for which the Frobenius norm is a natural and commonly used error metric; we refer to this task as \emph{Frobenius estimation}. Second, $\P$ is also a density tensor, for which the $\ell_1$ norm is proportional to the total variation distance; we refer to this task as \emph{$\ell_1$ estimation}. A main theme of this paper is that these two estimation goals are governed by different aspects of the multinomial tensor geometry. This viewpoint is closely related to recent work on low-rank density estimation and generalized multiview models for discrete and continuous data \citep{Jain2020linear,chhor2024generalized,Vandermeulen2020ImprovingND,vandermeulen2021beyond}. 

A central difficulty is that multinomial sampling does not produce homoskedastic additive noise. Instead, the noise is intrinsically \emph{heteroskedastic}: entries with larger probability mass have larger variance, while rare cells are much less noisy. In addition, multinomial counts are negatively dependent because the total count is fixed. Consequently, the statistical difficulty is governed not only by the ambient dimensions and rank, but also by how the probability mass of $\P$ is distributed. Standard low-rank tensor methods, which are typically developed under homoskedastic noise models, do not directly capture this signal-dependent variance structure \citep{zhang2018tensor,tang2025revisit}. Recent progress on heteroskedastic spectral estimation has focused mainly on matrix and covariance problems \citep{zhang_heteroskedastic_2021,yan2024inference,zhou_deflated_2024}, while classical nonparametric and semiparametric mixture models under conditional independence have largely emphasized identifiability and component estimation rather than low-rank estimation of the full multinomial probability tensor \citep{HallZhou2003,BenagliaChauveauHunter2009,ChauveauHunterLevine2015,XiangWangYao2019,aragam2023uniform}.

Our approach is based on a simple principle: instead of estimating $\P$ directly, we first apply a rank-one scaling to stabilize the heteroskedasticity of the multinomial model, and then estimate the resulting scaled tensor using spectral methods. This yields a general reduction from $\ell_1$ estimation to Frobenius estimation. 
It also identifies the fiber masses of the scaled variance profile, i.e., the largest marginal mass along one-dimensional tensor fibers, as the key quantities governing statistical difficulty. On the methodological side, this principle leads to both oracle and data-driven scaling rules. On the theoretical side, it yields upper and lower bounds that adapt to the local geometry of the target distribution. 

\subsection{Main contributions}



Our contributions are fourfold.

\noindent
{\bf (i) A general scaled estimation framework.}
We formulate a rank-one scaling framework for multiview multinomial tensors. The framework makes explicit the tradeoff between stabilizing the multinomial variance through scaling and controlling the error amplification incurred when transforming the estimator back to the original density scale.

\noindent
{\bf (ii) Spectral estimation for heteroskedastic multinomial tensors.}
We develop a spectral estimator based on Deflated-HeteroPCA and mode-wise subspace refinement. We prove a Frobenius norm upper bound for scaled low Tucker-rank tensors under multinomial noise, explicitly accounting for both heteroskedasticity and dependence.

\noindent
{\bf (iii) Minimax theory for multiview density tensors.}
Specializing the general theorem to multiview CP tensors and combining it with the minimax lower bound yields the optimal Frobenius estimation rate $\sqrt{R\operatorname{Fiber}_{\ell_1}(\P)/n}$ for sufficiently large $n$, where $\operatorname{Fiber}_{\ell_1}(\P)$ denotes the largest $\ell_1$ mass of any fiber of $\P$. For $\ell_1$ estimation, we establish the minimax lower bound $\sqrt{R p_{\max}/n}$, where $p_{\max}=\max_{i\in[d]} p_i$. We further use the proposed framework to develop both oracle and feasible data-driven estimators, with error bounds $R^d\sqrt{p_{\max}/n}$ and $\sqrt{\xi p_{\max}/n}$, respectively, where $\xi$ is a fiber-slice balance parameter introduced in Section~\ref{sec:slice-scaling}. The oracle estimator is optimal when $R$ is bounded, whereas the feasible estimator is optimal when $\xi$ is bounded.

\noindent
{\bf (iv) New concentration tools for multinomial tensor noise.}
A key technical contribution is a direct analysis of the heteroskedastic and dependent noise induced by multinomial sampling. 
Unlike analyses based on Poissonization (e.g., \cite{chhor2024generalized}), which replace multinomial counts by independent Poisson variables, our estimator and perturbation bounds involve nonlinear functions of the observed tensor and therefore require direct control of the fixed-sample multinomial noise. We decompose the centered histogram into a sum of independent one-sample multinomial tensors, reduce the off-diagonal Gram-matrix term arising in the estimator to a matrix-valued $U$-statistic, and control it through decoupling and matrix Bernstein-type inequalities.

\subsection{Related Work}

Our work is related to several strands of literature on latent-variable models, density estimation, and spectral methods for structured high-dimensional data.

First, multiview latent-variable models and mixtures of product distributions have long been studied in statistics and machine learning. Classical latent semantic models such as probabilistic latent semantic indexing and latent Dirichlet allocation provide prominent examples in text analysis \citep{hofmann1999,BleiNgJordan2003}. On the statistical side, latent structure and latent class models have been analyzed from the viewpoint of identifiability and algebraic structure; in particular, \cite{AllmanMatiasRhodes2009} established general identifiability results for latent structure models with many observed variables, building on tensorial representations and Kruskal-type arguments. Tensor methods later became a standard tool for learning latent-variable models and mixtures of product distributions; see, for example, \cite{anandkumar2014tensor,jain2014learning,feldman2008learning,freund1999estimating,chaudhuri2008learning,gordon2021source,rabani2014learning,li2015learning}. These works focus primarily on identifiability, decomposition, and learning of latent structures, whereas our goal is minimax estimation of the full joint probability tensor under multinomial sampling and explicit loss functions.

Second, our paper is related to a broad literature on nonparametric and semiparametric finite mixtures. Under conditional independence assumptions, early identifiability and estimation results for multivariate mixtures were developed by \cite{HallZhou2003}. Semiparametric and nonparametric estimation procedures for multivariate mixtures were studied in \cite{laird1978nonparametric, lesperance1992algorithm, BenagliaChauveauHunter2009,LevineHunterChauveau2011,ChauveauHunterLevine2015,kasahara2014non}, and a broader overview is given in \cite{XiangWangYao2019}. Operator-theoretic and grouped-sample perspectives on nonparametric mixture models were developed in \cite{vandermeulen2019operator,RitchieVandermeulenScott2020,VandermeulenSaitenmacher2024,Vandermeulen2023}. These papers clarify identifiability and statistical estimation for nonparametric mixtures in considerable generality, but they do not exploit low-rank tensor structure induced by multiview product representations, nor do they address multinomial tensor observations with heteroskedastic count noise.

Third, our work is closely connected to recent developments in low-rank density estimation. In the matrix case, \cite{Jain2020linear} showed that low-rank discrete distributions can be learned with sample complexity essentially linear in the ambient dimension and rank. More recently, \cite{chhor2024generalized} studied adaptive density estimation under low-rank constraints in a generalized multi-view model, covering both discrete and continuous bivariate settings. Other low-rank or tensor-based approaches to density estimation and mixture modeling include \cite{Vandermeulen2020ImprovingND,vandermeulen2021beyond,zheng2020nonparametric,pmlr-v89-kargas19a,9779133,9740538}. Among these papers, \cite{chhor2024generalized} is closest to our work in spirit. However, its primary emphasis is the bivariate setting, whereas we study higher-order multiview tensors under multinomial observations, develop a general scaled estimation principle, and derive Frobenius and $\ell_1$ norm bounds driven by fiber and slice masses.

Fourth, our methodology builds on the literature on spectral estimation under heteroskedastic noise. General background on spectral methods from a statistical perspective can be found in \cite{chen_spectral_2021}. For heteroskedastic matrix problems, \cite{zhang_heteroskedastic_2021} introduced HeteroPCA, and \cite{zhou_deflated_2024} developed Deflated-HeteroPCA to overcome ill-conditioning while retaining near-optimal guarantees. Perturbation analysis for singular subspaces also plays an important role in our proofs; see \cite{cai2018rate}. These works provide key algorithmic and technical ingredients for our procedure, but they do not address multinomial tensor observations, where the noise is both heteroskedastic and negatively dependent across entries. Moreover, some existing heteroskedastic tensor guarantees lead to bounds depending mainly on $\|\X\|_\infty$ (e.g., \cite[Corollary 2]{zhou_deflated_2024} and \cite{agterberg2024statistical}), which may fail to reflect the full heteroskedastic variance profile. Our analysis instead tracks fiber and slice masses, which are essential for obtaining adaptive rates.

Finally, our work is related to recent efforts to model probability tensors and count tensors directly. Signal-processing motivated approaches include estimation of joint distributions from partial marginals or incomplete observations \citep{yeredor2019maximum,ibrahim2020recovering,ibrahim2021recovering,vora2021recovery,KargasSidiropoulosFu2018}, as well as characteristic-tensor methods for low-rank density estimation \citep{9779133,9740538}. In a negative dependence different count-data setting, \cite{xu2025multivariate} studies low-rank tensor methods for multivariate Poisson intensity estimation by discretizing a continuous intensity function and representing the resulting coefficients through a low Tucker-rank tensor. Our setting is different: we observe multinomial histograms from a probability tensor, and the heteroskedasticity is intrinsic to the sampling distribution. On the tensor side, statistical and computational theory for low-rank estimation under homoskedastic noise has been developed in \cite{zhang2018tensor,bhaskara2014smoothed,ma2016polynomial,tang2025revisit}. Our problem differs from these works in two key respects: the observations are multinomial counts rather than generic tensor entries or homoskedastic perturbations, and our goal is to understand how heteroskedasticity interacts with low-rank structure under Frobenius and $\ell_1$ losses. The rank-one scaling framework developed here is designed precisely for this setting.

\subsection{Organization}

The rest of the paper is organized as follows. Section~\ref{sec:problem-setup} introduces notation, basic tensor algebra, the multiview multinomial model, and the loss functions considered in this paper. Section~\ref{sec:estimators} presents the proposed estimators and scaling principles. 
Section~\ref{sec:general-theory} develops the main technical theory: a Tucker-based upper bound under the Frobenius norm for scaled low-rank tensors under multinomial sampling. Section~\ref{sec:frobenius-estimation} specializes this result to the Frobenius estimation and proves the corresponding minimax lower bound. Section~\ref{sec:l1-density-estimation} studies density $\ell_1$ estimation, including oracle scaling, feasible slice normalization, and the minimax comparison. Section~\ref{sec:simulations} presents numerical experiments. Technical proofs are deferred to the Supplementary Materials.

\section{Problem Setup and Preliminaries}\label{sec:problem-setup}

\subsection{Notation}

For a positive integer $m$, write $[m]=\{1,\ldots,m\}$. For each $i$, let $e_i$ denote the $i$th canonical basis vector. For a matrix $U\in\RR^{p\times r}$, write $U\in\mathbb{O}_{p,r}$ if its columns are orthonormal. We use $C$ and $c$ to denote positive absolute constants whose values may change from line to line, and use $C_\alpha$ to denote a positive constant depending only on the parameter $\alpha$. We write $a\lesssim b$ if there exists an absolute constant $C>0$ such that $a\le Cb$. We write $a\gtrsim b$ if $b\lesssim a$, and $a\asymp b$ if both $a\lesssim b$ and $a\gtrsim b$ hold.

For a matrix $M\in\RR^{p\times q}$ and an integer $r\le \min\{p,q\}$, $\SVD_r(M)$ denotes the $p\times r$ matrix whose columns are the top $r$ left singular vectors of $M$, and $\sigma_r(M)$ denotes the $r$th largest singular value of $M$. We use $P_{\mathrm{off\text{-}diag}}(M)$ to denote the projection of a square matrix $M$ onto its off-diagonal entries, i.e., $\bigl(P_{\mathrm{off\text{-}diag}}(M)\bigr)_{ii}=0$ and $\bigl(P_{\mathrm{off\text{-}diag}}(M)\bigr)_{ij}=M_{ij}$ for $i\neq j$.
For a matrix $U$, define
$\|U\|_{2,\infty} := \max_i \|U_{i,\cdot}\|_2$.
For vectors $a_k \in \mathbb{R}^{p_k}$, $k \in [d]$, we write $\bigotimes_{k=1}^d a_k$
for their outer product, namely the tensor in
$\mathbb{R}^{p_1 \times \cdots \times p_d}$ with entries
\[
    \left(\bigotimes_{k=1}^d a_k\right)_{i_1,\ldots,i_d}
    = \prod_{k=1}^d (a_k)_{i_k}.
\]

Other notation will be introduced at its first appearance.

\subsection{Basic Tensor Algebra}\label{sec_tensor_algebra}

We briefly review the tensor notation used throughout the paper. Let $\X\in\RR^{p_1\times\cdots\times p_d}$ be an order-$d$ tensor, and write $p_{\max}=\max_{k\in[d]}p_k$. For any order-$d$ tensors $\X$ and $\Y$ of the same dimension, $*$ denotes the Hadamard (elementwise) product, defined by $(\X * \Y)_{i_1,\ldots, i_d} = \X_{i_1,\ldots, i_d} \cdot \Y_{i_1,\ldots, i_d}$. For vectors $u_k\in\RR^{p_k}$, $k\in[d]$, the outer product $u_1\circ\cdots\circ u_d$ is defined entrywise by $(u_1\circ\cdots\circ u_d)_{i_1,\ldots,i_d}=\prod_{k=1}^d (u_k)_{i_k}$. We use $\|\X\|_1$, $\|\X\|_F$, and $\|\X\|_\infty$ for the entrywise $\ell_1$, Frobenius, and entrywise maximum norms, respectively. Denote by $\X^{(-1)}$ the entrywise inverse of $\X$, defined by $(\X^{(-1)})_{i_1,\ldots, i_d} = 1/\X_{i_1,\ldots, i_d}$.

For $k\in[d]$, let $p_{-k}=\prod_{h\neq k}p_h$. We write $\MM_k(\X)\in\RR^{p_k\times p_{-k}}$ for the mode-$k$ matricization of $\X$, whose rows are indexed by the $k$th coordinate and whose columns are indexed by all remaining coordinates. For a matrix $A\in\RR^{m\times p_k}$, the mode-$k$ product $\X\times_k A$ is the tensor in $\RR^{p_1\times\cdots\times p_{k-1}\times m\times p_{k+1}\times\cdots\times p_d}$ defined by
\[
    (\X\times_k A)_{i_1,\ldots,i_{k-1},j,i_{k+1},\ldots,i_d}
    =
    \sum_{\ell=1}^{p_k}
    \X_{i_1,\ldots,i_{k-1},\ell,i_{k+1},\ldots,i_d}
    A_{j\ell}.
\]
For products over multiple modes, we write $\X\times_{k\in[d]} A_k=\X\times_1 A_1\cdots\times_d A_d$, and similarly $\X\times_{h\neq k} A_h$ for products over all modes except $k$.

The CP rank of $\X$, denoted by $\rank(\X)$, is the smallest integer $R$ such that
\[
    \X=\sum_{r=1}^R \lambda_r\,
    a_r^{(1)}\circ\cdots\circ a_r^{(d)}
\]
for some scalars $\lambda_r\in\RR$ and vectors $a_r^{(k)}\in\RR^{p_k}$. We also use Tucker decompositions. We say that $\X$ has Tucker rank at most $(r_1,\ldots,r_d)$ if it can be written as $\X=\S\times_{k\in[d]} U_k$, where $\S\in\RR^{r_1\times\cdots\times r_d}$ is the core tensor and $U_k\in\mathbb{O}_{p_k,r_k}$ for each $k\in[d]$. Equivalently, the Tucker rank of $\X$ is $\rank_{\operatorname{Tucker}}(\X)=\bigl(\operatorname{rank}(\MM_1(\X)),\ldots,\operatorname{rank}(\MM_d(\X))\bigr)$. When $\rank_{\operatorname{Tucker}}(\X)=(r_1,\ldots,r_d)$, define the Tucker singular value by $\lambda_{\operatorname{Tucker}}(\X)=\min_{k\in[d]}\sigma_{r_k}(\MM_k(\X))$. This quantity measures the weakest mode-wise signal strength and will appear in the signal-to-noise conditions below.

\subsection{The Multiview Multinomial Model}

We call a tensor $\P\in\RR^{p_1\times\cdots\times p_d}$ a \emph{density tensor} if $\P_{i_1\cdots i_d}\ge 0$ and $\|\P\|_1=1$. We call $\P$ a \emph{multiview density tensor} of rank at most $R$ if it admits the representation
\begin{equation}\label{eq:P-decomposition}
\P
=
\sum_{r=1}^R w_r\, a_r^{(1)} \circ \cdots \circ a_r^{(d)},
\end{equation}
where $w=(w_1,\ldots,w_R)$ is a probability vector and each $a_r^{(k)}\in\RR^{p_k}$ is a probability vector. Thus, $\rank(\P)\le R$, with equality when the representation is minimal.

The decomposition \eqref{eq:P-decomposition} has a latent-class interpretation. For each observation, draw $Z\in[R]$ with $\PP(Z=r)=w_r$. Conditional on $Z=r$, draw the $d$ coordinates independently from $a_r^{(1)},\ldots,a_r^{(d)}$. Then, for one observation $(X^{(1)},\ldots,X^{(d)})\in[p_1]\times\cdots\times[p_d]$, we have
\[
\PP(X^{(1)}=i_1,\ldots,X^{(d)}=i_d)
=
\sum_{r=1}^R w_r\prod_{k=1}^d a_r^{(k)}(i_k)
=
\P_{i_1\cdots i_d}.
\]

Given i.i.d.\ samples $(X_1^{(1)},\ldots,X_1^{(d)}),\ldots,(X_n^{(1)},\ldots,X_n^{(d)})$, identify each sample with the canonical basis tensor $\bH_\ell=e_{X_\ell^{(1)}}\circ\cdots\circ e_{X_\ell^{(d)}}\in\RR^{p_1\times\cdots\times p_d}$ for $\ell\in[n]$. The empirical histogram tensor is $\Y=\sum_{\ell=1}^n \bH_\ell$, so that $\Y_{i_1\cdots i_d}$ counts the number of observations falling in cell $(i_1,\ldots,i_d)$. We write $\Y\sim\operatorname{Multinomial}(n,\P)$, with $\EE(\Y)=n\P$, and refer to this as a \emph{tensor multinomial distribution}. When $\P$ is a multiview density tensor, we call this a \emph{multiview multinomial distribution}. 

\subsection{Losses and Heteroskedasticity Measures}

We consider two loss functions. 
Since the mean histogram is $n\P$, estimating the mean structure is equivalent, after normalization by $n$, to estimating $\P$ under Frobenius loss $\|\hat\P-\P\|_F$. We term this the \emph{Frobenius estimation}. 
On the other hand, $\P$ itself is the underlying density tensor. For density estimation, a natural loss is the entrywise $\ell_1$ loss
\[
    \|\hat\P-\P\|_1
    =
    \sum_{i_1,\ldots,i_d}
    |\hat \P_{i_1\cdots i_d}-\P_{i_1\cdots i_d}|,
\]
which is proportional to the total variation distance. We term this as the \emph{$\ell_1$ estimation}. 

To quantify the heteroskedasticity of the multinomial model, we use fiber and slice masses. For $\X\in\RR^{p_1\times\cdots\times p_d}$, define
\[
\operatorname{Fiber}_{\ell_1}(\X)
=
\max_{q\in[d],\,k\in[p_{-q}]}
\sum_{h\in[p_q]} |(\MM_q(\X))_{hk}|
\]
and
\[
\operatorname{Slice}_{\ell_1}(\X)
=
\max_{q\in[d],\,i\in[p_q]}
\sum_{h\in[p_{-q}]} |(\MM_q(\X))_{ih}|.
\]
Thus, $\operatorname{Fiber}_{\ell_1}(\X)$ is the maximum $\ell_1$ mass along any fiber, while $\operatorname{Slice}_{\ell_1}(\X)$ is the maximum $\ell_1$ mass of any slice. These quantities capture the local concentration of the variance profile and play a central role in both our upper and lower bounds.

\section{Estimators and Scaling Principles}\label{sec:estimators}

A central idea of this paper is that $\ell_1$ estimation can be reduced to estimating a suitably scaled low-rank tensor under Frobenius norm loss. This reduction separates the problem into two parts: choosing a scaling that mitigates heteroskedasticity, and estimating the resulting scaled tensor under Frobenius norm loss by spectral methods.

\subsection{The Scaled-Estimation Template}
\label{sec:density_and_scaled_tensor}

Let $\M=b_1\circ\cdots\circ b_d$ be a rank-one tensor with strictly positive entries, where $b_k\in\RR^{p_k}$ for $k\in[d]$, and define the scaled density tensor $\Q=\P * \M$. Since $\M$ is rank one, entrywise scaling preserves the CP structure of $\P$ in \eqref{eq:P-decomposition}. Indeed,
\[
\Q
=
\sum_{r=1}^R
w_r \,
\bigl(a_r^{(1)} * b_1\bigr)
\circ \cdots \circ
\bigl(a_r^{(d)} * b_d\bigr).
\]
Thus, $\rank(\Q)\le R$, so the scaled density tensor remains low rank.

Given an estimator $\hat \Q$ of $\Q$, define $\hat \P=\hat \Q * \M^{(-1)}$, where $\M^{(-1)}$ denotes the entrywise inverse of $\M$. Then $\hat \P-\P=(\hat \Q-\Q)*\M^{(-1)}$, and Hölder's inequality gives
\begin{equation}\label{eq_scaled_reduction_l1}
\|\hat \P-\P\|_1
=
\|(\hat \Q-\Q)*\M^{(-1)}\|_1
\le
\|\hat \Q-\Q\|_F
\|\M^{(-1)}\|_F .
\end{equation}
Therefore, once the inverse-scaling factor $\|\M^{(-1)}\|_F$ is controlled, $\ell_1$ estimation reduces to estimating the scaled density tensor $\Q$ in Frobenius norm.  

Thus, the scaling tensor has to balance two objectives: it should make the scaled estimation problem easier under Frobenius norm loss, while keeping $\|\M^{(-1)}\|_F=\prod_{k=1}^d \|b_k^{(-1)}\|_2$ bounded so that the inverse transformation does not amplify the error.


To see why scaling can improve the estimation problem, suppose $\Y\sim \mathrm{Multinomial}(n,\P)$. Then the scaled observation $\Y*\M$ has mean $\EE(\Y*\M)=n(\P*\M)=n\Q$. Moreover, as shown later in Theorem~\ref{thm_multinomial_tensor_generalized}, the Frobenius error of our spectral estimator is controlled by an effective complexity parameter of the scaled variance profile $\P*\M*\M$. Thus, the choice of $\M$ directly affects the statistical difficulty of estimating the scaled density tensor through the variance profile.

Motivated by \eqref{eq_scaled_reduction_l1}, one may view the ideal scaling as the solution to the oracle problem
\begin{equation}\label{eq_oracle_scaling_problem}
\min_{\M=b_1\circ\cdots\circ b_d}
\Bigl\{
\text{effective complexity of } \P * \M * \M
\Bigr\}
\quad
\text{subject to}
\quad
\prod_{k=1}^d \|b_k^{(-1)}\|_2 \lesssim 1.
\end{equation}
This optimization problem in \eqref{eq_oracle_scaling_problem} is oracle in nature because it depends on the unknown tensor $\P$. Nevertheless, it provides a useful guiding principle: a good scaling should reduce heteroskedasticity while preserving stability of the inverse transformation.


We next describe the base spectral estimator and the scaling choices used throughout the paper.

\subsection{The Base Spectral Estimation Algorithm}\label{sec_estimator}

We introduce a Tucker-based spectral estimator, Algorithm~\ref{algorithm_multinomial_tucker_wo_thinning}, for recovering the scaled density tensor $\Q=\P*\M$ from the scaled histogram tensor $\Y*\M$.
The estimator proceeds in three stages after scaling the observation. First, construct initial estimates of the mode-wise signal subspaces. This initialization step is based on Deflated-HeteroPCA \citep{zhou_deflated_2024}, which is designed to handle heteroskedastic noise and is introduced in Section \ref{sec_DeflatedHeteroPCA} of the Supplementary Materials. Second, refine each mode-wise subspace by applying a truncated SVD after projecting along the initial estimates from the other modes. Finally, we project the scaled observation onto the refined signal subspaces to obtain the reconstructed tensor. 

Algorithm~\ref{algorithm_multinomial_tucker_wo_thinning} makes use of the Tucker structure. Note that if $\rank(\P) = R$ and $\rank(\M) = 1$, then the Tucker rank of $\P*\M$ is at most  $(R,\ldots,R)$. 

\begin{algorithm}[ht]
    \caption{Scaled Density Tensor Estimation}
    \label{algorithm_multinomial_tucker_wo_thinning}
    \begin{algorithmic}[1]
        \algrenewcommand\algorithmicensure{\textbf{Run:}}
        \algrenewcommand\algorithmicrequire{\textbf{Input:}}
        \Require{Tensors $\Y$, $\M$, and target Tucker rank $(r_1,\ldots,r_d)$.}
        \State $\tilde{\Y} = \Y * \M$
        \For{$k=1,\ldots,d$}
        \State $G_{0,k}=P_{\mathrm{off\text{-}diag}}\!\bigl(\MM_k(\tilde{\Y})\MM_k(\tilde{\Y})^\top\bigr)$
        \State $\hat U_k^{(0)}=\operatorname{Deflated\text{-}HeteroPCA}(G_{0,k},r_k)$
        \EndFor

        \For{$k=1,\ldots,d$}
        \State $\hat U_k=\SVD_{r_k}\!\Bigl(\MM_k\bigl(\tilde{\Y}\times_{h\neq k}(\hat U_h^{(0)})^\top\bigr)\Bigr)$
        \EndFor

        \State $\hat \X=\tilde{\Y}\times_{k\in[d]}(\hat U_k\hat U_k^\top)$; $\hat \Q=\hat \X/n$\\
        \Return $\hat \Q$
    \end{algorithmic}
\end{algorithm}

\subsection{Scaling Choices and Summary of Estimators}\label{sec:scaling-choices}

The choice of the rank-one scaling tensor $\M=b_1\circ\cdots\circ b_d$ determines which loss the procedure is designed to optimize. For the Frobenius estimation, we use no scaling: $\M=\mathbf 1$, so Algorithm~\ref{algorithm_multinomial_tucker_wo_thinning} is applied directly to $\Y$, and the final estimator is $\hat\P=\hat\Q$.

For $\ell_1$ estimation, we consider two rank-one scaling rules. The first is an oracle factor-based scaling, which uses coordinate-wise information about the latent factors and serves as a benchmark. The second is a feasible slice-normalization rule, which uses only marginal slice masses and can be implemented from the data. Their precise definitions are given in Section~\ref{sec:l1-density-estimation}. The scaling choices are summarized below in Table \ref{tab:scaling-summary}, and their theoretical analysis is developed in later sections. 

\begin{table}[ht]
\centering
\begin{tabular}{llll}
\toprule
\midrule
Target loss & Scaling choice & Estimator & Analysis \\
\midrule
Frobenius & $\M=\mathbf 1$ & Direct spectral estimator & Section~\ref{sec:frobenius-estimation} \\
$\ell_1$ & Oracle factor-based scaling & Oracle scaled estimator & Section~\ref{sec:oracle-scaling} \\
$\ell_1$ & Slice normalization & Data-driven scaled estimator & Section~\ref{sec:slice-scaling}\\
\bottomrule
\end{tabular}
\caption{Summary of the scaling choices}
\label{tab:scaling-summary}
\end{table}

\section{General Theory for Scaled Tensor Estimation}\label{sec:general-theory}

The statistical idea and practical implementation are given by Algorithm \ref{algorithm_multinomial_tucker_wo_thinning}. For the theoretical analysis only, we use the following thinned version. The thinning is not part of the conceptual contribution; it is a proof device that separates the randomness used for initialization, refinement, and final projection.


Specifically, for each cell $i=(i_1,\ldots,i_d)$, we independently split the count $\Y_i$ into three parts according to a multinomial distribution with equal probabilities. This produces tensors $\Y_1,\Y_2,\Y_3$ satisfying $\Y=\Y_1+\Y_2+\Y_3$. Let $n_t=\|\Y_t\|_1$ for $t=1,2,3$. Conditional on the split sizes $(n_1,n_2,n_3)$, the tensors $\Y_1,\Y_2,\Y_3$ are independent and satisfy
\[
    \Y_t \mid n_t \sim \operatorname{Multinomial}(n_t,\P),
    \qquad t=1,2,3.
\]
This conditional independence allows us to decouple the initialization, refinement, and final projection steps, which considerably simplifies the proof.

\begin{algorithm}[ht]
    \caption{Scaled Density Tensor Estimation with Multinomial Thinning}
    \label{algorithm_multinomial_tucker}
    \begin{algorithmic}[1]
        \algrenewcommand\algorithmicensure{\textbf{Run:}}
        \algrenewcommand\algorithmicrequire{\textbf{Input:}}
        \Require{Tensors $\Y$, $\M$ and target Tucker rank $(r_1,\ldots,r_d)$.}
        
        \State For each cell $i=(i_1,\ldots,i_d)$, independently draw
        \[
            (\Y_{1,i},\Y_{2,i},\Y_{3,i}) \mid \Y_i
            \sim
            \operatorname{Multinomial}\!\left(\Y_i;1/3,1/3,1/3\right).
        \]
        \State Let $n_t=\|\Y_t\|_1$ for $t=1,2,3$.
        \State Let $\tilde{\Y}_t=\Y_t * \M$ for $t=1,2,3$.
        \For{$k=1,\ldots,d$}
        \State $G_{0,k}=P_{\mathrm{off\text{-}diag}}\!\bigl(\MM_k(\tilde{\Y}_1)\MM_k(\tilde{\Y}_1)^\top\bigr)$
        \State $\hat U_k^{(0)}=\operatorname{Deflated\text{-}HeteroPCA}(G_{0,k},r_k)$
        \EndFor

        \For{$k=1,\ldots,d$}
        \State $\hat U_k=\SVD_{r_k}\!\Bigl(\MM_k\bigl(\tilde{\Y}_2\times_{h\neq k}(\hat U_h^{(0)})^\top\bigr)\Bigr)$
        \EndFor

        \State $\hat \X=\tilde{\Y}_3\times_{k\in[d]}(\hat U_k\hat U_k^\top)$; $\hat \Q=\hat \X/n_3$\\
        \Return $\hat \Q$
    \end{algorithmic}
\end{algorithm}

Although Algorithm~\ref{algorithm_multinomial_tucker} is used for the theoretical analysis, the simulations in Section~\ref{sec:simulations} are based on the no-thinning implementation in Algorithm~\ref{algorithm_multinomial_tucker_wo_thinning}. Supplementary Section~\ref{app:sim_thinning} provides an empirical comparison of the two implementations and shows that the no-thinning version performs better in finite samples, likely because each step is carried out using the full sample.

We now state a general estimation error upper bound for the scaled density tensor estimated by Algorithm~\ref{algorithm_multinomial_tucker}. This result is the main technical ingredient behind our guarantees under both Frobenius and $\ell_1$ losses.
\begin{Theorem}\label{thm_multinomial_tensor_generalized}
    Assume that $\Y$ is sampled from the tensor multinomial distribution
    $\operatorname{Multinomial}(n,\P)$.
    Let $\M$ be a rank-one scaling tensor with positive entries. Let
    $\Q = \P * \M$. 
    Assume that $\Q$ has Tucker rank $(r_1, \ldots, r_d)$ with Tucker decomposition $\Q = \tilde \S \times_{k\in[d]} V_k,$ 
    and that the factor matrices $V_k$ satisfy
    \begin{equation}\label{eq_coherence_condition_0}
        \|V_k\|_{2,\infty}\lesssim 1/r_k.
    \end{equation}
    Define $r_{\max} = \max_{k\in[d]} r_k$ and $\rho_\star =
        \max_{h\in[d]}
        \big\|
        \bigotimes_{k\neq h} V_k
        \big\|_{2,\infty}.$
    Let $\hat \Q$ be the output of
    Algorithm \ref{algorithm_multinomial_tucker}. 
    
    For any $\tau \geq 2$, if
    \begin{equation}\label{eq_sigma_master_condition_0}
    \begin{aligned}
    \lambda_{\operatorname{Tucker}}(n\Q) \gtrsim\;&
    dr_{\max}\tau\log(p_{\max})
    \Bigl(
        \rho_\star \sqrt{n\operatorname{Slice}_{\ell_1}(\M*\M*\P)}
        +
        \|\M\|_\infty
    \Bigr)\\
    &\quad+
    (d+\sqrt{r_{\max}d})\tau\log(p_{\max})
    \sqrt{n\bigl(\operatorname{Fiber}_{\ell_1}(\M*\M*\P)\vee r_{\max} \|\M*\M*\P\|_\infty\bigr)} \\
    &\quad+
    \sqrt{r_{\max} n\,\tau\log(p_{\max})}\,
    \bigl(\operatorname{Fiber}_{\ell_1}(\M*\M*\P) \|\M\|_\infty^2\bigr)^{1/4}\\
    &\quad+
    \sqrt{r_{\max}} \,\|\M\|_\infty\, n^{1/4}
    \bigl(\tau\log(p_{\max})\bigr)^{3/4},
    \end{aligned}
    \end{equation}
    then
    \[
        \PP\left(
            n\|\hat \Q - \Q\|_F
            \leq
            \tau C_d \log(p_{\max})
            \Bigl( \sqrt{n \eta_{\max}} +\sqrt{r_{\max}}\|\M\|_\infty\Bigr)
        \right)
        \geq 1-c_d p_{\max}^{1-\tau}-6\exp(-n/30), 
    \]
    where
    \[
        \eta_{\max}
        =
        r_{\max}\operatorname{Fiber}_{\ell_1}(\M*\M*\P)
        \;\vee\;
        r_{\cdot} \|\M*\M*\P\|_\infty,
        \qquad
        r_{\cdot} = \prod_{q\in[d]} r_q,
    \]
    and $c_d<\infty$ is an absolute constant depending only on $d$.
\end{Theorem}

Theorem~\ref{thm_multinomial_tensor_generalized} serves two purposes. 
First, taking $\M$ to be the all-ones tensor yields a Frobenius estimation error upper bound for the original density tensor $\P$;
Second, by choosing $\M$ to reduce the fiber complexity of $\P*\M*\M$ while keeping the inverse scaling stable, the theorem leads to an upper bound for $\ell_1$ estimation; see Section~\ref{sec:l1-density-estimation}.

We next discuss the assumptions and consequences of Theorem~\ref{thm_multinomial_tensor_generalized}.

\medskip
\noindent
{\bf Structure of the signal.}
Both the algorithm and the analysis rely on Tucker low-rank structure. In the multiview multinomial model, the density tensor $\P$ is assumed to admit a CP decomposition of rank $R$, which implies that the corresponding mean tensor has Tucker rank at most $(R,\ldots,R)$. Since $\M$ is rank one, entrywise scaling by $\M$ preserves this low-rank structure.

\medskip
\noindent
{\bf Signal-to-noise ratio.}
Condition \eqref{eq_sigma_master_condition_0} is a signal-to-noise requirement. For fixed $\P$ and $\M$,
\[
    \lambda_{\operatorname{Tucker}}(n\Q)
    =
    \lambda_{\operatorname{Tucker}}(n\P * \M)
    \asymp n,
\]
whereas the dominant terms on the right-hand side of \eqref{eq_sigma_master_condition_0} scale no faster than $\sqrt n$ up to other problem-dependent factors. Therefore, the condition is satisfied when the sample size is sufficiently large.

Although \eqref{eq_sigma_master_condition_0} has a technical form, it admits a simpler interpretation in special cases. Suppose that $r_{\max}$ and $d$ are bounded constants and that $\M$ is the all-ones tensor. Since $\|\Q\|_1=\|\P\|_1=1$, condition \eqref{eq_sigma_master_condition_0} reduces, up to logarithmic factors and constants, to
\[
    \lambda_{\operatorname{Tucker}}(n\Q)
    \gtrsim
    \log(p_{\max})
    \Bigl(
        \rho_\star \sqrt{n\operatorname{Slice}_{\ell_1}(\P)}
        +
        \sqrt n\,\operatorname{Fiber}_{\ell_1}(\P)^{1/4}
        +
        n^{1/4}
    \Bigr).
\]
Equivalently, writing the condition in terms of the mean tensor $\X=n\P$, one obtains the simplified scale
\[
    \lambda_{\operatorname{Tucker}}(\X)
    \gtrsim
    \log(p_{\max})
    \Bigl(
        \rho_\star \sqrt{\operatorname{Slice}_{\ell_1}(\X)}
        +
        \bigl(\|\X\|_1\operatorname{Fiber}_{\ell_1}(\X)\bigr)^{1/4}
        +
        \|\X\|_1^{1/4}
    \Bigr).
\]
If, in addition, Assumption~(39b) of \cite{zhou_deflated_2024} holds, namely
$\rho_\star \leq p_{\max}^{(1-d)/4}$, and if the model is homogeneous with entries of $\X$ of constant order, then the condition further simplifies to
\begin{equation}\label{eq_snr_condition_simplified}
    \lambda_{\operatorname{Tucker}}(\X)
    \gtrsim
    \log(p_{\max})p^{(d+1)/4}.
\end{equation}
This condition is stronger than the corresponding SNR requirement
$\lambda_{\operatorname{Tucker}}(\X)\gtrsim p^{d/4}$ obtained under independent entrywise noise in \cite{zhang2018tensor, zhou_deflated_2024}. The gap reflects the additional technical difficulty caused by dependence among the entries of the multinomial tensor.

\medskip
\noindent
{\bf Error bound.}
The error bound in Theorem~\ref{thm_multinomial_tensor_generalized} has a natural variance-profile interpretation. The leading term is
\[
    \sqrt{n\eta_{\max}},
    \qquad
    \eta_{\max}
    =
    r_{\max}\operatorname{Fiber}_{\ell_1}(\M*\M*\P)
    \vee
    r_{\cdot}\|\M*\M*\P\|_\infty .
\]
Here, $\operatorname{Fiber}_{\ell_1}(\M*\M*\P)$ captures the largest aggregated noise level along a fiber of the scaled variance profile, while $\|\M*\M*\P\|_\infty$ controls the largest individual entry. Thus, $\eta_{\max}$ plays the role of an effective complexity parameter for the heteroskedastic scaled estimation problem.

The remaining term, $\sqrt{r_{\max}}\|\M\|_\infty$, 
does not grow with $n$ for fixed $\P$ and $\M$. Hence, in the large-sample regime, it is typically dominated by the leading term $\sqrt{n\eta_{\max}}$.

This interpretation is consistent with the independent Poisson heuristic. If one ignores the dependence among multinomial entries and instead regards $\Y$ as a Poisson tensor with $\EE\Y=\Var(\Y)=n\P$, then the scaled observation $\Y*\M$ has variance profile
\[
    \Var(\Y*\M) = n\P * \M * \M.
\]
The theorem shows that the Frobenius error is governed by fiber sums and entrywise maxima of this scaled variance profile, paralleling matrix results for heteroskedastic PCA \citep{zhang_heteroskedastic_2021, zhou_deflated_2024}.

\medskip
\noindent
{\bf Sketch of Proof for Theorem \ref{thm_multinomial_tensor_generalized}.}
The proof of Theorem \ref{thm_multinomial_tensor_generalized} must address two main obstacles: heteroskedasticity and dependence among entries of the multinomial tensor. Heteroskedasticity prevents a direct application of standard SVD or Tucker decomposition to the observation tensor. Algorithm~\ref{algorithm_multinomial_tucker} thus begins with Deflated-HeteroPCA \citep{zhou_deflated_2024}, which is designed for heteroskedastic noise. 

The second obstacle is dependence. The theoretical results in \cite{zhou_deflated_2024} assume independent noise entries, whereas the entries of a multinomial tensor are negatively dependent. Thus, these guarantees cannot be invoked directly. Instead, we start from the deterministic part of \cite[Theorem~4]{zhou_deflated_2024} and establish concentration bounds tailored to multinomial noise.

Let $Z=\MM_k(\Z_i)$ denote a centered multinomial noise matrix arising from the mode-$k$ unfolding. First-order quantities such as $\|Z\|$ can be controlled by Bernstein-type inequalities; see Lemma~\ref{lemma_multinomial_noise_control}. The main difficulty is the second-order term $P_{\mathrm{off\text{-}diag}}(ZZ^\top)$, 
which appears in the Deflated-HeteroPCA analysis. A crude estimate such as
$\|P_{\mathrm{off\text{-}diag}}(ZZ^\top)\|\leq 2\|ZZ^\top\|$
is too loose, because $Z$ is a very wide matrix and the diagonal of $ZZ^\top$ dominates its variance structure.

To obtain a sharper bound, we decompose the multinomial noise as $Z=\sum_{\ell=1}^n E_\ell,$ 
where the $E_\ell$ are i.i.d. centered $\operatorname{Multinomial}(1,P)$ noise matrices. Expanding
$P_{\mathrm{off\text{-}diag}}(ZZ^\top)$ under this representation yields a matrix-valued $U$-statistic. We then apply the decoupling inequality of \cite[Theorem~1]{delapena1995decoupling}, followed by the matrix Bernstein-type inequality in \cite[Remark~3.11]{minskerwei2019moment} to control the resulting decoupled sum. This gives the key off-diagonal concentration bound in Lemma~\ref{lemma_offdiag_EEt_multinomial_mw_corrected_v2}, which completes the probabilistic ingredient needed for Theorem~\ref{thm_multinomial_tensor_generalized}.

\section{Frobenius Estimation for Multiview Density Tensors}\label{sec:frobenius-estimation}

\subsection{Upper Bound}

We first specialize the general scaled estimation theorem to the original multiview multinomial model without additional scaling. 
Taking the scaling tensor $\M$ to be the all-ones tensor in Theorem~\ref{thm_multinomial_tensor_generalized} yields the following Frobenius-norm upper bound for estimating $\P$ itself. 

\begin{Corollary}\label{thm_multinomial_tensor_decomp}
    Assume $\Y$ is sampled from some rank-$R$ multiview multinomial distribution,  $\operatorname{Multinomial}(n,\P)$. 
    Assume that $\P$ admits the Tucker decomposition $\P = \tilde \S \times_{k\in[d]} V_k$ with Tucker rank $(r_1, \ldots, r_d)$, 
    and that the factor matrices $V_k$ satisfy
    \begin{equation}\label{eq_coherence_condition_1}
        \|V_k\|_{2,\infty}\lesssim 1/R.
    \end{equation}
    Define $\rho_\star =
        \max_{h\in[d]}
        \big\|
        \bigotimes_{k\neq h} V_k
        \big\|_{2,\infty}$. 
    For any $\tau \geq 2$, suppose that
    \be\label{eq_snr_condition_1}
        \lambda_{\operatorname{Tucker}}(\P) 
        \gtrsim 
        \tau C_{d, R} n^{-1/2}\log(p_{\max}) \Bigl(
            \rho_\star \sqrt{\operatorname{Slice}_{\ell_1}(\P)}
            +
            (\operatorname{Fiber}_{\ell_1}(\P))^{1/4}
            +
            n^{-1/4}.
        \Bigr)
    \ee
    Let $\hat \P$ be the output from Algorithm~\ref{algorithm_multinomial_tucker} with $\M$ as the all-ones tensor. Then
    \[
        \PP\left(
            \|\hat \P - \P\|_F
            \leq
            \tau C_d \log(p_{\max})
            \Bigl( \sqrt{\eta_{\max} / n} +R/n\Bigr)
        \right)
        \geq 1-c_d p_{\max}^{1-\tau}-6\exp(-n/30),
    \]
    where
    \[
        \eta_{\max}
        =
        R\operatorname{Fiber}_{\ell_1}(\P)
        \;\vee\;
        R^d \|\P\|_\infty,
    \]
    and $c_d<\infty$ is an absolute constant depending only on $d$. 
\end{Corollary}

\begin{Remark}[Tucker Rank]
    When the loading vectors of the density tensor are collinear across all modes, we will have $r_{\max}<R$. In this case, the dependence on $R$ in the final upper bounds of Corollary~\ref{thm_multinomial_tensor_decomp} is suboptimal and can be replaced by $r_{\max}$. The same refinement applies to the subsequent corollaries.
\end{Remark}

\begin{Remark}[Projection to the Probability Simplex]
    The estimator $\hat \P$ defined in the corollary is not necessarily a density tensor. This is not an issue if $\hat \P$ is viewed as an estimator of the mean tensor. If it is instead viewed as an estimator of the density tensor, one may further project it onto the probability simplex using Euclidean projection, which will not increase the Frobenius or $\ell_1$ error rate, as justified in Supplementary Section~\ref{sec_projection}. The same argument applies to the subsequent corollaries.
\end{Remark}

Corollary~\ref{thm_multinomial_tensor_decomp} shows that, for large $n$, the Frobenius norm error is governed by two quantities:
$R\operatorname{Fiber}_{\ell_1}(\P)$, which measures the effective noise level contributed by the heaviest fibers, and
$R^d\|\P\|_\infty$, which controls the contribution of exceptionally large entries. In many regimes of interest, especially when $R$ is fixed or moderate and $\P$ is not too spiky, the fiber term dominates. Thus the bound adapts to the local geometry of the distribution: it improves when the mass of $\P$ is spread relatively evenly and deteriorates when the mass is concentrated on a few fibers.



\subsection{Minimax Lower Bound}

We next show that the dependence on the fiber mass in Corollary~\ref{thm_multinomial_tensor_decomp} is statistically unavoidable. 
The lower bound is stated over a class of multiview tensors with controlled fiber mass, coherence, and rank.

\begin{Theorem}[Frobenius lower bound for order-$d$ multiview tensors]
\label{thm_l2_lower_bound}
Fix an integer $d\ge 3$. Let $p$ and $R$ be integers with $1\le R\le c_2p$. Let $F_{\operatorname{fiber}}$ satisfy ${C_0}/{p^{d-1}}
    \le
    F_{\operatorname{fiber}}
    \le
    {c_1}/{R},$ 
where $C_0>0$ is sufficiently large and $c_1,c_2>0$ are sufficiently small constants depending only on $d$. Define
\[
    m
    =
    \left\lfloor
    c_0
    \min\left\{
        R,
        \left(p^{d-1}F_{\operatorname{fiber}}\right)^{1/(d-2)}
    \right\}
    \right\rfloor
    \vee 1 ,
\]
where $c_0>0$ is a sufficiently small constant depending only on $d$. Then, for every $n
    \ge
    C_n
    m^{(d-2)/(d-1)}
    F_{\operatorname{fiber}}^{-1/(d-1)},$
there exists a finite collection $\mathcal G_n$ of nonnegative probability tensors in $\mathbb R^{p^d}$ such that every $\P\in\mathcal G_n$ satisfies $\operatorname{rank}_{\mathrm{CP}}(\P)\le R$ and $\operatorname{Fiber}_{\ell_1}(\P)\le F_{\operatorname{fiber}}.$
Moreover, for some orthonormal bases $U_1,\ldots,U_d$ of the mode-wise column spaces of $\P$, $ \max_{q\in[d]}\|U_q\|_{2,\infty}
    \le
    C
    \left(mF_{\operatorname{fiber}}\right)^{1/(2(d-1))}.$
Finally,
\[
    \inf_{\widehat \P}
    \sup_{\P\in\mathcal G_n}
    \mathbb E_{\P}\|\widehat \P-\P\|_F
    \ge
    c
    \sqrt{\frac{mF_{\operatorname{fiber}}}{n}},
\]
where $c,C,C_n>0$ depend only on $d$. In particular, if $F_{\operatorname{fiber}}
    \gtrsim
    {R^{d-2}} / {p^{d-1}},$
then $m\asymp R$ and, for $n
    \gtrsim
    R^{(d-2)/(d-1)}
    F_{\operatorname{fiber}}^{-1/(d-1)},$
we have
\[
    \inf_{\widehat \P}
    \sup_{\P\in\mathcal G_n}
    \mathbb E_{\P}\|\widehat \P-\P\|_F
    \gtrsim
    \sqrt{\frac{RF_{\operatorname{fiber}}}{n}}.
\]
\end{Theorem}
\begin{Remark}[Exact rank]\label{rem:exact-rank}
In the regime $F_{\operatorname{fiber}}
    \gtrsim
    {R^{d-2}} / {p^{d-1}},$ the above construction
can be modified to have exact CP rank $R$. 
\end{Remark}

\begin{Remark}[Coherence condition]
In Theorem~\ref{thm_l2_lower_bound}, we construct a family of tensors satisfying the coherence condition, since such a condition is required for the upper bound in Theorem~\ref{thm_multinomial_tensor_decomp}. 
Similar coherence assumptions also appear in the heteroskedastic PCA literature; see, for example, \cite{zhang_heteroskedastic_2021,zhou_deflated_2024}. 
However, in contrast to the rank and fiber-mass constraints, the coherence restriction does not affect the minimax error rate. This is consistent with the matrix setting with independent noise studied in \cite{zhang_heteroskedastic_2021}.
\end{Remark}

Theorem~\ref{thm_l2_lower_bound} shows that the lower bound depends explicitly on the fiber mass parameter $F_{\operatorname{fiber}}$, confirming that this quantity is not merely an artifact of the upper-bound analysis but a genuine statistical feature of the problem. 

\subsection{Optimality and the Role of Signal Strength}\label{sec_optimality_l2}

We now compare the upper and lower bounds and clarify the role of the signal-to-noise condition in Corollary~\ref{thm_multinomial_tensor_decomp}.

Consider first the fixed-rank case.
For any probability tensor $\P \in \mathbb{R}^{p^d}$, the quantity $\operatorname{Fiber}_{\ell_1}(\P)$ is at least of order $1/p^{d-1}$, so the parameter $F_{\operatorname{fiber}}$ in Theorem~\ref{thm_l2_lower_bound} ranges over the full feasible scale up to constants. Moreover, when $R$ is fixed, the class $\mathcal{G}_n$ in Theorem~\ref{thm_l2_lower_bound} can satisfy the coherence condition \eqref{eq_coherence_condition_1} in Corollary~\ref{thm_multinomial_tensor_decomp}. Hence, for sufficiently large $n$, the signal-to-noise condition \eqref{eq_snr_condition_1} holds, and Corollary~\ref{thm_multinomial_tensor_decomp} together with Theorem~\ref{thm_l2_lower_bound} yields the minimax-optimal Frobenius rate $\sqrt{\operatorname{Fiber}_{\ell_1}(\P)/n}$, up to logarithmic factors and constants depending only on the fixed rank. 

This conclusion should be interpreted with care. Theorem~\ref{thm_l2_lower_bound} is a purely statistical minimax lower bound, characterizing the best possible estimation accuracy without computational constraints. By contrast, Corollary~\ref{thm_multinomial_tensor_decomp} gives an upper bound for a specific polynomial-time estimator, and this upper bound requires the additional signal-to-noise condition \eqref{eq_snr_condition_1}. Therefore, the matching upper and lower bounds do \emph{not} imply that our estimator is consistent as soon as $n \gtrsim \operatorname{Fiber}_{\ell_1}(\P)$, even when $R$ and $d$ are fixed. Rather, they show that once the signal is strong enough for the spectral procedure to succeed, the resulting Frobenius rate is minimax optimal.

Identifying the optimal computational threshold for signal tensor estimation in the multiview multinomial model, and more generally under heteroskedastic and dependent noise, remains an important open problem. 

\section{$\ell_1$ Estimation via Tensor Scaling}\label{sec:l1-density-estimation}

\subsection{Why $\ell_1$ Estimation Needs Scaling: Matrix Intuition}

We now use the general reduction of Section~\ref{sec:density_and_scaled_tensor} to construct $\ell_1$ estimators. We first revisit the matrix case, where row-column rescaling gives a useful guide for the tensor setting.

When $d=2$, the minimax-optimal $\ell_1$ error rate for estimating a multiview density matrix is $\sqrt{Rp_{\max}/n}$, where $p_{\max}=\max\{p_1,p_2\}$ \citep{chhor2024generalized}. A direct low-rank approximation of the empirical histogram is generally suboptimal because the multinomial noise is heteroskedastic: rows and columns with larger marginal probabilities have larger variances. A suitable row-column rescaling mitigates this heteroskedasticity and restores the optimal rate \citep{Jain2020linear}.

To see this, consider $Y \sim \mathrm{Multinomial}(n,P)$ with $P\in\RR^{p_1\times p_2}$, and let $\bar Y=Y/n$ be the empirical density matrix. Define a rank-one scaling matrix $M=ab^\top$, where $a_i=\bigl(\max\{P_{i*},1/p_1\}\bigr)^{-1/2}$ for $i\in[p_1]$ and $b_j=\bigl(\max\{P_{*j},1/p_2\}\bigr)^{-1/2}$ for $j\in[p_2]$. Let $Q=P*M$, let $\hat Q$ be a rank-$R$ approximation of $\bar Y*M$, and set $\hat P=\hat Q*M^{(-1)}$. Then
\begin{equation}\label{eq_matrix_holder}
\|P-\hat P\|_1
=
\|(Q-\hat Q)*M^{(-1)}\|_1
\leq
\|Q-\hat Q\|_F\,\|M^{(-1)}\|_F .
\end{equation}
Thus, as in the tensor case, the $\ell_1$ error decomposes into a Frobenius estimation error term for the scaled object and a stability term for the inverse scaling.

For the scaling above, the matrix analogue of Theorem \ref{thm_multinomial_tensor_generalized} (see, e.g., \cite{Jain2020linear}) yields $\|Q-\hat Q\|_F \lesssim \sqrt{Rp_{\max}/n}$, while the inverse scaling remains stable  $\|M^{(-1)}\|_F\lesssim 1$. Combining these bounds with \eqref{eq_matrix_holder} gives $\|P-\hat P\|_1\lesssim \sqrt{Rp_{\max}/n}$, which matches the minimax-optimal rate. The key point is that rank-one scaling rebalances the heteroskedastic multinomial noise without amplifying the error through the inverse transformation.

This row-column rescaling also has a useful variance-profile interpretation. In the matrix case, the Frobenius estimation error is governed by row and column sums of the variance profile \cite{zhang_heteroskedastic_2021}. The scaling above approximately balances these row and column sums. Since rows and columns are precisely the fibers of a matrix, this suggests a natural higher-order analogue: rescale the tensor to balance its fiber-wise variance profile. However, an arbitrary fiber-specific scaling tensor may increase the Tucker or CP rank of the signal, making the scaled problem harder rather than easier. We therefore restrict attention to rank-one scaling tensors, which preserve the low-rank structure while still allowing substantial control over heteroskedasticity.

\subsection{Oracle Scaling and Upper Bound}\label{sec:oracle-scaling}

We now return to the tensor setting and instantiate the general scaling principle in \eqref{eq_oracle_scaling_problem}. We first study an oracle scaling rule, which serves as a benchmark for what can be achieved when one has approximate coordinatewise information about the latent factors.

Recall the multiview representation
\[
\P=\sum_{r=1}^R w_r\, a_r^{(1)}\circ\cdots\circ a_r^{(d)}.
\]
For each mode $k\in[d]$ and coordinate $i\in[p_k]$, define
\[
\sigma_{k,i}=\gamma_{k,i}\max_{r\in[R]} (a_r^{(k)})_i,
\]
where the constants $\gamma_{k,i}$ quantify the quality of the oracle information. Assume throughout that $c\le \gamma_{k,i}\le C$ for some absolute constants $c,C>0$. We then define the scaling vectors
\[
(b_k)_i=
\sqrt{\frac{\sum_{h=1}^{p_k}\sigma_{k,h}}{\sigma_{k,i}}},
\qquad k\in[d],\ i\in[p_k],
\]
and let $\M=b_1\circ\cdots\circ b_d$. 
To ensure that $(b_k)_i$ is well-defined, we assume without loss of generality that $\sigma_{k,i}>0$. Indeed, if $\sigma_{k,i}=0$, then $\max_{r\in[R]} (a_r^{(k)})_i=0$, so the corresponding slice of the probability tensor is entirely zero. Such a slice can therefore be removed, and the analysis can be carried out on the remaining tensor.

This choice upweights coordinates for which all components are small, and downweights coordinates where at least one component is large. Intuitively, it compensates for the uneven coordinatewise mass of the factor vectors and thereby reduces the heteroskedasticity of the scaled estimation problem. At the same time, the normalization by $\sum_h \sigma_{k,h}$ ensures that the inverse scaling remains stable.

The following proposition, proved in Appendix~\ref{sec_appendix_oracle_scaling}, summarizes the two key properties of this construction.

\begin{Proposition}\label{prop_oracle_scaling_properties}
Let $\M$ be defined as above with $\sigma_{k,i} > 0$. Then:
\begin{enumerate}
    \item the inverse scaling is bounded,
    \[
    \|\M^{(-1)}\|_F \asymp 1;
    \]
    \item the scaled tensor has controlled fiber mass,
    \[
    \operatorname{Fiber}_{\ell_1}(\P * \M * \M)\lesssim R^d p_{\max}.
    \]
\end{enumerate}
\end{Proposition}

Combining Proposition~\ref{prop_oracle_scaling_properties} with the general scaled estimation bound in Theorem \ref{thm_multinomial_tensor_generalized} and \eqref{eq_scaled_reduction_l1} yields the following oracle $\ell_1$ upper bound.
\begin{Corollary}\label{thm_multiview_tensor_decomp_0}
Assume $\Y$ is sampled from some rank-$R$ multiview multinomial distribution
$\operatorname{Multinomial}(n,\P)$. 
Let $\Q=\P * \M$, where $\M$ is the oracle scaling tensor defined above with $\sigma_{k,i} > 0$.
Assume that $\Q$ admits a Tucker decomposition $\Q=\tilde \S\times_{k\in[d]}V_k$ with Tucker rank $(r_1, \ldots, r_d)$, and that the factor matrices satisfy $\|V_k\|_{2,\infty}\lesssim 1/R$ for all $k\in[d]$.
Define $\rho_\star=\max_{h\in[d]}\big\|\bigotimes_{k\neq h}V_k\big\|_{2,\infty}$.
Let $\hat \Q$ be the output of Algorithm~\ref{algorithm_multinomial_tucker}, and define $\hat \P=\hat \Q * \M^{(-1)}$. 
Then there exist constants $C_{d,R},C_d,c_d>0$, depending only on $d$ and $R$, such that for any $\tau\ge 2$, if
\begin{equation}\label{eq_snr_condition_oracle_scaling}
\begin{aligned}
\lambda_{\operatorname{Tucker}}(\Q)
\gtrsim{}&
\tau C_{d,R} n^{-1/2}\log(p_{\max})
\Biggl[
\rho_\star \sqrt{\operatorname{Slice}_{\ell_1}(\P * \M * \M)}
+
\sqrt{p_{\max}}
\\
&\qquad\qquad\qquad\qquad\qquad
+
(p_{\max}\|\M\|_\infty^2)^{1/4}
+
n^{-1/4}\|\M\|_\infty
\Biggr].
\end{aligned}
\end{equation}
Then
\[
\PP\left(
\|\hat \P-\P\|_1
\le
\tau C_d \log(p_{\max})
\left(
R^d\sqrt{\frac{p_{\max}}{n}}
+
\frac{ \sqrt R\,\|\M\|_\infty}{n}
\right)
\right)
\ge
1-c_d p_{\max}^{1-\tau}-6\exp(-n/30).
\]
\end{Corollary}

Finally, note that the oracle quantities $\max_r (a_r^{(k)})_i$ are typically unknown. However, Corollary~\ref{thm_multiview_tensor_decomp_0} only requires them up to constant factors. Thus, in principle, a preliminary estimator of these coordinatewise maxima could be used to approximate the oracle scaling; for example, one may use a pilot CP decomposition procedure \citep{tang2025revisit}. A full analysis of such an approximate oracle implementation is beyond the scope of this paper.

\subsection{Slice Normalization: A Feasible Scaling Rule}\label{sec:slice-scaling}

Although the oracle scaling above provides a useful benchmark, it depends on approximate knowledge of the coordinatewise maxima $\max_{r\in[R]} (a_r^{(k)})_i$, which are generally unavailable in practice. We now introduce a feasible alternative motivated by the same variance-stabilization principle, but based only on low-dimensional slice marginals of $\P$.

For a tensor $\X$, define the mode-$q$ slice and fiber $\ell_1$ norms by
\[
\operatorname{Slice}_{\ell_1}^{(q,i)}(\X)
= \sum_{h\in[p_{-q}]} |(\MM_q(\X))_{ih}|,
\qquad
\operatorname{Fiber}_{\ell_1}^{(q,k)}(\X)
= \sum_{h\in[p_q]} |(\MM_q(\X))_{hk}|.
\]
The quantity $\operatorname{Slice}_{\ell_1}^{(q,i)}(\P)$ is simply the marginal mass of the $i$th slice in mode $q$. Unlike the oracle quantities above, these slice marginals are only $p_q$-dimensional and can be estimated accurately from data.

Motivated by the matrix scaling rule, we define
\begin{equation}\label{eq_bk_slice}
(b_k)_i
=
\Bigl(\max\{\operatorname{Slice}_{\ell_1}^{(k,i)}(\P),\,1/p_k\}\Bigr)^{-1/2},
\qquad k\in[d],\ i\in[p_k],
\end{equation}
and let $\M^\prime=b_1\circ\cdots\circ b_d$. We refer to this choice as \emph{slice normalization}. It rescales each coordinate according to the marginal mass of the corresponding slice: large slices are downweighted, while small slices are protected by the truncation at $1/p_k$. Thus, slice normalization stabilizes the heteroskedasticity of the tensor using only easily estimable marginal information.

The next proposition, proved in Appendix~\ref{sec_appendix_slice_scaling}, provides an upper bound of $\|(\M^\prime)^{(-1)}\|_F$. 

\begin{Proposition}\label{prop_slice_scaling_properties}
Let $\M^\prime$ be defined by \eqref{eq_bk_slice}. Then the inverse scaling is bounded, 
\[
    \|(\M^\prime)^{(-1)}\|_F \le C_d.   
\]
\end{Proposition}

Applying the general bound in Theorem \ref{thm_multinomial_tensor_generalized} with Proposition \ref{prop_slice_scaling_properties} and \eqref{eq_scaled_reduction_l1} immediately gives the following $\ell_1$-norm upper bound.

\begin{Corollary}\label{thm_slice_scaling}
Assume $\Y$ is sampled from some rank-$R$ multiview multinomial distribution
$\operatorname{Multinomial}(n,\P)$. 
Let $\Q=\P * \M^\prime$, where $\M^\prime$ is the slice-normalization tensor defined in \eqref{eq_bk_slice}. 
Assume that $\Q$ admits a Tucker decomposition $\Q=\tilde \S\times_{k\in[d]}V_k$ with Tucker rank $(r_1, \ldots, r_d)$, and that the factor matrices satisfy $\|V_k\|_{2,\infty}\lesssim 1/R$ for all $k\in[d]$.
Define $\rho_\star=\max_{h\in[d]}\left\|\bigotimes_{k\neq h}V_k\right\|_{2,\infty}$.
Let $\hat \Q$ be the output of Algorithm~\ref{algorithm_multinomial_tucker} with $\M = \M^\prime$, and define $\hat \P=\hat \Q * \M^{\prime(-1)}.$
There exist constants $C_{d,R},c_d>0$, depending only on $d$ and $R$, such that for any $\tau\ge 2$, if
\begin{equation}\label{eq_snr_condition_slice_scaling}
\begin{aligned}
\lambda_{\operatorname{Tucker}}(\Q)
\gtrsim{}&
\tau C_{d,R}\log(p_{\max})
\Biggl[
\rho_\star
\sqrt{
\frac{\operatorname{Slice}_{\ell_1}(\P * \M^\prime * \M^\prime)}{n}
}
+
\sqrt{\frac{\eta_{\max}}{n}}
\\
&\qquad\qquad\qquad\qquad
+
\frac{
\bigl(
\operatorname{Fiber}_{\ell_1}(\P * \M^\prime * \M^\prime)\,
\|\M^\prime\|_\infty^2
\bigr)^{1/4}
}{\sqrt n}
+
\frac{\|\M^\prime\|_\infty}{n^{3/4}}
\Biggr],
\end{aligned}
\end{equation}
then
\[
\PP\left(
\|\hat \P-\P\|_1
\le
\tau C_d \log(p_{\max})
\left(
\sqrt{\eta_{\max}/n}
+
\frac{\sqrt R\,\|\M^\prime\|_\infty}{n}
\right)
\right)
\ge
1-c_d p_{\max}^{1-\tau}-6\exp(-n/30),
\]
where
\[
\eta_{\max}
=
\bigl(
R\operatorname{Fiber}_{\ell_1}(\P*\M^\prime*\M^\prime)
\bigr)
\vee
\bigl(
R^d\|\P*\M^\prime*\M^\prime\|_\infty
\bigr).
\]
\end{Corollary}

\begin{Remark}
    A theoretical guarantee for the fully empirical slice-normalization estimator is provided in Supplementary Section~\ref{sec_appendix_plugin_slice_scaling}.
\end{Remark}

Corollary~\ref{thm_slice_scaling} is stated for the population slice-normalization rule. Since the required slice marginals are low-dimensional, this rule naturally admits a plug-in implementation based on empirical marginals; this data-driven version is studied in the simulations.

The remaining complexity of the upper bound in Corollary~\ref{thm_slice_scaling} is captured by the normalized tensor $\tilde\Q=\P*\M^\prime*\M^\prime$. In many regimes, the dominant term in $\eta_{\max}$ is the fiber term $R\operatorname{Fiber}_{\ell_1}(\tilde\Q)$, so understanding the quality of slice normalization reduces to understanding how well slice masses control fiber masses.

To make this explicit, define
\begin{equation}\label{eq_xi_new}
\xi
=
\min\Bigl\{
a:
\operatorname{Fiber}_{\ell_1}^{(q,k)}(\P)
\le
a
\!\!\prod_{(t,h)\in\mathcal H_{(q,k)}}\!\!
\operatorname{Slice}_{\ell_1}^{(t,h)}(\P)
\Bigr\},
\end{equation}
where $\mathcal H_{(q,k)}$ denotes the collection of the $d-1$ slices containing the fiber $(q,k)$. The parameter $\xi$ measures how well the mass of each fiber can be controlled by the product of the masses of the surrounding slices. When $\xi$ is small, slice normalization is particularly effective.

The following simplification, proved in Appendix~\ref{sec_appendix_slice_scaling}, shows that in this case slice normalization achieves a near-optimal $\ell_1$ rate.

\begin{Corollary}\label{thm_slice_scaling_simple}
Assume the same conditions as in Corollary~\ref{thm_slice_scaling}, and let $\tilde \Q=\P*\M^\prime*\M^\prime.$ 
If
\[
R\operatorname{Fiber}_{\ell_1}(\tilde \Q)
\gtrsim
R^d\|\tilde \Q\|_\infty,
\qquad
n\,\operatorname{Fiber}_{\ell_1}(\tilde \Q)
\gtrsim
\|\M^\prime\|_\infty^2,
\]
then
\[
\PP\left(
\|\hat \P-\P\|_1
\le
\tau C_d \log(p_{\max})
\sqrt{\frac{\xi R p_{\max}}{n}}
\right)
\ge
1-c_d p_{\max}^{1-\tau}-6\exp(-n/30).
\]
\end{Corollary}

The parameter $\xi$ is benign in many important cases. For example, when $d=2$ or when $\P$ is approximately homogeneous, one has $\xi=1$. More generally, for multiview tensors with balanced mixture weights, $\xi$ remains controlled as $\xi \leq w_{\min}^{2-d}$ (see proof in Section \ref{sec_appendix_slice_scaling}), so slice normalization is near-optimal while requiring only low-dimensional marginal information. This makes it a practical alternative to the oracle scaling rule.

\subsection{Minimax Lower Bound}

We now turn to the fundamental statistical difficulty of $\ell_1$ estimation for the multiview model. The following theorem is based on the same packing construction idea as Theorem~\ref{thm_l2_lower_bound}. It
shows that, even under controlled rank, fiber mass, and coherence, the
$\ell_1$ loss cannot be smaller than the scale determined by the effective
number of perturbed coordinates in the construction.

\begin{Theorem}[$\ell_1$ lower bound]
\label{thm_lower_bound}
Fix an integer $d\ge 3$. Let $p$ and $R$ be integers satisfying $1\le R\le c_d p$, where $c_d>0$ is a sufficiently small constant depending only on $d$. Then, for every $n\ge C_d Rp,$
there exists a finite collection $\mathcal G_n$ of nonnegative probability tensors in $\RR^{p^d}$ such that every $\P\in\mathcal G_n$ satisfies $\operatorname{rank}_{\mathrm{CP}}(\P)=R$ and $
    \operatorname{Fiber}_{\ell_1}(\P)
    \le
    C_d{R^{d-2}}/{p^{d-1}},$
and, for some orthonormal bases $U_1,\ldots,U_d$ of the mode-wise column spaces of $\P$, $\max_{q\in[d]}\|U_q\|_{2,\infty}^2
    \le
    C_d{R}/{p}.$
Moreover,
\[
    \inf_{\widehat \P}
    \sup_{\P\in\mathcal G_n}
    \mathbb E_{\P}\|\widehat \P-\P\|_1
    \ge
    c_d\sqrt{\frac{Rp}{n}},
\]
where $c_d,C_d>0$ depend only on $d$.
\end{Theorem}

\begin{Remark}
In the lower-bound construction, the quantities $\max_{r\in[R]}(a_r^{(k)})_i$ required in Corollary \ref{thm_multiview_tensor_decomp_0} are fixed, or known up to universal constants, over the packing set. 
Thus, the minimax lower bound is not caused by the difficulty of estimating the scaling.
\end{Remark}

In the next subsection, we compare the oracle and slice-normalization
estimators against this lower bound.

\subsection{Near-optimality and Comparison of the Two Scalings}

We now compare the two scaling strategies above---the oracle scaling in
Corollary~\ref{thm_multiview_tensor_decomp_0} and the slice-normalization scaling in
Corollaries~\ref{thm_slice_scaling} and \ref{thm_slice_scaling_simple}---against the lower bound in
Theorem~\ref{thm_lower_bound}. As in Section~\ref{sec_optimality_l2}, this comparison should be interpreted under the additional signal-to-noise conditions required by the corresponding upper bounds.

\paragraph{Oracle scaling.}
The oracle construction is mainly useful as a benchmark. Under the signal-to-noise condition
\eqref{eq_snr_condition_oracle_scaling}, Corollary~\ref{thm_multiview_tensor_decomp_0} gives
\[
\|\hat \P - \P\|_1
\lesssim
\tau \log(p_{\max})
\left(
R^d\sqrt{\frac{p_{\max}}{n}}
+
\frac{ \sqrt R\,\|\M\|_\infty}{n}
\right).
\]
When $n$ is sufficiently large so that the second term is negligible, the dominant contribution is $R^d\sqrt{{p_{\max}} / {n}}.$ 
Thus, for fixed $d$ and fixed $R$, the oracle estimator has the correct dependence on $p_{\max}$ and $n$ up to logarithmic factors. 
However, its dependence on the rank is much stronger than that in the lower bound of Theorem~\ref{thm_lower_bound}, which scales as $\sqrt{pR/n}$ under the incoherent construction.
In this sense, the oracle construction should be viewed primarily as an analytically convenient benchmark showing that suitable rank-one rescaling can substantially reduce heteroskedasticity.

\paragraph{Slice normalization.}
By contrast, slice normalization depends only on slice marginals of $\P$, which are low-dimensional and can be estimated accurately from data. More importantly, under the additional structural condition encoded by $\xi$, Corollary~\ref{thm_slice_scaling_simple} yields
\[
\|\hat \P-\P\|_1
\lesssim
\tau \log(p_{\max})
\sqrt{\frac{\xi R p_{\max}}{n}}.
\]
Hence, this matches the lower bound in Theorem~\ref{thm_lower_bound} up to logarithmic factors and the structural parameter $\xi$. In particular, whenever $\xi$ is bounded by an absolute constant, slice normalization achieves the minimax-optimal dependence on $p$, $R$, and $n$ up to logarithmic factors.

The two constructions serve different purposes. Oracle scaling provides a conceptual benchmark for how much one can gain from ideal coordinatewise information, but its resulting upper bound is generally suboptimal in $R$. 
Slice normalization, on the other hand, is data-driven up to estimation of low-dimensional marginals and, under the mild additional condition measured by $\xi$, achieves the statistically optimal scaling in $n$, $p$, and $R$ up to logarithmic factors. This makes slice normalization the more practically relevant procedure, while the oracle rule remains useful for understanding the geometry of variance stabilization in the multiview multinomial model.

\section{Simulation Studies}\label{sec:simulations}

We evaluate the finite-sample performance of the proposed estimators under the multiview multinomial model
\[
\P = \sum_{r=1}^R w_r \, a_r^{(1)} \circ a_r^{(2)} \circ a_r^{(3)},
\]
where each factor vector $a_r^{(k)}$ is nonnegative and sums to one, and the weights satisfy $w_r \ge 0$ and $\sum_{r=1}^R w_r = 1$. 
Throughout, we generate the histogram tensor $\Y \sim \mathrm{Multinomial}(n,\P),$ where $n$ is the total sample size. 
Unless otherwise stated, we take $d=3$ and $p_1=p_2=p_3=p$. 
All simulations are implemented in R and run in parallel on a Linux platform. 

All estimators considered in this simulation section are implemented using the no-thinning procedure in Algorithm~\ref{algorithm_multinomial_tucker_wo_thinning}. A comparison between the no-thinning implementation in Algorithm~\ref{algorithm_multinomial_tucker_wo_thinning} and the multinomial-thinning implementation in Algorithm~\ref{algorithm_multinomial_tucker} is provided in Supplementary Section~\ref{app:sim_thinning}. 
Codes are provided in GitHub\footnote{\url{https://github.com/RunshiTang/Optimal-Estimation-of-Discrete-Multiview-Distributions-under-Heteroskedastic-Multinomial-Sampling/}}. 

\subsection{Generation of Heteroskedastic Factors}

For each mode $k$ with dimension $p_k$, we generate $R$ probability vectors
$\{a^{(k)}_r\}_{r=1}^R$ from heteroskedastic Dirichlet distributions.

We first construct a global heteroskedastic profile. Let $H \ge 1$ denote the
\emph{heteroskedasticity strength} parameter and define
\[
b_j = \exp\!\left(\frac{j-1}{p_{\max}-1}\log H\right),
\quad j=1,\ldots,p_{\max},
\]
where $p_{\max} = \max_k p_k$. 
We normalize this profile so that its mean equals one:
\[
\tilde b_j =
\frac{b_j}{\frac{1}{p_{\max}}\sum_{i=1}^{p_{\max}} b_i},
\qquad j=1,\ldots,p_{\max}.
\]

For each mode $k$ and component $r$, we construct the Dirichlet parameter vector as follows.
Let $\tilde b^{(k)} = (\tilde b_1,\ldots,\tilde b_{p_k})$ be the truncated profile and randomly permute its entries. 
The Dirichlet concentration vector is then defined as
\[
{\alpha}^{(k)} = \alpha \tilde b^{(k)}, \qquad \alpha = 0.8.
\]
Finally, we draw the factor vector
\[
a^{(k)}_r \sim \mathrm{Dirichlet}\!\left({\alpha}^{(k)}\right).
\]

The resulting vectors are dense probability vectors whose coordinates exhibit controlled heteroskedasticity. The parameter $H$ determines the level of heteroskedasticity across coordinates, with $H=1$ corresponding to homogeneous Dirichlet parameters.

\subsection{Frobenius Error}

In this section, we study the Frobenius estimation of $\P$. We compare our proposed estimation procedure with the baseline histogram estimator and examine whether the empirical behavior matches the minimax error rate stated in Corollary~\ref{thm_multinomial_tensor_decomp} and Theorem~\ref{thm_l2_lower_bound}. We consider the following two methods.

\begin{enumerate}
    \item \texttt{Histogram}: the empirical histogram $\Y/n$.
    \item \texttt{Unscaled}: Algorithm~\ref{algorithm_multinomial_tucker_wo_thinning} with scaling tensor $\M$ as the all-ones tensor.
\end{enumerate}

\paragraph{Varying sample size $n$.} 
In the first experiment, we investigate the error as the sample size $n$ varies. 
We fix $R=4,~H = 100,~w_r = 1/R$, and $p_1=p_2=p_3=50$. 
For each setting, we repeat the experiment independently over 50 Monte Carlo replications and report the average and standard error of the estimation errors. 
The result is visualized in Figure~\ref{fig_exp2_F}. 
In the left panel, we plot the Frobenius error versus different sample sizes. 
It can be seen that our approach \texttt{Unscaled} demonstrates clear superiority over the \texttt{Histogram} estimator. 
In the right panel, we normalize the Frobenius error by the minimax error rate, i.e., for each simulation experiment, we let 
\bea\label{eq_F_sacle}
    \operatorname{Normalized\, Error} = \|\widehat{\P}-\P\|_F \cdot \sqrt{\frac{n}{R \operatorname{Fiber}_{\ell_1}(\P)}}. 
\eea
It can be seen that the normalized errors are approximately flat in the right panel, which is consistent with the theoretical minimax error rate.   

\begin{figure}[htbp]
    \centering
    \includegraphics[width=0.8\linewidth]{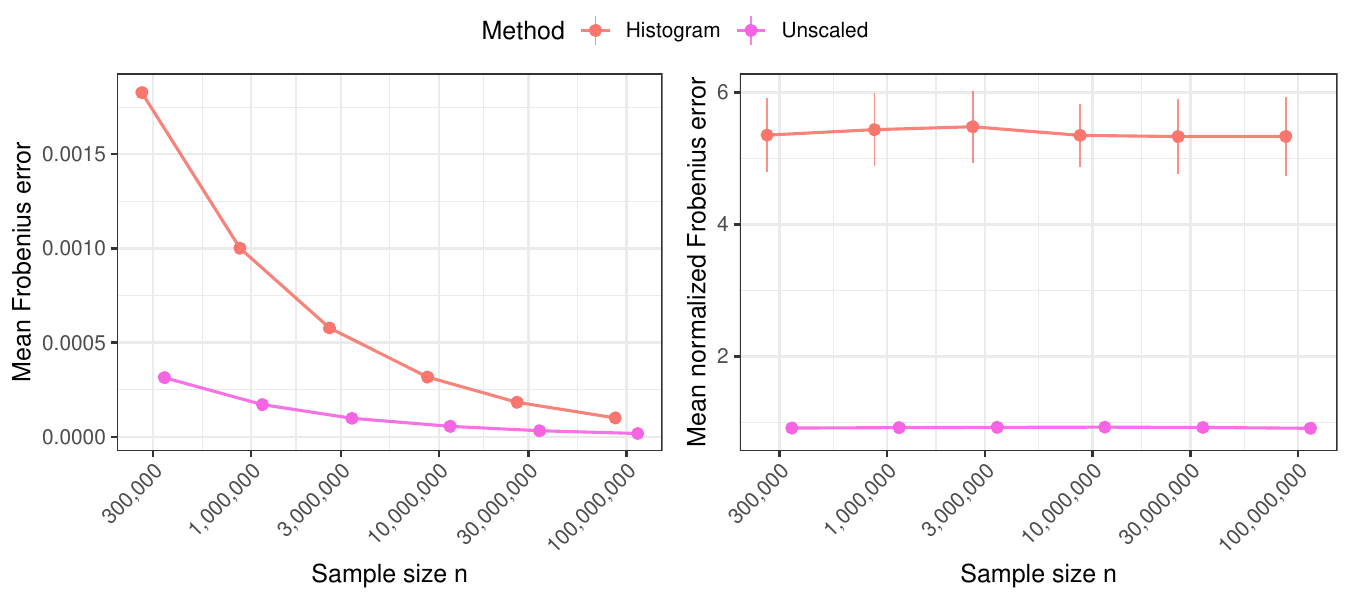}
    \caption{Comparison of Frobenius error across different sample sizes $n$}
    \label{fig_exp2_F}
\end{figure}

\paragraph{Varying heteroskedasticity strength $H$.}

In this experiment, we investigate the error as the heteroskedasticity strength $H$ varies. 
We fix $R=4,~n = 200000,~w_r = 1/R$, and $p_1=p_2=p_3=50$. 
For each setting, we repeat the experiment independently over 50 Monte Carlo replications and report the average and standard error of the Frobenius errors. 
The result is visualized in Figure~\ref{fig_exp3_F}. 
Again, we plot the Frobenius error and normalized Frobenius error by \eqref{eq_F_sacle} in the left and right panels, respectively. 
It can be seen that our approach \texttt{Unscaled} again demonstrates clear superiority over the \texttt{Histogram} estimator. 
Additionally, the normalized errors are approximately flat for our proposed method \texttt{Unscaled} in the right panel, which indicates that the dependence on $\operatorname{Fiber}_{\ell_1}(\P)$ in the theoretical minimax error rate is correct. 
Notably, the normalized \texttt{Histogram} estimation error decreases as $H$ increases. This may occur because when the probability tensor $\P$ becomes more heteroskedastic, the mass concentrates on fewer entries, which reduces the effective dimension of $\P$ and leads to improved performance of the \texttt{Histogram} estimator. 

\begin{figure}[htbp]
    \centering
    \includegraphics[width=0.8\linewidth]{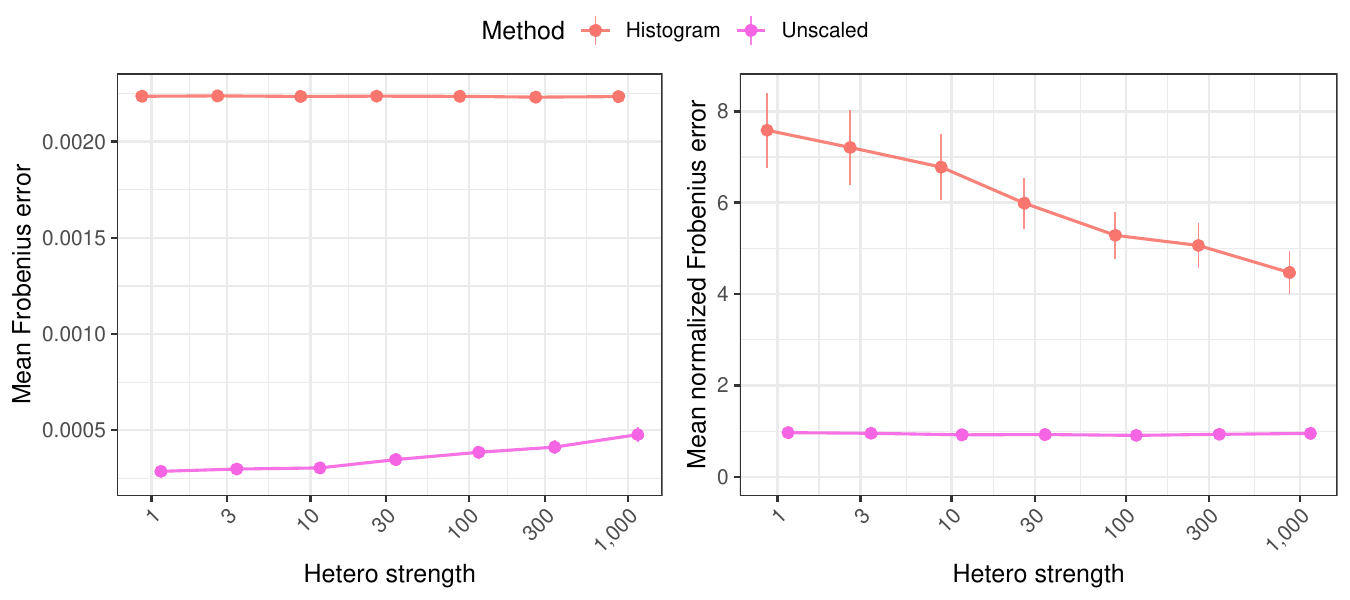}
    \caption{Comparison of Frobenius error across different heteroskedasticity strength $H$}
    \label{fig_exp3_F}
\end{figure}


\subsection{$\ell_1$ error}

In this section, we study the $\ell_1$ estimation of $\P$. We compare the scaled versions of our proposed estimation procedures with the baseline histogram estimator and examine whether the empirical behavior matches the minimax error rate. We consider the following methods.
\begin{enumerate}
    \item \texttt{Histogram}: the pooled empirical histogram $\Y/n$.
    \item \texttt{Oracle}: Algorithm~\ref{algorithm_multinomial_tucker_wo_thinning} with the oracle scaling tensor $\M$ constructed from the true model parameters as in Section \ref{sec:oracle-scaling}. 
    \item \texttt{Oracle-CP}: Algorithm~\ref{algorithm_multinomial_tucker_wo_thinning} with the estimated oracle scaling tensor $\M$ constructed from CP decomposition by alternating least squares. 
    \item \texttt{Slice-oracle}: Algorithm~\ref{algorithm_multinomial_tucker_wo_thinning} with the slice-normalization tensor $\M$ built from the true slice marginals as in Section \ref{sec:slice-scaling}. 
    \item \texttt{Slice-est}: Algorithm~\ref{algorithm_multinomial_tucker_wo_thinning} with the estimated slice-normalization tensor $\M$ built from the empirical slice marginals. 
\end{enumerate}
For \texttt{Oracle-CP} and \texttt{Slice-est}, no independent pilot sample is used; the scaling tensor is estimated from the same histogram tensor $\Y$ used in the final estimator.

\paragraph{Varying sample size $n$.} 
In the first experiment, we investigate the error as the sample size $n$ varies. 
We fix $R=4,~H = 100,~w_r = 1/R$, and $p_1=p_2=p_3=50$. 
For each setting, we repeat the experiment independently over 50 Monte Carlo replications and report the average and standard error of the estimation errors. 
The result is visualized in Figure~\ref{fig_exp2_l1}. 

In the left panel, we plot the $\ell_1$ error against different sample sizes. It can be seen that all of our proposed approaches demonstrate clear superiority over the \texttt{Histogram} estimator. 
In the middle panel, we remove the \texttt{Histogram} method to better compare the remaining approaches. It can be observed that the \texttt{Oracle} achieves slightly better performance than the other two methods, \texttt{Slice-oracle} and \texttt{Slice-est}. 
In the right panel, we normalize the $\ell_1$ error by the minimax error rate. Specifically, for each simulation experiment, we compute
\bea\label{eq_l1_sacle}
    \operatorname{Normalized\, Error} = \|\widehat{\P}-\P\|_1 \cdot \sqrt{\frac{n}{R p_{\max}}}. 
\eea
The normalized errors appear approximately flat in the right panel, which is consistent with the theoretical minimax error rate. 
Additionally, the performance of \texttt{Slice-oracle} and \texttt{Slice-est} is extremely close. This is expected, since estimating slice marginals only involves estimating dimension-$p$ probability vectors, which can be done accurately given the large sample sizes $n$ considered here. 

Notably, CP decomposition via ALS performs poorly in this setting. 
See the simulation for varying heteroskedasticity strength $H$ in Figure~\ref{fig_exp3_l1} for more details. 
As a result, the errors of \texttt{Oracle-CP} are extremely large and fall outside the plotting range, so they are omitted from the figure. 
Developing tensor CP decomposition methods that are robust to heteroskedastic noise is an interesting direction for future research. 

\begin{figure}[htbp]
    \centering
    \includegraphics[width=\linewidth]{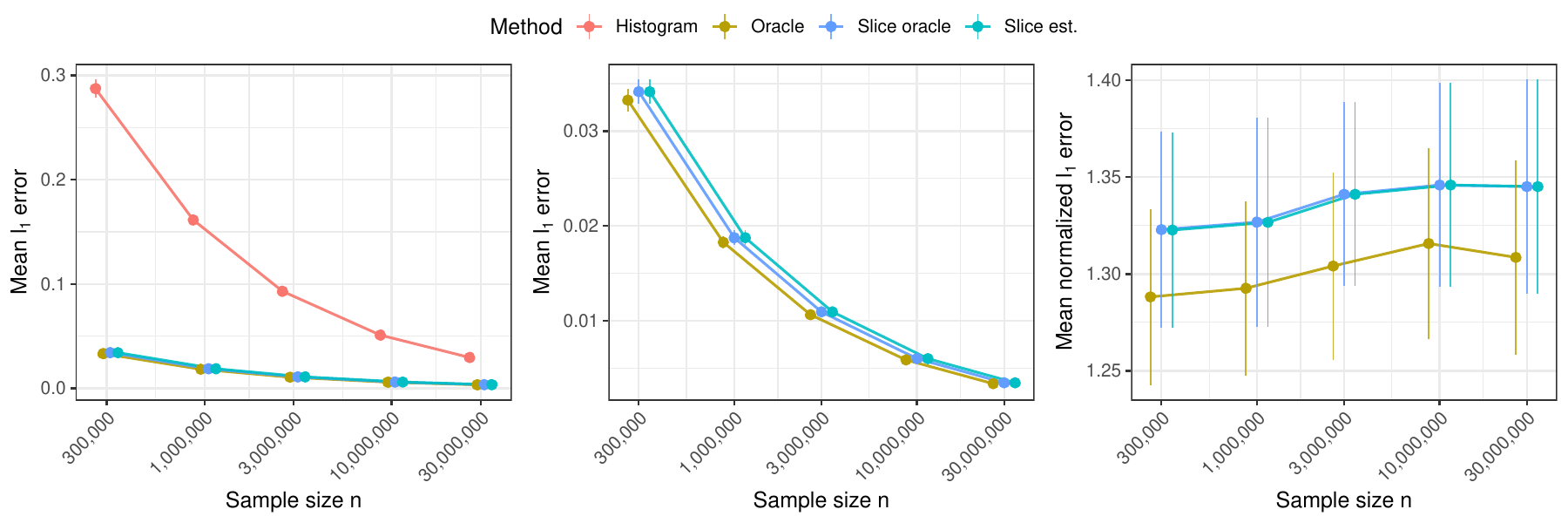}
    \caption{Comparison of $\ell_1$ error across different sample sizes $n$}
    \label{fig_exp2_l1}
\end{figure}

\paragraph{Varying heteroskedasticity strength $H$.}

In this experiment, we investigate how the estimation error changes as the heteroskedasticity strength $H$ varies. 
We fix $R=4,~n = 200000,~w_r = 1/R$, and $p_1=p_2=p_3=50$. 
For each setting, we repeat the experiment independently over 50 Monte Carlo replications and report the average and standard error of the estimation errors. 
The results are visualized in Figure~\ref{fig_exp3_l1}. 
In the left panel, we observe that the performance of \texttt{Oracle-CP} deteriorates as the heteroskedasticity strength $H$ increases. This occurs because the ALS procedure used in CP decomposition is not robust to heteroskedastic noise. 
In the right panel, we remove \texttt{Oracle-CP} and normalize the error by \eqref{eq_l1_sacle}, to better compare the remaining methods. Again, we observe the clear superiority of our proposed approaches over the \texttt{Histogram} estimator. 
Moreover, the estimation error does not increase as the heteroskedasticity strength $H$ increases, demonstrating the robustness of our proposed approaches to heteroskedasticity.
Similar to the Frobenius error in the right panel of Figure~\ref{fig_exp3_F}, the normalized \texttt{Histogram} estimation error decreases as $H$ increases. This may occur because when the probability tensor $\P$ becomes more heteroskedastic, the mass concentrates on fewer entries, which reduces the effective dimension of $\P$ and leads to improved performance of the \texttt{Histogram} estimator. 

\begin{figure}[htbp]
    \centering
    \includegraphics[width=0.8\linewidth]{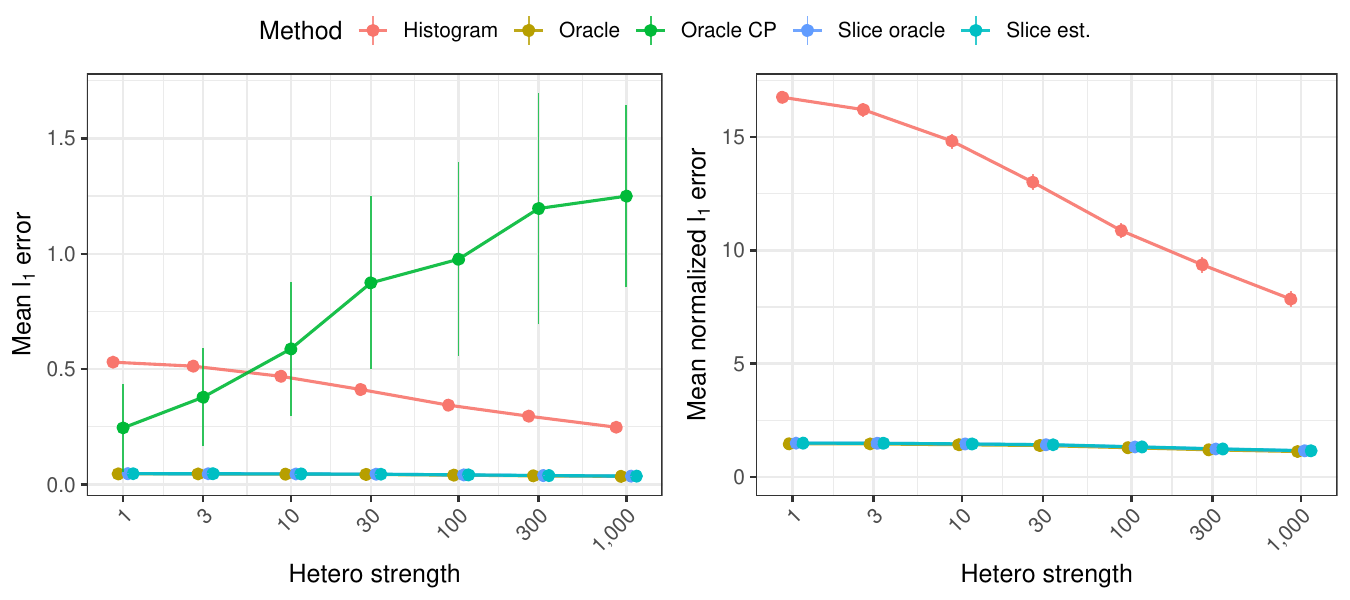}
    \caption{Comparison of $\ell_1$ error across heteroskedasticity strength $H$}
    \label{fig_exp3_l1}
\end{figure}

\paragraph{Varying dimension $p$ and rank $R$.}

In this experiment, we investigate how the estimation error changes as the dimension $p$ and rank $R$ vary. 
We let $w_r = 1/R,~H = 1,$ and $p_1=p_2=p_3=p$ with $p \in \{60, 80, 100, 120\}$ and $R \in \{2, 4, 6, 8\}$. We set $n = C p R$ with $C = 120$. 
For each setting, we repeat the experiment independently over 30 Monte Carlo replications and report the average and standard error of the estimation errors. 
The results are visualized in Figure~\ref{fig_exp4_l1}. 

Note that the sample size $n$ is chosen according to the scaling suggested by the minimax error rate, and the error is normalized by \eqref{eq_l1_sacle}. 
The normalized errors appear approximately flat across all four panels, indicating that the dependence on $p$ predicted by the theoretical minimax error rate is accurate across different values of $R$. 
Additionally, as $R$ increases, \texttt{Oracle} is outperformed by the slice-normalization methods, \texttt{Slice-oracle} and \texttt{Slice-est}. 
This observation is consistent with our theoretical upper bounds for \texttt{Oracle}: the dependence on $R$ is relatively strong, as indicated in Corollary~\ref{thm_multiview_tensor_decomp_0}.

\begin{figure}[htbp]
    \centering
    \includegraphics[width=0.9\linewidth]{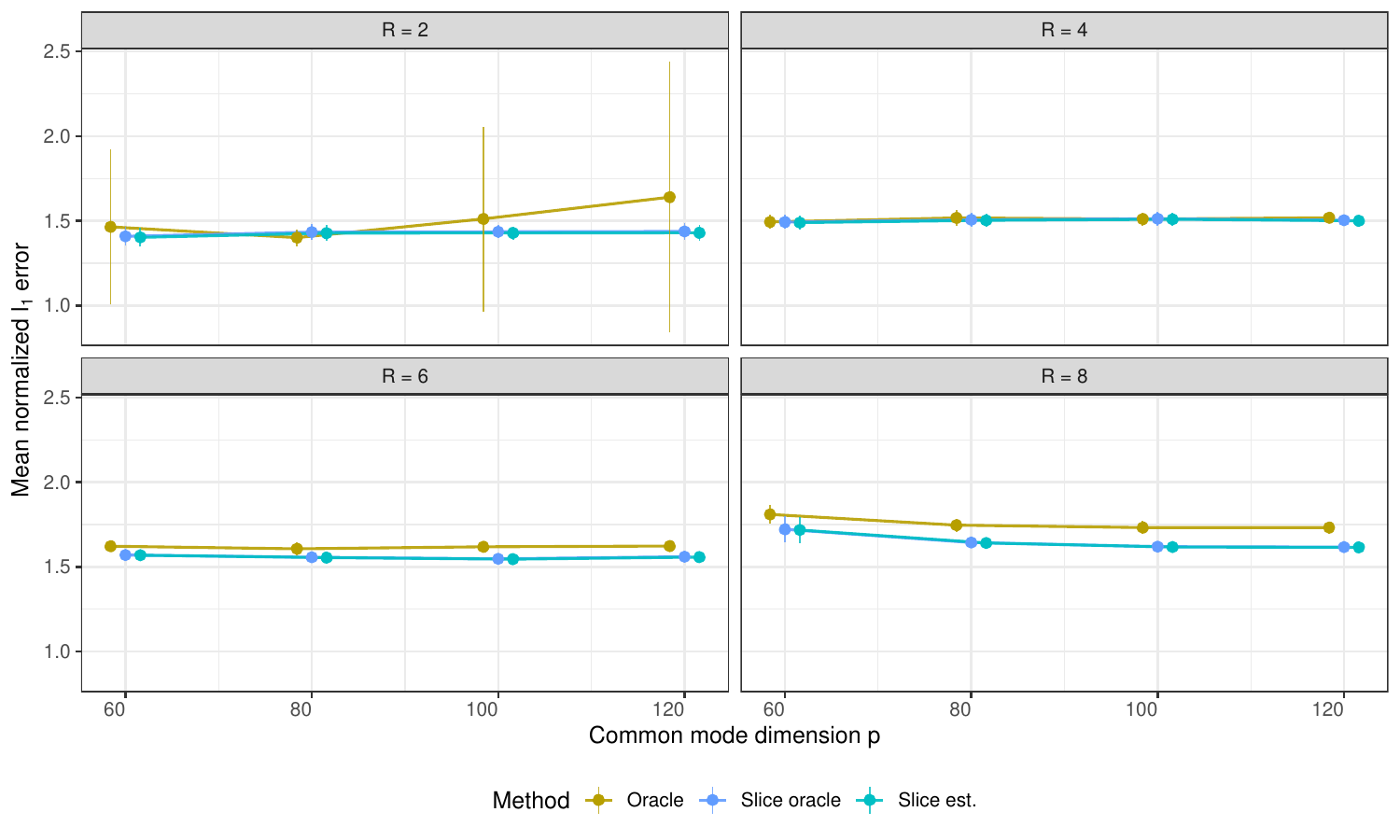}
    \caption{Comparison of $\ell_1$ error across different dimension $p$ and rank $R$}
    \label{fig_exp4_l1}
\end{figure}

\section{Discussions}\label{sec:discussions}

This paper studies multiview density estimation under multinomial sampling, with an emphasis on how heteroskedasticity interacts with low-rank tensor structure. For estimation under Frobenius norm, our results show that the local distribution of probability mass is essential: the maximum fiber mass appears in both the upper and lower bounds, and therefore governs the statistical minimax rate in the fixed-rank regime. This is different from homoskedastic tensor estimation, where the difficulty is usually summarized mainly by dimension and rank. 

For density estimation under $\ell_1$ loss, however, the minimax rate is governed by the ambient dimension and rank, rather than directly by the fiber mass of the original tensor. Yet, a naive estimation procedure under Frobenius norm does not generally yield the optimal $\ell_1$ rate because the multinomial noise is highly heteroskedastic. Rank-one scaling provides a way to rebalance the variance profile while preserving the low-rank structure of the density tensor. In this sense, scaling is not only a technical device, but also the mechanism that connects heteroskedastic estimation under Frobenius norm loss to optimal density estimation under $\ell_1$ loss.

The oracle scaling clarifies what can be achieved when approximate coordinatewise information about the latent factors is available. It balances the scaled variance profile while keeping the inverse scaling stable, leading to the optimal $\ell_1$ rate for fixed rank. Since this construction depends on unknown factor-level quantities, it serves mainly as a benchmark. The feasible slice-normalization rule replaces this oracle information by low-dimensional marginals and can be data-driven. Our theory shows that, under mild structural conditions, this feasible scaling is near-optimal; the simulations further suggest that its performance is close to that of the oracle procedure in representative heteroskedastic settings. 

Several questions remain open. First, our polynomial-time upper bounds require signal-to-noise conditions that ensure accurate spectral subspace estimation. These assumptions are natural for the proposed method. Our results establish optimal statistical rates once the signal is strong enough, but they do not characterize the sharp computational threshold for estimation. Determining whether these signal-to-noise conditions can be weakened, or whether a genuine statistical-computational gap exists under heteroskedastic and dependent multinomial noise, is an important direction for future work.

Second, the current analysis uses thinning to separate initialization, subspace refinement, and final projection. This simplifies the proof by reducing dependence between estimated subspaces and the noise used in later stages. 
Developing a theory directly for the no-thinning implementation, or for partially split procedures, may improve constants and lead to sharper finite-sample guarantees. 

Finally, we have treated the target rank as known. Rank selection is important in applications such as latent class models and topic models, where the number of latent components is rarely known in advance. Since the noise level is highly nonuniform, rank selection procedures designed for homoskedastic models may not be appropriate. Constructing adaptive methods that jointly choose the rank, estimate the scaling, and recover the density tensor is an interesting problem. 

\section*{Acknowledgements}
The work of Olga Klopp was supported by Labex MME-DII, la Fondation de Sciences de la Modélisation \& CY Initiative. The work of Alexandre B. Tsybakov was supported by Labex ECODEC (ANR-11-LABEX-0047) and by ANR MaLIP (ANR-25-CE40-3228-01). 

\bibliographystyle{alpha}
\bibliography{biblio}

@article{yan2024inference,
  title={Inference for heteroskedastic PCA with missing data},
  author={Yan, Yuling and Chen, Yuxin and Fan, Jianqing},
  journal={The Annals of Statistics},
  volume={52},
  number={2},
  pages={729--756},
  year={2024},
  publisher={Institute of Mathematical Statistics}
}

@article{agterberg2024statistical,
  title={Statistical inference for low-rank tensors: Heteroskedasticity, subgaussianity, and applications},
  author={Agterberg, Joshua and Zhang, Anru},
  journal={arXiv preprint arXiv:2410.06381},
  year={2024}
}

@article{vershynin2010introduction,
  title={Introduction to the non-asymptotic analysis of random matrices},
  author={Vershynin, Roman},
  journal={arXiv preprint arXiv:1011.3027},
  year={2010}
}

@article{minskerwei2019moment,
  author  = {Minsker, Stanislav and Wei, Xiaohan},
  title   = {Moment inequalities for matrix-valued U-statistics of order 2},
  journal = {Electronic Journal of Probability},
  volume  = {24},
  number  = {133},
  pages   = {1--32},
  year    = {2019}
}

@article{delapena1995decoupling,
  author  = {de la Pe\~na, Victor H. and Montgomery-Smith, Stephen J.},
  title   = {Decoupling Inequalities for the Tail Probabilities of Multivariate U-statistics},
  journal = {The Annals of Probability},
  volume  = {23},
  number  = {2},
  pages   = {806--816},
  year    = {1995}
}

@article{kasahara2014non,
  title={Non-parametric identification and estimation of the number of components in multivariate mixtures},
  author={Kasahara, Hiroyuki and Shimotsu, Katsumi},
  journal={Journal of the Royal Statistical Society Series B: Statistical Methodology},
  volume={76},
  number={1},
  pages={97--111},
  year={2014},
  publisher={Oxford University Press}
}

@article{lesperance1992algorithm,
  title={An algorithm for computing the nonparametric MLE of a mixing distribution},
  author={Lesperance, Mary L and Kalbfleisch, John D},
  journal={Journal of the American Statistical Association},
  volume={87},
  number={417},
  pages={120--126},
  year={1992},
  publisher={Taylor \& Francis}
}

@article{laird1978nonparametric,
  title={Nonparametric maximum likelihood estimation of a mixing distribution},
  author={Laird, Nan},
  journal={Journal of the American Statistical Association},
  volume={73},
  number={364},
  pages={805--811},
  year={1978},
  publisher={Taylor \& Francis}
}

@article{aragam2023uniform,
  title={Uniform consistency in nonparametric mixture models},
  author={Aragam, Bryon and Yang, Ruiyi},
  journal={The Annals of Statistics},
  volume={51},
  number={1},
  pages={362--390},
  year={2023},
  publisher={Institute of Mathematical Statistics}
}

@article{HallZhou2003,
  author  = {Hall, Peter and Zhou, Xiao-Hua},
  title   = {Nonparametric Estimation of Component Distributions in a Multivariate Mixture},
  journal = {The Annals of Statistics},
  year    = {2003},
  volume  = {31},
  number  = {1},
  pages   = {201--224},
  doi     = {10.1214/aos/1046294462}
}

@article{BenagliaChauveauHunter2009,
  author  = {Benaglia, Tatiana and Chauveau, Didier and Hunter, David R.},
  title   = {An {EM}-Like Algorithm for Semi- and Nonparametric Estimation in Multivariate Mixtures},
  journal = {Journal of Computational and Graphical Statistics},
  year    = {2009},
  volume  = {18},
  number  = {2},
  pages   = {505--526},
  doi     = {10.1198/jcgs.2009.07175}
}

@article{LevineHunterChauveau2011,
  author  = {Levine, Michael and Hunter, David R. and Chauveau, Didier},
  title   = {Maximum Smoothed Likelihood for Multivariate Mixtures},
  journal = {Biometrika},
  year    = {2011},
  volume  = {98},
  number  = {2},
  pages   = {403--416},
  doi     = {10.1093/biomet/asq079}
}

@article{ChauveauHunterLevine2015,
  author  = {Chauveau, Didier and Hunter, David R. and Levine, Michael},
  title   = {Semi-Parametric Estimation for Conditional Independence Multivariate Finite Mixture Models},
  journal = {Statistics Surveys},
  year    = {2015},
  volume  = {9},
  pages   = {1--31},
  doi     = {10.1214/15-SS108}
}

@article{XiangWangYao2019,
  author  = {Xiang, Sijia and Wang, Weixin and Yao, Weixin},
  title   = {An Overview of Semiparametric Extensions of Finite Mixture Models},
  journal = {Statistical Science},
  year    = {2019},
  volume  = {34},
  number  = {3},
  pages   = {391--404},
  doi     = {10.1214/19-STS698}
}

@article{RitchieVandermeulenScott2020,
  author  = {Ritchie, Alexander and Vandermeulen, Robert A. and Scott, Clayton},
  title   = {Consistent Estimation of Identifiable Nonparametric Mixture Models from Grouped Observations},
  journal = {Advances in Neural Information Processing Systems},
  year    = {2020},
  volume  = {33},
  pages   = {11603--11614}
}

@article{KargasSidiropoulosFu2018,
  author  = {Kargas, Nikos and Sidiropoulos, Nicholas D. and Fu, Xiao},
  title   = {Tensors, Learning, and {Kolmogorov} Extension for Finite-Alphabet Random Vectors},
  journal = {IEEE Transactions on Signal Processing},
  year    = {2018},
  volume  = {66},
  number  = {18},
  pages   = {4854--4868},
  doi     = {10.1109/TSP.2018.2862383}
}

@article{VandermeulenSaitenmacher2024,
  author  = {Vandermeulen, Robert A. and Saitenmacher, Ren{\'e}},
  title   = {Generalized Identifiability Bounds for Mixture Models With Grouped Samples},
  journal = {IEEE Transactions on Information Theory},
  year    = {2024},
  volume  = {70},
  pages   = {2746--2758},
  doi     = {10.1109/TIT.2024.3367433}
}

@article{AllmanMatiasRhodes2009,
  author  = {Allman, Elizabeth S. and Matias, Catherine and Rhodes, John A.},
  title   = {Identifiability of Parameters in Latent Structure Models with Many Observed Variables},
  journal = {The Annals of Statistics},
  year    = {2009},
  volume  = {37},
  number  = {6A},
  pages   = {3099--3132},
  doi     = {10.1214/09-AOS689}
}

@article{BleiNgJordan2003,
  author  = {Blei, David M. and Ng, Andrew Y. and Jordan, Michael I.},
  title   = {Latent Dirichlet Allocation},
  journal = {Journal of Machine Learning Research},
  year    = {2003},
  volume  = {3},
  pages   = {993--1022}
}

@article{tang2025revisit,
  title={Revisit CP Tensor Decomposition: Statistical Optimality and Fast Convergence},
  author={Tang, Runshi and Chhor, Julien and Klopp, Olga and Zhang, Anru R},
  journal={arXiv preprint arXiv:2505.23046},
  year={2025}
}

@article{zhou_deflated_2024,
  author  = {Zhou, Yuchen and Chen, Yuxin},
  title   = {Deflated {HeteroPCA}: Overcoming the curse of ill-conditioning in heteroskedastic {PCA}},
  journal = {The Annals of Statistics},
  volume  = {53},
  number  = {1},
  pages   = {91--116},
  year    = {2025},
  doi     = {10.1214/24-AOS2456}
}

@article{chen_spectral_2021,
	title = {Spectral {Methods} for {Data} {Science}: {A} {Statistical} {Perspective}},
	volume = {14},
	issn = {1935-8237, 1935-8245},
	shorttitle = {Spectral {Methods} for {Data} {Science}},
	url = {http://arxiv.org/abs/2012.08496},
	doi = {10.1561/2200000079},
	number = {5},
	urldate = {2025-05-31},
	journal = {Foundations and Trends® in Machine Learning},
	author = {Chen, Yuxin and Chi, Yuejie and Fan, Jianqing and Ma, Cong},
	year = {2021},
	note = {arXiv:2012.08496 [stat]},
	pages = {566--806},
}

@article{zhang_heteroskedastic_2021,
  author  = {Zhang, Anru R. and Cai, T. Tony and Wu, Yihong},
  title   = {Heteroskedastic {PCA}: Algorithm, optimality, and applications},
  journal = {The Annals of Statistics},
  volume  = {50},
  number  = {1},
  pages   = {53--80},
  year    = {2022},
  doi     = {10.1214/21-AOS2074}
}

@article{cai2018rate,
author = {T. Tony Cai and Anru Zhang},
title = {{Rate-optimal perturbation bounds for singular subspaces with applications to high-dimensional statistics}},
volume = {46},
journal = {The Annals of Statistics},
number = {1},
publisher = {Institute of Mathematical Statistics},
pages = {60 -- 89},
year = {2018},
doi = {10.1214/17-AOS1541},
URL = {https://doi.org/10.1214/17-AOS1541}
}

@article{anandkumar2014tensor,
  title={Tensor decompositions for learning latent variable models.},
  author={Anandkumar, Animashree and Ge, Rong and Hsu, Daniel J and Kakade, Sham M and Telgarsky, Matus and others},
  journal={J. Mach. Learn. Res.},
  volume={15},
  number={1},
  pages={2773--2832},
  year={2014}
}

@article{zheng2020nonparametric,
  title={Nonparametric estimation of multivariate mixtures},
  author={Zheng, Chaowen and Wu, Yichao},
  journal={Journal of the American Statistical Association},
  volume={115},
  number={531},
  pages={1456--1471},
  year={2020},
  publisher={Taylor \& Francis}
}

@article{xu2025multivariate,
  title={Multivariate Poisson intensity estimation via low-rank tensor decomposition},
  author={Xu, Haotian and Padilla, Carlos Misael Madrid and Padilla, Oscar Hernan Madrid and Wang, Daren},
  journal={arXiv preprint arXiv:2504.15879},
  year={2025}
}

@article{chhor2024generalized,
  author  = {Chhor, Julien and Klopp, Olga and Tsybakov, Alexandre B.},
  title   = {Generalized multi-view model: Adaptive density estimation under low-rank constraints},
  journal = {Journal of Machine Learning Research},
  volume  = {26},
  number  = {236},
  pages   = {1--52},
  year    = {2025}
}

@inproceedings{hofmann1999,
  title={Probabilistic latent semantic indexing},
  author={Hofmann, Thomas},
  booktitle={Proceedings of the 22nd annual international ACM SIGIR conference on Research and development in information retrieval},
  pages={50--57},
  year={1999}
}

@article{Vandermeulen2023,
author = {Vandermeulen, Robert A.},
title = {Sample Complexity Using Infinite Multiview Models},
year = {2023},
journal={arXiv preprint arXiv:2302.04292}
}

@inproceedings{Jain2020linear,
author = {Jain, Ayush and Orlitsky, Alon},
title = {Linear-Sample Learning of Low-Rank Distributions},
year = {2020},
isbn = {9781713829546},
booktitle = {Proceedings of the 34th International Conference on Neural Information Processing Systems},
pages={19201--19211},
series = {NeurIPS'20}
}

@article{Vandermeulen2020ImprovingND,
  title={Improving Nonparametric Density Estimation with Tensor Decompositions},
  author={Robert A. Vandermeulen},
  journal={ArXiv},
  year={2020},
  volume={abs/2010.02425}
}

@InProceedings{pmlr-v89-kargas19a,
  title = 	 {Learning Mixtures of Smooth Product Distributions: Identifiability and Algorithm},
  author =       {Kargas, Nikos and Sidiropoulos, Nicholas D.},
  booktitle = 	 {Proceedings of the 22nd International Conference on Artificial Intelligence and Statistics},
  pages = 	 {388--396},
  year = 	 {2019},
  publisher =    {PMLR},
  url = 	 {https://proceedings.mlr.press/v89/kargas19a.html}
}

@ARTICLE{9740538,
  author={Amiridi, Magda and Kargas, Nikos and Sidiropoulos, Nicholas D.},
  journal={IEEE Transactions on Signal Processing}, 
  title={Low-Rank Characteristic Tensor Density Estimation Part {II}: Compression and Latent Density Estimation}, 
  year={2022},
  volume={70},
  number={},
  pages={2669-2680},
  doi={10.1109/TSP.2022.3158422}}

@ARTICLE{9779133,
  author={Amiridi, Magda and Kargas, Nikos and Sidiropoulos, Nicholas D.},
  journal={IEEE Transactions on Signal Processing}, 
  title={Low-Rank Characteristic Tensor Density Estimation Part {I}: Foundations}, 
  year={2022},
  volume={70},
  number={},
  pages={2654-2668},
  doi={10.1109/TSP.2022.3175608}}

@article{vandermeulen2021beyond,
  title={Beyond Smoothness: Incorporating Low-Rank Analysis into Nonparametric Density Estimation},
  author={Vandermeulen, Robert A. and Ledent, Antoine},
  journal={Advances in Neural Information Processing Systems},
  volume={34},
  pages={12180--12193},
  year={2021}
}

@article{ibrahim2021recovering,
  title={Recovering joint probability of discrete random variables from pairwise marginals},
  author={Ibrahim, Shahana and Fu, Xiao},
  journal={IEEE Transactions on Signal Processing},
  year={2021},
  publisher={IEEE}
}

@article{yeredor2019maximum,
  title={Maximum likelihood estimation of a low-rank probability mass tensor from partial observations},
  author={Yeredor, Arie and Haardt, Martin},
  journal={IEEE Signal Processing Letters},
  volume={26},
  number={10},
  pages={1551--1555},
  year={2019},
  publisher={IEEE}
}

@inproceedings{ibrahim2020recovering,
  title={Recovering Joint PMF from Pairwise Marginals},
  author={Ibrahim, Shahana and Fu, Xiao},
  booktitle={2020 54th Asilomar Conference on Signals, Systems, and Computers},
  pages={356--360},
  year={2020},
  organization={IEEE}
}

@article{vora2021recovery,
  title={Recovery of Joint Probability Distribution from one-way marginals: Low rank Tensors and Random Projections},
  author={Vora, Jian and Gurumoorthy, Karthik S and Rajwade, Ajit},
  journal={arXiv preprint arXiv:2103.11864},
  year={2021}
}

@article{kruskal1977three,
  title={Three-way arrays: rank and uniqueness of trilinear decompositions, with application to arithmetic complexity and statistics},
  author={Kruskal, Joseph B},
  journal={Linear algebra and its applications},
  volume={18},
  number={2},
  pages={95--138},
  year={1977},
  publisher={Elsevier}
}

@inproceedings{ma2016polynomial,
  title={Polynomial-time tensor decompositions with sum-of-squares},
  author={Ma, Tengyu and Shi, Jonathan and Steurer, David},
  booktitle={2016 IEEE 57th Annual Symposium on Foundations of Computer Science (FOCS)},
  pages={438--446},
  year={2016},
  organization={IEEE}
}

@inproceedings{jain2014learning,
  title={Learning mixtures of discrete product distributions using spectral decompositions},
  author={Jain, Prateek and Oh, Sewoong},
  booktitle={Conference on Learning Theory},
  pages={824--856},
  year={2014},
  organization={PMLR}
}

@article{feldman2008learning,
  title={Learning Mixtures of Product Distributions over Discrete Domains},
  author={Feldman, Jon and O'Donnell, Ryan and Servedio, Rocco A},
  journal={SIAM Journal on Computing},
  volume={37},
  number={5},
  pages={1536},
  year={2008},
  publisher={Society for Industrial and Applied Mathematics}
}

@inproceedings{freund1999estimating,
  title={Estimating a mixture of two product distributions},
  author={Freund, Yoav and Mansour, Yishay},
  booktitle={Proceedings of the twelfth annual conference on Computational learning theory},
  pages={53--62},
  year={1999}
}

@inproceedings{chaudhuri2008learning,
  title={Learning Mixtures of Product Distributions Using Correlations and Independence.},
  author={Chaudhuri, Kamalika and Rao, Satish},
  booktitle={COLT},
  volume={4},
  pages={9--20},
  year={2008}
}

@inproceedings{gordon2021source,
  title={Source identification for mixtures of product distributions},
  author={Gordon, Spencer and Mazaheri, Bijan H and Rabani, Yuval and Schulman, Leonard},
  booktitle={Conference on Learning Theory},
  pages={2193--2216},
  year={2021},
  organization={PMLR}
}

@inproceedings{bhaskara2014smoothed,
  title={Smoothed analysis of tensor decompositions},
  author={Bhaskara, Aditya and Charikar, Moses and Moitra, Ankur and Vijayaraghavan, Aravindan},
  booktitle={Proceedings of the forty-sixth annual ACM symposium on Theory of computing},
  pages={594--603},
  year={2014}
}

@inproceedings{rabani2014learning,
  title={Learning mixtures of arbitrary distributions over large discrete domains},
  author={Rabani, Yuval and Schulman, Leonard J and Swamy, Chaitanya},
  booktitle={Proceedings of the 5th conference on Innovations in theoretical computer science},
  pages={207--224},
  year={2014}
}

@inproceedings{li2015learning,
  title={Learning arbitrary statistical mixtures of discrete distributions},
  author={Li, Jian and Rabani, Yuval and Schulman, Leonard J and Swamy, Chaitanya},
  booktitle={Proceedings of the forty-seventh annual ACM symposium on Theory of computing},
  pages={743--752},
  year={2015}
}

@article{vandermeulen2019operator,
  title={An operator theoretic approach to nonparametric mixture models},
  author={Vandermeulen, Robert A. and Scott, Clayton D},
  journal={The Annals of Statistics},
  volume={47},
  number={5},
  pages={2704--2733},
  year={2019},
  publisher={Institute of Mathematical Statistics}
}

@article{zhang2018tensor,
  title={Tensor {SVD}: Statistical and computational limits},
  author={Zhang, Anru and Xia, Dong},
  journal={IEEE Transactions on Information Theory},
  volume={64},
  number={11},
  pages={7311--7338},
  year={2018},
  publisher={IEEE}
}

@article{kolda2009tensor,
  title={Tensor decompositions and applications},
  author={Kolda, Tamara G and Bader, Brett W},
  journal={SIAM review},
  volume={51},
  number={3},
  pages={455--500},
  year={2009},
  publisher={SIAM}
}

\clearpage
\appendix

\begin{center}
{\bf\LARGE Supplementary Materials for ``Optimal Estimation of Discrete Multiview Distributions under Heteroskedastic Multinomial Sampling''}
\end{center}

\section{HeteroPCA and Deflated HeteroPCA}\label{sec_DeflatedHeteroPCA}

In this appendix section, we record the heteroskedastic PCA subroutines used in the initialization step of both the no-thinning estimator in Algorithm~\ref{algorithm_multinomial_tucker_wo_thinning} and the multinomial-thinning estimator in Algorithm~\ref{algorithm_multinomial_tucker}. 
The goal of these subroutines is to estimate the leading signal subspace of a low-rank matrix from a noisy Gram-type matrix whose diagonal entries may be substantially more contaminated than its off-diagonal entries. 
This situation naturally appears in our tensor problem. 
For example, in the mode-$k$ initialization step of the no-thinning estimator, we form
\[
G_{0,k}
=
P_{\mathrm{off\text{-}diag}}
\bigl(\MM_k(\tilde\Y)\MM_k(\tilde\Y)^\top\bigr).
\]
For the multinomial-thinning estimator, the same construction is applied to the first thinned and scaled tensor $\tilde\Y_1$.

The off-diagonal entries contain useful information about the row-space covariance structure of the signal, whereas the diagonal entries are strongly affected by the heteroskedastic multinomial noise. 
A direct eigendecomposition of the raw Gram matrix can therefore be unstable, especially when the diagonal noise dominates the low-rank signal. 
HeteroPCA addresses this issue by iteratively replacing the unreliable diagonal entries with the diagonal entries implied by a low-rank approximation.

For a square matrix $G$, let $\mathcal P_{\mathrm{diag}}(G)$ denote its diagonal part and let $\mathcal P_{\mathrm{off\text{-}diag}}(G)$ denote its off-diagonal part. 
Given an input matrix $G_{\mathrm{in}}$, a target rank $r$, and an iteration number $t_{\max}$, HeteroPCA starts from $G^{(0)}=G_{\mathrm{in}}$. 
At each iteration, it computes the leading rank-$r$ eigendecomposition of the current matrix $G^{(t)}$ and uses the resulting low-rank approximation to update only the diagonal entries. 
The off-diagonal entries are kept fixed, since they are viewed as the reliable part of the Gram matrix. 
After several iterations, the algorithm returns the estimated rank-$r$ signal subspace.

\begin{algorithm}[ht]
    \caption{HeteroPCA$(G_{\mathrm{in}},\, r,\, t_{\max})$ \citep{zhang_heteroskedastic_2021}}
    \label{algorithm_HPCA}
    \begin{algorithmic}[1]
        \algrenewcommand\algorithmicensure{\textbf{Run:}}
        \algrenewcommand\algorithmicrequire{\textbf{Input:}}
        \Require Symmetric matrix $G_{\mathrm{in}}$, rank $r$, number of iterations $t_{\max}$
        \State \textbf{Initialize:} $G^{(0)} \gets G_{\mathrm{in}}$
        \For{$t = 0,1,\ldots,t_{\max}$}
            \State Compute the rank-$r$ leading eigendecomposition
            \[
            U^{(t)} \Lambda^{(t)} {U^{(t)}}^{\!\top}
            \leftarrow
            \text{top-}r\ \text{eig}(G^{(t)}).
            \]
            \State Update
            \[
            G^{(t+1)}
            \gets
            \mathcal P_{\mathrm{off\text{-}diag}}\!\left(G^{(t)}\right)
            +
            \mathcal P_{\mathrm{diag}}\!\left(
            U^{(t)} \Lambda^{(t)} {U^{(t)}}^{\!\top}
            \right).
            \]
        \EndFor\\
        \Return Matrix estimate $G \gets G^{(t_{\max})}$ and subspace estimate $U \gets U^{(t_{\max})}$
    \end{algorithmic}
\end{algorithm}

The basic HeteroPCA algorithm works well when the leading singular values are reasonably well conditioned. 
However, when the spectrum of the signal matrix is ill-conditioned, applying HeteroPCA directly with the full target rank $r$ may be suboptimal. 
The reason is that the weaker components can be masked by the stronger components and by the heteroskedastic diagonal perturbation. 
Deflated-HeteroPCA \citep{zhou_deflated_2024} is designed to overcome this issue by recovering the signal subspace in several spectral blocks. 
Instead of estimating all $r$ directions at once, it progressively identifies groups of singular directions whose singular values are of comparable size and are separated from the remaining spectrum.

We now describe the block selection rule. 
Suppose that at the beginning of the $k$th deflation step, the current matrix is $G_{k-1}$ and the previously recovered rank is $r_{k-1}$. 
The algorithm searches for a rank $r^\prime>r_{k-1}$ such that two conditions hold. 
First, the singular values from index $r_{k-1}+1$ to $r^\prime$ are within a constant factor of each other, so they form a well-conditioned block. 
Second, there is a non-negligible eigengap after $r^\prime$, so this block can be separated from the remaining spectrum. 
Formally, define
\be
    \mathcal{R}_k
    :=
    \left\{
    r^{\prime}:
    r_{k-1}<r^{\prime} \leq r,\,
    \frac{\sigma_{r_{k-1}+1}\left(\boldsymbol{G}_{k-1}\right)}
    {\sigma_{r^{\prime}}\left(\boldsymbol{G}_{k-1}\right)}
    \leq 4
    \text{ and }
    \sigma_{r^{\prime}}\left(\boldsymbol{G}_{k-1}\right)
    -
    \sigma_{r^{\prime}+1}\left(\boldsymbol{G}_{k-1}\right)
    \geq
    \frac{1}{r}\sigma_{r^{\prime}}\left(\boldsymbol{G}_{k-1}\right)
    \right\}.
\ee
The next working rank is selected as
\be\label{eq_r_select_DHPCA}
    r_k
    =
    \begin{cases}
    \max \mathcal{R}_k, & \text{if } \mathcal{R}_k \neq \emptyset,\\
    r, & \text{otherwise}.
    \end{cases}
\ee
Thus, whenever a well-conditioned and separated spectral block can be found, the algorithm expands the estimated rank only up to the end of that block. 
If no such block is detected, the algorithm sets $r_k=r$ and completes the remaining estimation in one step.

After selecting $r_k$, Deflated-HeteroPCA applies HeteroPCA to the current matrix $G_{k-1}$ with target rank $r_k$. 
The output matrix $G_k$ is then used as the input for the next deflation step. 
This iterative procedure continues until $r_k=r$. 
The final output is the estimated rank-$r$ subspace. 
In our tensor estimator, this procedure is applied separately to each mode-wise Gram matrix, producing the initial subspace estimates $\hat U_k^{(0)}$ used in the subsequent refinement step.

\begin{algorithm}[ht]
    \caption{Deflated-HeteroPCA \cite{zhou_deflated_2024}}
    \label{algorithm_DHPCA}
    \begin{algorithmic}[1]
        \algrenewcommand\algorithmicensure{\textbf{Run:}}
        \algrenewcommand\algorithmicrequire{\textbf{Input:}}
        \Require $G_0$, target rank $r$, maximum iteration counts $\{t_i\}_{i=1,2,\ldots}$
        \State \textbf{Initialize:} $k \gets 0$, $r_0 \gets 0$
        \While{$r_k < r$}
            \State $k \gets k + 1$
            \State Select $r_k$ by \eqref{eq_r_select_DHPCA}
            \State $(G_k, U_k) \gets \text{HeteroPCA}\!\left(G_{k-1},\, r_k,\, t_k\right)$ by Algorithm~\ref{algorithm_HPCA}
        \EndWhile\\
        \Return Subspace estimate $U \gets U_k$
    \end{algorithmic}
\end{algorithm}

To summarize, HeteroPCA corrects the diagonal bias caused by heteroskedastic noise, while Deflated-HeteroPCA improves robustness to ill-conditioned spectra by estimating the signal subspace in separated spectral blocks. 
These two features are both important for our multinomial tensor problem: the noise is heteroskedastic because cell variances depend on the underlying probabilities, and the mode-wise unfoldings of the signal tensor may have singular values with very different magnitudes.

\section{Additional Simulation: Effect of Multinomial Thinning}
\label{app:sim_thinning}

In this appendix, we compare the no-thinning implementation used in the main simulation section with the multinomial-thinning implementation. The purpose of this experiment is to examine the finite-sample cost of multinomial thinning. We focus on the unscaled estimator, since this comparison isolates the effect of using the same histogram tensor in all stages versus using three thinned histograms for the initialization, refinement, and final projection steps.

We generate the probability tensor $\P$ from the same dense heteroskedastic multiview model described in Section~\ref{sec:simulations}. Specifically, we take $\P = \sum_{r=1}^R w_r \, a_r^{(1)} \circ a_r^{(2)} \circ a_r^{(3)},$
with balanced weights $w_r=1/R$. The factor vectors are generated from the heteroskedastic Dirichlet model with concentration scale $\alpha=0.8$. We set $R=4$, $p_1=p_2=p_3=30$, and the heteroskedasticity strength $H=50$. For each sample size $n$, we generate one histogram tensor $\Y \sim \mathrm{Multinomial}(n,\P).$

We compare the following two implementations of the unscaled Tucker estimator.

\begin{enumerate}
    \item \texttt{Unscaled-no-thinning}: the estimator in Algorithm~\ref{algorithm_multinomial_tucker_wo_thinning}. The same histogram tensor $\Y$ is used to estimate the initial subspaces, refine the subspaces, and form the final projected estimator. The final estimate is normalized by $n$.

    \item \texttt{Unscaled-thinning}: the estimator in Algorithm~\ref{algorithm_multinomial_tucker}. Conditional on $\Y$, we independently split the count in each cell $i=(i_1,i_2,i_3)$ according to
    \[
        (\Y_{1,i},\Y_{2,i},\Y_{3,i}) \mid \Y_i
        \sim
        \mathrm{Multinomial}\left(\Y_i;1/3,1/3,1/3\right).
    \]
    The tensor $\Y_1$ is used to estimate the initial mode-wise subspaces, $\Y_2$ is used for the refinement step, and $\Y_3$ is used for the final projection. The final estimate is normalized by $n_3=\|\Y_3\|_1$.
\end{enumerate}

We vary the sample size over $n \in \{10^4, 3\times 10^4, 10^5, 3\times 10^5, 10^6\}.$ 
For each setting, we repeat the experiment independently over 30 Monte Carlo replications and report the average and standard error of the Frobenius errors.

Figure~\ref{fig_exp1_thinning_F} compares the Frobenius errors. The no-thinning implementation consistently improves over the thinning implementation. This behavior is expected: the thinning estimator uses only about one third of the observations in each stage, whereas the no-thinning estimator uses the full histogram tensor throughout the procedure. In the right panel, we normalize the Frobenius error using the same normalization as in \eqref{eq_F_sacle},
\[
    \|\widehat{\P}-\P\|_F
    \cdot
    \sqrt{\frac{n}{R\operatorname{Fiber}_{\ell_1}(\P)}}.
\]
The normalized errors are relatively stable across $n$, suggesting that both implementations follow the same qualitative sample-size scaling, while the no-thinning implementation has a better finite-sample constant.

\begin{figure}[htbp]
    \centering
    \includegraphics[width=0.8\linewidth]{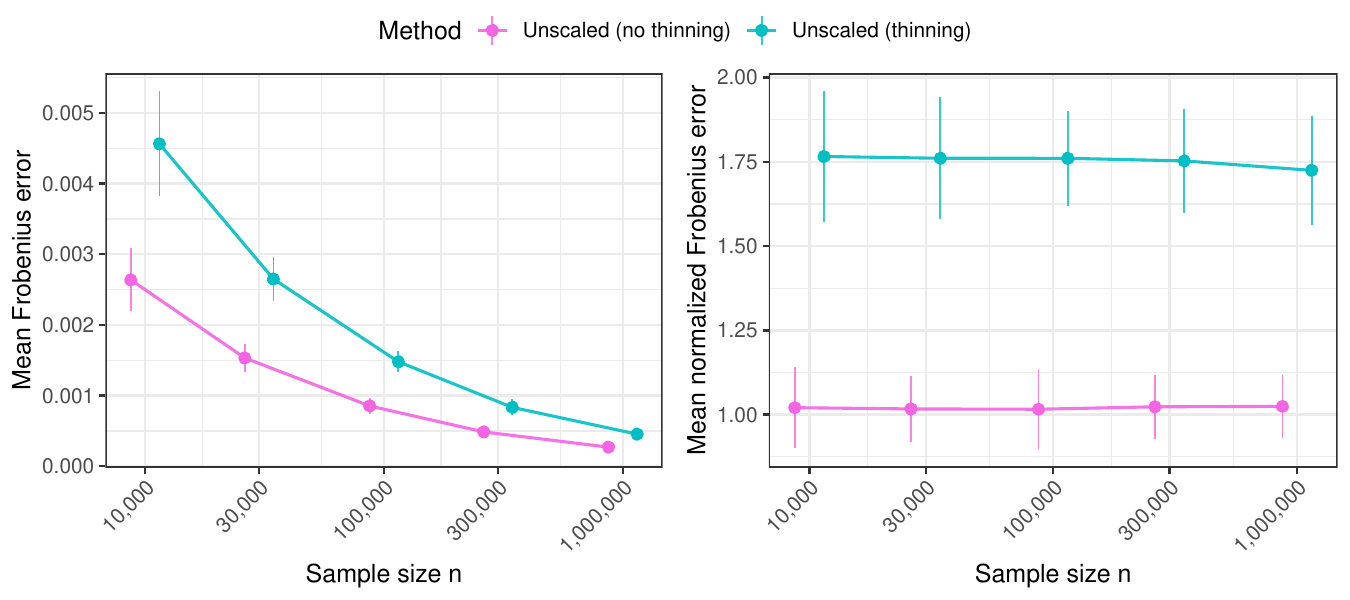}
    \caption{Comparison of the no-thinning and multinomial-thinning implementations for the unscaled estimator under Frobenius error.}
    \label{fig_exp1_thinning_F}
\end{figure}

\section{Projection to the Probability Simplex}
\label{sec_projection}

We justify that the estimator $\hat \P$ in all corollaries can be converted into a density tensor without affecting the statistical guarantees. Let
\[
    m=\prod_{k=1}^d p_k
\]
and identify an order-$d$ tensor in $\RR^{p_1\times\cdots\times p_d}$ with a vector in $\RR^m$ through vectorization. Define the probability simplex
\[
    \Delta_m
    =
    \left\{
        x\in\RR^m:
        x_i\ge 0 \text{ for all } i\in[m],
        \quad
        \sum_{i=1}^m x_i=1
    \right\}.
\]
For an arbitrary estimator $\hat \P$, let $\hat p=\operatorname{vec}(\hat \P)$ and define
\[
    \tilde p
    =
    \Pi_{\Delta_m}(\hat p)
    =
    \argmin_{q\in\Delta_m}
    \|\hat p-q\|_2,
\]
where $\Pi_{\Delta_m}$ denotes the Euclidean projection onto $\Delta_m$. Let $\tilde \P$ be the tensor obtained by reshaping $\tilde p$ back to dimension $p_1\times\cdots\times p_d$. Then $\tilde \P$ is a density tensor.

\begin{Lemma}
\label{lemma_simplex_projection}
Let $\P$ be a density tensor and let $\tilde \P$ be the Euclidean projection of $\hat \P$ onto the probability simplex after vectorization. Then
\[
    \|\tilde \P-\P\|_F
    \le
    \|\hat \P-\P\|_F
\]
and
\[
    \|\tilde \P-\P\|_1
    \le
    2\|\hat \P-\P\|_1.
\]
Consequently, simplex projection does not increase the Frobenius error and preserves the entrywise $\ell_1$ error bound up to a universal constant.
\end{Lemma}

\begin{proof}
Let $\hat p=\operatorname{vec}(\hat \P)$, $p=\operatorname{vec}(\P)$, and $\tilde p=\operatorname{vec}(\tilde \P)$. Since $\P$ is a density tensor, $p\in\Delta_m$. We first prove the Frobenius bound. The Euclidean projection onto a closed convex set is nonexpansive in Euclidean norm. Therefore,
\[
    \|\tilde p-p\|_2
    =
    \|\Pi_{\Delta_m}(\hat p)-\Pi_{\Delta_m}(p)\|_2
    \le
    \|\hat p-p\|_2.
\]
Since vectorization preserves the Euclidean norm, this gives
\[
    \|\tilde \P-\P\|_F
    \le
    \|\hat \P-\P\|_F.
\]

It remains to prove the $\ell_1$ bound. We first show that the Euclidean projection onto the simplex is also an $\ell_1$-closest point to the simplex. For any $x\in\RR^m$, let $z=\Pi_{\Delta_m}(x)$. The Euclidean projection onto the simplex has the form
\[
    z_i=(x_i-\theta)_+,
    \qquad
    i\in[m],
\]
where $\theta\in\RR$ is chosen so that $\sum_{i=1}^m z_i=1$. Let
\[
    x_i^+=\max\{x_i,0\},
    \qquad
    x_i^-=\max\{-x_i,0\},
\]
and write
\[
    A=\sum_{i=1}^m x_i^+,
    \qquad
    B=\sum_{i=1}^m x_i^-.
\]
Since the function
\[
    \theta\mapsto \sum_{i=1}^m (x_i-\theta)_+
\]
is nonincreasing and equals $A$ at $\theta=0$, we have $\theta\ge0$ when $A\ge1$ and $\theta<0$ when $A<1$.

If $\theta\ge0$, then $z_i\le x_i$ for $x_i>0$ and $z_i=0$ for $x_i\le0$. Hence
\[
    \|x-z\|_1
    =
    A+B-1.
\]
If $\theta<0$, then $z_i\ge x_i$ for all $i$, and therefore
\[
    \|x-z\|_1
    =
    1-\sum_{i=1}^m x_i
    =
    1-A+B.
\]

Now take any $q\in\Delta_m$ and define
\[
    t=\sum_{i:x_i<0}q_i.
\]
Then
\[
\begin{aligned}
    \|x-q\|_1
    &\ge
    \sum_{i:x_i<0}(q_i-x_i)
    +
    \left|
        \sum_{i:x_i\ge0}(x_i-q_i)
    \right|  \\
    &=
    B+t
    +
    \left|
        A-1+t
    \right|.
\end{aligned}
\]
If $A\ge1$, then
\[
    B+t+|A-1+t|
    =
    A+B-1+2t
    \ge
    A+B-1.
\]
If $A<1$, then
\[
    B+t+|A-1+t|
    =
    B+t+|t-(1-A)|
    \ge
    B+1-A.
\]
Thus, for every $q\in\Delta_m$,
\[
    \|x-q\|_1
    \ge
    \|x-z\|_1.
\]
Since $z\in\Delta_m$, this proves
\[
    \|x-\Pi_{\Delta_m}(x)\|_1
    =
    \inf_{q\in\Delta_m}\|x-q\|_1.
\]

Applying this identity with $x=\hat p$ and using $p\in\Delta_m$, we obtain
\[
    \|\hat p-\tilde p\|_1
    \le
    \|\hat p-p\|_1.
\]
Therefore, by the triangle inequality,
\[
    \|\tilde p-p\|_1
    \le
    \|\tilde p-\hat p\|_1
    +
    \|\hat p-p\|_1
    \le
    2\|\hat p-p\|_1.
\]
Since vectorization preserves the entrywise $\ell_1$ norm, this proves
\[
    \|\tilde \P-\P\|_1
    \le
    2\|\hat \P-\P\|_1.
\]
The proof is complete.
\end{proof}

As a result, any high-probability bound of the form
\[
    \|\hat \P-\P\|_F\le r_F
\]
immediately implies
\[
    \|\tilde \P-\P\|_F\le r_F.
\]
Similarly, any high-probability bound of the form
\[
    \|\hat \P-\P\|_1\le r_1
\]
implies
\[
    \|\tilde \P-\P\|_1\le 2r_1.
\]
Thus, the projection step does not change any Frobenius rate, and it preserves all $\ell_1$ rates up to an absolute constant.

\section{Fully Empirical Slice Normalization}
\label{sec_appendix_plugin_slice_scaling}

The slice-normalization estimator in Corollary~\ref{thm_slice_scaling} uses the population slice marginals of $\P$.
In this section, we justify a fully empirical plug-in version based on an independent pilot sample.
The independence assumption is used only to decouple the random scaling tensor from the spectral estimation step.
The simulations use the same empirical histogram to estimate the slice marginals and to run the final estimator, and can be viewed as the corresponding non-split implementation.

Define
\[
    s_{k,i}
    =
    \operatorname{Slice}_{\ell_1}^{(k,i)}(\P),
    \qquad
    t_{k,i}
    =
    s_{k,i}\vee p_k^{-1}.
\]
The population slice-normalization vector is
\[
    (b_k)_i=t_{k,i}^{-1/2},
    \qquad
    \M^\prime=b_1\circ\cdots\circ b_d.
\]
Let $\Y_0\sim\operatorname{Multinomial}(n_0,\P)$ be an independent pilot histogram tensor.
Define
\[
    \hat s_{k,i}
    =
    \operatorname{Slice}_{\ell_1}^{(k,i)}(\Y_0/n_0),
    \qquad
    \hat t_{k,i}
    =
    \hat s_{k,i}\vee p_k^{-1},
\]
and set
\[
    (\hat b_k)_i=\hat t_{k,i}^{-1/2},
    \qquad
    \widehat \M_{\rm sl}=\hat b_1\circ\cdots\circ \hat b_d.
\]

\begin{algorithm}[ht]
    \caption{Fully Empirical Slice-Normalized Density Tensor Estimation}
    \label{algorithm_plugin_slice_normalization}
    \begin{algorithmic}[1]
        \algrenewcommand\algorithmicensure{\textbf{Run:}}
        \algrenewcommand\algorithmicrequire{\textbf{Input:}}
        \Require{An independent pilot tensor $\Y_0$, an independent main tensor $\Y$, and target Tucker rank $(r_1,\ldots,r_d)$.}
        \State Let $n_0=\|\Y_0\|_1$.
        \For{$k=1,\ldots,d$}
        \For{$i=1,\ldots,p_k$}
        \State $\hat s_{k,i}=\operatorname{Slice}_{\ell_1}^{(k,i)}(\Y_0/n_0)$.
        \State $(\hat b_k)_i=(\hat s_{k,i}\vee p_k^{-1})^{-1/2}$.
        \EndFor
        \EndFor
        \State $\widehat \M_{\rm sl}=\hat b_1\circ\cdots\circ \hat b_d$.
        \State Apply Algorithm~\ref{algorithm_multinomial_tucker} to the main tensor $\Y$ with scaling tensor $\widehat \M_{\rm sl}$, and denote the output by $\widehat \Q_{\rm plug}$.
        \State $\widehat \P_{\rm plug}=\widehat \Q_{\rm plug}*\widehat \M_{\rm sl}^{(-1)}$.\\
        \Return $\widehat \P_{\rm plug}$.
    \end{algorithmic}
\end{algorithm}

The following lemma controls the estimation error in the first step, the marginal estimation. 

\begin{Lemma}[Uniform control of empirical slice marginals]
\label{lemma_plugin_slice_marginal}
There exists a universal constant $C>0$ such that, for any $\tau\ge 2$, if
\[
    n_0
    \ge
    C\tau p_{\max}\log(e p_{\max}),
\]
then, with probability at least $1-2d p_{\max}^{1-\tau}$,
\[
    \frac12 t_{k,i}
    \le
    \hat t_{k,i}
    \le
    2t_{k,i},
    \qquad
    \text{for all } k\in[d],\ i\in[p_k].
\]
Consequently, on the same event,
\[
    2^{-1/2}(b_k)_i
    \le
    (\hat b_k)_i
    \le
    2^{1/2}(b_k)_i,
    \qquad
    \text{for all } k\in[d],\ i\in[p_k],
\]
and hence
\[
    2^{-d/2}\mathcal M^\prime
    \le
    \widehat{\mathcal M}_{\rm sl}
    \le
    2^{d/2}\mathcal M^\prime
\]
entrywise.
\end{Lemma}

\begin{proof}
Fix $k\in[d]$ and $i\in[p_k]$. Let
\[
    a_k=p_k^{-1}.
\]
Then
\[
    t_{k,i}=s_{k,i}\vee a_k,
    \qquad
    \hat t_{k,i}=\hat s_{k,i}\vee a_k.
\]
Since $\hat s_{k,i}$ is computed from the pilot multinomial sample,
\[
    n_0\hat s_{k,i}
    \sim
    \operatorname{Binomial}(n_0,s_{k,i}).
\]
By Bernstein's inequality for binomial variables, for every $u>0$,
\[
    \mathbb P\left(
        |\hat s_{k,i}-s_{k,i}|\ge u
    \right)
    \le
    2\exp\left(
        -
        \frac{n_0u^2}{2s_{k,i}+2u/3}
    \right).
\]
Take
\[
    u=\frac12 t_{k,i}.
\]
Since $s_{k,i}\le t_{k,i}$, we have
\[
    2s_{k,i}+\frac{2u}{3}
    \le
    2t_{k,i}+\frac{t_{k,i}}{3}
    =
    \frac{7}{3}t_{k,i}.
\]
Therefore, for a universal constant $c>0$,
\[
    \mathbb P\left(
        |\hat s_{k,i}-s_{k,i}|
        \ge
        \frac12 t_{k,i}
    \right)
    \le
    2\exp(-c n_0t_{k,i}).
\]
By the definition of $t_{k,i}$,
\[
    t_{k,i}
    \ge
    p_k^{-1}
    \ge
    p_{\max}^{-1}.
\]
Thus,
\[
    \mathbb P\left(
        |\hat s_{k,i}-s_{k,i}|
        \ge
        \frac12 t_{k,i}
    \right)
    \le
    2\exp\left(-c\frac{n_0}{p_{\max}}\right).
\]
If
\[
    n_0
    \ge
    C\tau p_{\max}\log(e p_{\max})
\]
for a sufficiently large universal constant $C>0$, then
\[
    \mathbb P\left(
        |\hat s_{k,i}-s_{k,i}|
        \ge
        \frac12 t_{k,i}
    \right)
    \le
    2p_{\max}^{-\tau}.
\]
Taking a union bound over all pairs $(k,i)$ gives
\[
    \mathbb P\left(
        \exists k\in[d],\ i\in[p_k]:
        |\hat s_{k,i}-s_{k,i}|
        \ge
        \frac12 t_{k,i}
    \right)
    \le
    \sum_{k=1}^d 2p_kp_{\max}^{-\tau}
    \le
    2d p_{\max}^{1-\tau}.
\]

On the complement of this event,
\[
    |\hat s_{k,i}-s_{k,i}|
    \le
    \frac12 t_{k,i}
\]
for all $k\in[d]$ and $i\in[p_k]$. We now show that this implies the desired multiplicative control of $\hat t_{k,i}$.

For the upper bound,
\[
    \hat t_{k,i}
    =
    \hat s_{k,i}\vee a_k
    \le
    \left(s_{k,i}+\frac12 t_{k,i}\right)\vee a_k
    \le
    \frac32 t_{k,i}
    \le
    2t_{k,i}.
\]
For the lower bound, if $t_{k,i}=a_k$, then
\[
    \hat t_{k,i}
    \ge
    a_k
    =
    t_{k,i}
    \ge
    \frac12 t_{k,i}.
\]
If $t_{k,i}=s_{k,i}$, then
\[
    \hat t_{k,i}
    \ge
    \hat s_{k,i}
    \ge
    s_{k,i}-\frac12 t_{k,i}
    =
    \frac12 t_{k,i}.
\]
Hence, in all cases,
\[
    \frac12 t_{k,i}
    \le
    \hat t_{k,i}
    \le
    2t_{k,i}.
\]

Since
\[
    (b_k)_i=t_{k,i}^{-1/2},
    \qquad
    (\hat b_k)_i=\hat t_{k,i}^{-1/2},
\]
the preceding display implies
\[
    2^{-1/2}(b_k)_i
    \le
    (\hat b_k)_i
    \le
    2^{1/2}(b_k)_i.
\]
Finally, because
\[
    \mathcal M^\prime_{i_1,\ldots,i_d}
    =
    \prod_{k=1}^d (b_k)_{i_k},
    \qquad
    (\widehat{\mathcal M}_{\rm sl})_{i_1,\ldots,i_d}
    =
    \prod_{k=1}^d (\hat b_k)_{i_k},
\]
multiplying the preceding inequalities over $k=1,\ldots,d$ yields
\[
    2^{-d/2}\mathcal M^\prime
    \le
    \widehat{\mathcal M}_{\rm sl}
    \le
    2^{d/2}\mathcal M^\prime
\]
entrywise. This proves the lemma.
\end{proof}

In order to apply Theorem~\ref{thm_multinomial_tensor_generalized}, we need its assumptions to hold for the realized empirical scaling tensor $\widehat \M_{\rm sl}$. The following lemma transfers the required rank, signal-strength, incoherence, and variance-profile bounds from the population scaling $\M^\prime$ to any rank-one scaling that is entrywise comparable to it.

\begin{Lemma}[Stability under constant rank-one rescaling]
\label{lemma_rank_one_rescaling_stability}
Let $\M^\prime=b_1\circ\cdots\circ b_d$ be a positive rank-one tensor, and let $\widehat \M=\hat b_1\circ\cdots\circ \hat b_d$ satisfy $c_0 (b_k)_i\le(\hat b_k)_i\le C_0 (b_k)_i$ for all $k\in[d]$ and $i\in[p_k]$, where $0<c_0\le C_0<\infty$ are constants. Let $\Q^\prime=\P*\M^\prime$ and $\widehat \Q=\P*\widehat \M$. Then $\Q^\prime$ and $\widehat \Q$ have the same mode-wise ranks. Moreover, for every $k\in[d]$,
\[
    c_0^d
    \sigma_j(\MM_k(\Q^\prime))
    \le
    \sigma_j(\MM_k(\widehat \Q))
    \le
    C_0^d
    \sigma_j(\MM_k(\Q^\prime)),
    \qquad
    j\ge 1.
\]
In particular, $\lambda_{\operatorname{Tucker}}(\widehat \Q)\asymp\lambda_{\operatorname{Tucker}}(\Q^\prime)$, where the constants depend only on $d,c_0,C_0$.

If $\Q^\prime$ admits a Tucker decomposition $\Q^\prime=\S^\prime\times_{k\in[d]}V_k$, where each $V_k$ has orthonormal columns and $\|V_k\|_{2,\infty}\le \mu_k$, then $\widehat \Q$ admits a Tucker decomposition $\widehat \Q=\widehat \S\times_{k\in[d]}\widehat V_k$, where each $\widehat V_k$ has orthonormal columns and
\[
    \|\widehat V_k\|_{2,\infty}
    \le
    \frac{C_0}{c_0}\mu_k.
\]
Consequently, if $\rho_\star=\max_{h\in[d]}\|\bigotimes_{k\neq h}V_k\|_{2,\infty}$ and $\widehat \rho_\star=\max_{h\in[d]}\|\bigotimes_{k\neq h}\widehat V_k\|_{2,\infty}$, then
\[
    \widehat \rho_\star
    \le
    \left(\frac{C_0}{c_0}\right)^{d-1}\rho_\star.
\]
Finally,
\[
    \operatorname{Fiber}_{\ell_1}(\P*\widehat \M*\widehat \M)
    \le
    C_0^{2d}
    \operatorname{Fiber}_{\ell_1}(\P*\M^\prime*\M^\prime),
\]
\[
    \operatorname{Slice}_{\ell_1}(\P*\widehat \M*\widehat \M)
    \le
    C_0^{2d}
    \operatorname{Slice}_{\ell_1}(\P*\M^\prime*\M^\prime),
\]
and
\[
    \|\P*\widehat \M*\widehat \M\|_\infty
    \le
    C_0^{2d}
    \|\P*\M^\prime*\M^\prime\|_\infty.
\]
\end{Lemma}

\begin{proof}
Let
\[
    D_k
    =
    \operatorname{diag}
    \left(
        \frac{(\hat b_k)_1}{(b_k)_1},
        \ldots,
        \frac{(\hat b_k)_{p_k}}{(b_k)_{p_k}}
    \right).
\]
Then
\[
    c_0 I_{p_k}
    \preceq
    D_k
    \preceq
    C_0 I_{p_k}.
\]
Moreover,
\[
    \widehat \Q
    =
    \Q^\prime\times_{k\in[d]}D_k.
\]
Thus the mode-$k$ matricization satisfies
\[
    \MM_k(\widehat \Q)
    =
    D_k\MM_k(\Q^\prime)
    \left(
        \bigotimes_{h\neq k}D_h
    \right),
\]
up to the conventional ordering of the Kronecker product.

Because every $D_k$ is invertible, this transformation preserves the mode-wise ranks.
Also,
\[
    \sigma_{\min}(D_k)\ge c_0,
    \qquad
    \sigma_{\max}(D_k)\le C_0,
\]
and
\[
    \sigma_{\min}
    \left(
        \bigotimes_{h\neq k}D_h
    \right)
    \ge
    c_0^{d-1},
    \qquad
    \sigma_{\max}
    \left(
        \bigotimes_{h\neq k}D_h
    \right)
    \le
    C_0^{d-1}.
\]
Therefore,
\[
    c_0^d
    \sigma_j(\MM_k(\Q^\prime))
    \le
    \sigma_j(\MM_k(\widehat \Q))
    \le
    C_0^d
    \sigma_j(\MM_k(\Q^\prime)).
\]

We next prove the incoherence bound.
The mode-$k$ column space of $\widehat \Q$ is
\[
    \operatorname{col}(\MM_k(\widehat \Q))
    =
    D_k\operatorname{col}(\MM_k(\Q^\prime)).
\]
Since $V_k$ has orthonormal columns spanning $\operatorname{col}(\MM_k(\Q^\prime))$, an orthonormal basis for the rescaled column space is
\[
    \widehat V_k
    =
    D_kV_k
    \left(
        V_k^\top D_k^2V_k
    \right)^{-1/2}.
\]
Because
\[
    c_0^2 I
    \preceq
    V_k^\top D_k^2V_k
    \preceq
    C_0^2 I,
\]
we have
\[
    \left\|
        \left(
            V_k^\top D_k^2V_k
        \right)^{-1/2}
    \right\|
    \le
    c_0^{-1}.
\]
Hence, for every row $i$,
\[
    \|(\widehat V_k)_{i,\cdot}\|_2
    \le
    C_0c_0^{-1}
    \|(V_k)_{i,\cdot}\|_2.
\]
Therefore,
\[
    \|\widehat V_k\|_{2,\infty}
    \le
    \frac{C_0}{c_0}
    \|V_k\|_{2,\infty}
    \le
    \frac{C_0}{c_0}\mu_k.
\]
The bound for $\widehat \rho_\star$ follows because each row of a Kronecker product has Euclidean norm equal to the product of the Euclidean norms of the corresponding rows.

Finally,
\[
    \P*\widehat \M*\widehat \M
    =
    (\P*\M^\prime*\M^\prime)
    *
    \left(
        \frac{\widehat \M}{\M^\prime}
    \right)
    *
    \left(
        \frac{\widehat \M}{\M^\prime}
    \right).
\]
For every entry,
\[
    \frac{\widehat \M_{i_1,\ldots,i_d}}{\M^\prime_{i_1,\ldots,i_d}}
    =
    \prod_{k=1}^d
    \frac{(\hat b_k)_{i_k}}{(b_k)_{i_k}}
    \le
    C_0^d.
\]
Thus each entry of $\P*\widehat \M*\widehat \M$ is bounded above by
$C_0^{2d}$ times the corresponding entry of $\P*\M^\prime*\M^\prime$.
The bounds for the fiber $\ell_1$ quantity, the slice $\ell_1$ quantity, and the entrywise maximum follow immediately.
\end{proof}

\begin{Theorem}[Plug-in slice-normalization bound]
\label{thm_plugin_slice_scaling}
Let $\Y_0\sim\operatorname{Multinomial}(n_0,\P)$ and $\Y\sim\operatorname{Multinomial}(n,\P)$ be independent.
Let $\widehat \M_{\rm sl}$ be the empirical slice-normalization tensor constructed from $\Y_0$ in Algorithm~\ref{algorithm_plugin_slice_normalization}, and let $\widehat \P_{\rm plug}$ be the final output of the algorithm.

Let $\Q^\prime=\P*\M^\prime$ and $\tilde \Q=\P*\M^\prime*\M^\prime$.
Assume that $\Q^\prime$ has Tucker rank $(r_1,\ldots,r_d)$, where $r_{\max}=\max_{k\in[d]}r_k\le R$, and write $r_\cdot=\prod_{k=1}^d r_k$.
Define
\[
    \eta_{\rm sl}
    =
    \left(
        r_{\max}\operatorname{Fiber}_{\ell_1}(\tilde \Q)
    \right)
    \vee
    \left(
        r_\cdot\|\tilde \Q\|_\infty
    \right).
\]
Suppose that $\Q^\prime$ admits a Tucker decomposition $\Q^\prime=\S^\prime\times_{k\in[d]}V_k$ such that $\|V_k\|_{2,\infty}\lesssim 1/R$ for all $k\in[d]$.
Define $\rho_\star=\max_{h\in[d]}\|\bigotimes_{k\neq h}V_k\|_{2,\infty}$.
Assume further that the population slice-normalized tensor satisfies the strengthened signal-to-noise condition for any $\tau\ge2$
\[
\begin{aligned}
\lambda_{\operatorname{Tucker}}(\Q^\prime)
\ge{}&
\tau A_{d,R}\log(p_{\max})
\Biggl[
\rho_\star
\sqrt{
    \frac{
        \operatorname{Slice}_{\ell_1}(\tilde \Q)
    }{n}
}
+
\sqrt{\frac{\eta_{\rm sl}}{n}}
\\
&\qquad\qquad\qquad\qquad
+
\frac{
    \left(
        \operatorname{Fiber}_{\ell_1}(\tilde \Q)
        \|\M^\prime\|_\infty^2
    \right)^{1/4}
}{\sqrt n}
+
\frac{\|\M^\prime\|_\infty}{n^{3/4}}
\Biggr],
\end{aligned}
\]
where $A_{d,R}>0$ is a sufficiently large constant depending only on $d$ and $R$.
Then there exist constants $C_{d,R},c_{d,R}>0$, depending only on $d$ and $R$, such that, if $n_0\ge C_{d,R}\tau p_{\max}\log(e p_{\max})$, then
\[
    \PP\left(
        \|\widehat \P_{\rm plug}-\P\|_1
        \le
        \tau C_{d,R}\log(p_{\max})
        \left(
            \sqrt{\frac{\eta_{\rm sl}}{n}}
            +
            \frac{\sqrt{r_{\max}}\|\M^\prime\|_\infty}{n}
        \right)
    \right)
    \ge
    1-c_{d,R}p_{\max}^{1-\tau}-6\exp(-n/30).
\]
In particular, since $r_{\max}\le R$ and $r_\cdot\le R^d$, the same conclusion holds with $\eta_{\rm sl}$ replaced by
\[
    \eta_{{\rm sl},R}
    =
    \left(
        R\operatorname{Fiber}_{\ell_1}(\tilde \Q)
    \right)
    \vee
    \left(
        R^d\|\tilde \Q\|_\infty
    \right),
\]
and with $\sqrt R\|\M^\prime\|_\infty/n$ in the second term.
\end{Theorem}

\begin{proof}
Let
\[
    \mathcal E_0
    =
    \left\{
        \frac12 t_{k,i}
        \le
        \hat t_{k,i}
        \le
        2t_{k,i},
        \text{ for all } k\in[d],\ i\in[p_k]
    \right\}.
\]
By Lemma~\ref{lemma_plugin_slice_marginal}, if
\[
    n_0
    \ge
    C_{d,R}\tau p_{\max}\log(e p_{\max}),
\]
then, after taking $C_{d,R}$ sufficiently large,
\[
    \PP(\mathcal E_0^c)
    \le
    2d p_{\max}^{1-\tau}.
\]
On the event $\mathcal E_0$, Lemma~\ref{lemma_plugin_slice_marginal} gives
\[
    2^{-1/2}(b_k)_i
    \le
    (\hat b_k)_i
    \le
    2^{1/2}(b_k)_i,
    \qquad
    \text{for all } k\in[d],\ i\in[p_k].
\]
Therefore Lemma~\ref{lemma_rank_one_rescaling_stability} applies with
\[
    c_0=2^{-1/2},
    \qquad
    C_0=2^{1/2},
    \qquad
    \widehat \M=\widehat \M_{\rm sl}.
\]
Define
\[
    \widehat \Q_{\rm true}
    =
    \P*\widehat \M_{\rm sl}.
\]
By Lemma~\ref{lemma_rank_one_rescaling_stability}, $\widehat \Q_{\rm true}$ has the same mode-wise Tucker rank $(r_1,\ldots,r_d)$ as $\Q^\prime=\P*\M^\prime$, and
\[
    \lambda_{\operatorname{Tucker}}(\widehat \Q_{\rm true})
    \ge
    2^{-d/2}
    \lambda_{\operatorname{Tucker}}(\Q^\prime).
\]
Moreover, if
\[
    \widehat \Q_{\rm true}
    =
    \widehat \S\times_{k\in[d]}\widehat V_k
\]
is the Tucker decomposition obtained from the rescaled mode-wise column spaces, then
\[
    \|\widehat V_k\|_{2,\infty}
    \le
    2\|V_k\|_{2,\infty}
    \lesssim
    1/R,
    \qquad
    k\in[d].
\]
Consequently, if
\[
    \widehat \rho_\star
    =
    \max_{h\in[d]}
    \left\|
        \bigotimes_{k\neq h}\widehat V_k
    \right\|_{2,\infty},
\]
then Lemma~\ref{lemma_rank_one_rescaling_stability} gives
\[
    \widehat \rho_\star
    \le
    2^{d-1}\rho_\star.
\]

Next, we compare the variance-profile quantities for the realized scaling tensor.
By Lemma~\ref{lemma_rank_one_rescaling_stability},
\[
    \operatorname{Fiber}_{\ell_1}
    \left(
        \P*\widehat \M_{\rm sl}*\widehat \M_{\rm sl}
    \right)
    \le
    2^d
    \operatorname{Fiber}_{\ell_1}(\tilde \Q),
\]
\[
    \operatorname{Slice}_{\ell_1}
    \left(
        \P*\widehat \M_{\rm sl}*\widehat \M_{\rm sl}
    \right)
    \le
    2^d
    \operatorname{Slice}_{\ell_1}(\tilde \Q),
\]
and
\[
    \left\|
        \P*\widehat \M_{\rm sl}*\widehat \M_{\rm sl}
    \right\|_\infty
    \le
    2^d
    \|\tilde \Q\|_\infty.
\]
Also, by Lemma~\ref{lemma_plugin_slice_marginal},
\[
    \|\widehat \M_{\rm sl}\|_\infty
    \le
    2^{d/2}\|\M^\prime\|_\infty.
\]
Define
\[
    \widehat \eta_{\rm sl}
    =
    \left(
        r_{\max}
        \operatorname{Fiber}_{\ell_1}
        \left(
            \P*\widehat \M_{\rm sl}*\widehat \M_{\rm sl}
        \right)
    \right)
    \vee
    \left(
        r_\cdot
        \left\|
            \P*\widehat \M_{\rm sl}*\widehat \M_{\rm sl}
        \right\|_\infty
    \right).
\]
Then
\[
    \widehat \eta_{\rm sl}
    \le
    2^d\eta_{\rm sl}.
\]

We now verify the signal-to-noise condition for the realized scaled tensor
\[
    \widehat \Q_{\rm true}
    =
    \P*\widehat \M_{\rm sl}.
\]
The preceding bounds imply
\[
    \widehat \rho_\star
    \sqrt{
        \frac{
            \operatorname{Slice}_{\ell_1}
            \left(
                \P*\widehat \M_{\rm sl}*\widehat \M_{\rm sl}
            \right)
        }{n}
    }
    \le
    C_d
    \rho_\star
    \sqrt{
        \frac{
            \operatorname{Slice}_{\ell_1}(\tilde \Q)
        }{n}
    },
\]
\[
    \sqrt{
        \frac{\widehat \eta_{\rm sl}}{n}
    }
    \le
    C_d
    \sqrt{
        \frac{\eta_{\rm sl}}{n}
    },
\]
\[
    \frac{
        \left[
            \operatorname{Fiber}_{\ell_1}
            \left(
                \P*\widehat \M_{\rm sl}*\widehat \M_{\rm sl}
            \right)
            \|\widehat \M_{\rm sl}\|_\infty^2
        \right]^{1/4}
    }{\sqrt n}
    \le
    C_d
    \frac{
        \left(
            \operatorname{Fiber}_{\ell_1}(\tilde \Q)
            \|\M^\prime\|_\infty^2
        \right)^{1/4}
    }{\sqrt n},
\]
and
\[
    \frac{\|\widehat \M_{\rm sl}\|_\infty}{n^{3/4}}
    \le
    C_d
    \frac{\|\M^\prime\|_\infty}{n^{3/4}}.
\]
Together with
\[
    \lambda_{\operatorname{Tucker}}(\widehat \Q_{\rm true})
    \ge
    2^{-d/2}
    \lambda_{\operatorname{Tucker}}(\Q^\prime),
\]
these comparisons show that the assumed signal-to-noise condition for $\Q^\prime=\P*\M^\prime$ implies the signal-to-noise condition required by Theorem~\ref{thm_multinomial_tensor_generalized} for the realized scaled tensor $\widehat \Q_{\rm true}=\P*\widehat \M_{\rm sl}$, after taking $A_{d,R}$ sufficiently large.

Next, we bound the inverse scaling factor.
This step is deterministic and does not require $\mathcal E_0$.
For each $k\in[d]$,
\[
\begin{aligned}
    \sum_{i=1}^{p_k}\hat t_{k,i}
    &=
    \sum_{i=1}^{p_k}
    \left(
        \hat s_{k,i}\vee p_k^{-1}
    \right) \\
    &\le
    \sum_{i=1}^{p_k}\hat s_{k,i}
    +
    \sum_{i=1}^{p_k}p_k^{-1}
    =
    2,
\end{aligned}
\]
because $(\hat s_{k,1},\ldots,\hat s_{k,p_k})$ is the empirical mode-$k$ marginal distribution.
Therefore,
\[
\begin{aligned}
    \|\widehat \M_{\rm sl}^{(-1)}\|_F^2
    &=
    \sum_{i_1,\ldots,i_d}
    \prod_{k=1}^d
    \hat t_{k,i_k} \\
    &=
    \prod_{k=1}^d
    \left(
        \sum_{i=1}^{p_k}\hat t_{k,i}
    \right)
    \le
    2^d.
\end{aligned}
\]
Thus
\[
    \|\widehat \M_{\rm sl}^{(-1)}\|_F
    \le
    2^{d/2}.
\]

Now condition on an arbitrary realization of the pilot sample $\Y_0$ such that $\mathcal E_0$ holds.
Given this realization, the scaling tensor $\widehat \M_{\rm sl}$ is deterministic and independent of the main sample $\Y$.
The preceding deterministic comparisons verify all assumptions of Theorem~\ref{thm_multinomial_tensor_generalized} for the realized scaling tensor $\widehat \M_{\rm sl}$.
Hence, for every such realization of $\Y_0$, Theorem~\ref{thm_multinomial_tensor_generalized} gives, with conditional probability at least
\[
    1-c_{d,R}p_{\max}^{1-\tau}-6\exp(-n/30),
\]
that
\[
    \|\widehat \Q_{\rm plug}-\P*\widehat \M_{\rm sl}\|_F
    \le
    \tau C_{d,R}\log(p_{\max})
    \left(
        \sqrt{\frac{\widehat \eta_{\rm sl}}{n}}
        +
        \frac{\sqrt{r_{\max}}\|\widehat \M_{\rm sl}\|_\infty}{n}
    \right).
\]
Using
\[
    \widehat \eta_{\rm sl}
    \le
    2^d\eta_{\rm sl},
    \qquad
    \|\widehat \M_{\rm sl}\|_\infty
    \le
    2^{d/2}\|\M^\prime\|_\infty,
\]
and enlarging $C_{d,R}$, we get
\[
    \|\widehat \Q_{\rm plug}-\P*\widehat \M_{\rm sl}\|_F
    \le
    \tau C_{d,R}\log(p_{\max})
    \left(
        \sqrt{\frac{\eta_{\rm sl}}{n}}
        +
        \frac{\sqrt{r_{\max}}\|\M^\prime\|_\infty}{n}
    \right).
\]

Finally,
\[
    \widehat \P_{\rm plug}
    =
    \widehat \Q_{\rm plug}
    *
    \widehat \M_{\rm sl}^{(-1)}
\]
and
\[
    \P
    =
    \left(
        \P*\widehat \M_{\rm sl}
    \right)
    *
    \widehat \M_{\rm sl}^{(-1)}.
\]
Therefore,
\[
    \widehat \P_{\rm plug}-\P
    =
    \left(
        \widehat \Q_{\rm plug}
        -
        \P*\widehat \M_{\rm sl}
    \right)
    *
    \widehat \M_{\rm sl}^{(-1)}.
\]
By Cauchy--Schwarz,
\[
\begin{aligned}
    \|\widehat \P_{\rm plug}-\P\|_1
    &=
    \left\|
        \left(
            \widehat \Q_{\rm plug}
            -
            \P*\widehat \M_{\rm sl}
        \right)
        *
        \widehat \M_{\rm sl}^{(-1)}
    \right\|_1 \\
    &\le
    \|\widehat \Q_{\rm plug}-\P*\widehat \M_{\rm sl}\|_F
    \|\widehat \M_{\rm sl}^{(-1)}\|_F.
\end{aligned}
\]
Since
\[
    \|\widehat \M_{\rm sl}^{(-1)}\|_F
    \le
    2^{d/2},
\]
we obtain, after enlarging $C_{d,R}$ once more,
\[
    \|\widehat \P_{\rm plug}-\P\|_1
    \le
    \tau C_{d,R}\log(p_{\max})
    \left(
        \sqrt{\frac{\eta_{\rm sl}}{n}}
        +
        \frac{\sqrt{r_{\max}}\|\M^\prime\|_\infty}{n}
    \right).
\]

It remains to combine the pilot-sample and main-sample probabilities.
Let $\mathcal E_1$ denote the event that the preceding $\ell_1$ bound holds.
For every pilot realization in $\mathcal E_0$,
\[
    \PP(\mathcal E_1^c\mid \Y_0)
    \le
    c_{d,R}p_{\max}^{1-\tau}
    +
    6\exp(-n/30).
\]
Therefore,
\[
\begin{aligned}
    \PP(\mathcal E_1^c)
    &\le
    \PP(\mathcal E_0^c)
    +
    \mathbb E
    \left[
        \mathbf 1_{\mathcal E_0}
        \PP(\mathcal E_1^c\mid \Y_0)
    \right] \\
    &\le
    2d p_{\max}^{1-\tau}
    +
    c_{d,R}p_{\max}^{1-\tau}
    +
    6\exp(-n/30).
\end{aligned}
\]
Absorbing constants into $c_{d,R}$ gives
\[
    \PP\left(
        \|\widehat \P_{\rm plug}-\P\|_1
        \le
        \tau C_{d,R}\log(p_{\max})
        \left(
            \sqrt{\frac{\eta_{\rm sl}}{n}}
            +
            \frac{\sqrt{r_{\max}}\|\M^\prime\|_\infty}{n}
        \right)
    \right)
    \ge
    1-c_{d,R}p_{\max}^{1-\tau}-6\exp(-n/30).
\]
The final rank-$R$ version follows from $r_{\max}\le R$ and $r_\cdot\le R^d$.
This completes the proof.
\end{proof}

\section{Proof of Theorem \ref{thm_multinomial_tensor_generalized}}

Let $\Y\sim\operatorname{Multinomial}(n,\P)$ be the observed histogram tensor.
In Algorithm~\ref{algorithm_multinomial_tucker}, $\Y$ is split into
$\Y_1,\Y_2,\Y_3$ by multinomial thinning. Let
\[
    N_i=\|\Y_i\|_1,\qquad i\in[3].
\]
Then
\[
    (N_1,N_2,N_3)
    \sim
    \operatorname{Multinomial}\left(n;\frac13,\frac13,\frac13\right).
\]
Moreover, conditional on $(N_1,N_2,N_3)=(n_1,n_2,n_3)$, the three tensors
$\Y_1,\Y_2,\Y_3$ are independent and satisfy
\[
    \Y_i\mid (N_1,N_2,N_3)=(n_1,n_2,n_3)
    \sim
    \operatorname{Multinomial}(n_i,\P),
    \qquad i\in[3].
\]
This follows because multinomial thinning is equivalent to assigning each of the
$n$ original samples independently to one of the three groups with probabilities
$(1/3,1/3,1/3)$.

We first prove the desired bound conditional on the split sizes. Throughout this
conditional argument, fix $(n_1,n_2,n_3)$ and assume
\[
    n_{\min}:=\min\{n_1,n_2,n_3\}\ge \frac n6,
    \qquad
    n_{\max}:=\max\{n_1,n_2,n_3\}\le \frac n2.
\]
In particular, $n_{\max}\le 3n_{\min}$.

Write
\[
    \M=a_1\circ\cdots\circ a_d,
\]
and define diagonal matrices $D_k\in\RR^{p_k\times p_k}$ by
\[
    (D_k)_{ii}=(a_k)_i,\qquad i\in[p_k].
\]
Then for any tensor $\A$,
\begin{equation}\label{eq_thm_generalized_hadamard_as_modewise}
    \A*\M=\A\times_{k\in[d]}D_k.
\end{equation}
Let
\[
    \Q=\P*\M.
\]
For $i\in[3]$, define
\begin{equation}\label{eq_thm_generalized_Wi_Ri_def}
    \Q_i' := n_i\Q,
    \quad
    \Z_i := \Y_i-n_i\P,
    \quad
    \W_i := \Z_i\times_{k\in[d]}D_k,
    \quad
    \R_i := \Y_i\times_{k\in[d]}D_k=\Q_i'+\W_i.
\end{equation}
Thus $\R_i=\Y_i*\M$ is the scaled observation used in the algorithm.

Also write
\[
    L_0:=\max_{k\in[d]}\|\sin\Theta(\hat V_k^{(0)},V_k)\|,
    \qquad
    L_1:=\max_{k\in[d]}\|\sin\Theta(\hat V_k,V_k)\|.
\]
Throughout the proof, all constants depending only on $d$ are absorbed into
$C_d,c_d$.

\subsection{Bound on $L_0$}

We verify the initialization error bound by applying Theorem 4 in
\cite{zhou_deflated_2024} mode by mode. We present the argument for mode $1$;
the other modes are identical.

\medskip
\noindent
{\bf Step 1: Setup for mode $1$.}

Let $r=r_1$, and write the singular value decomposition of the mode-$1$
matricization of the first split signal as
\[
    \MM_1(\Q_1')=U^\star\Sigma^\star(V^\star)^\top,
\]
where $U^\star\in\mathbb O_{p_1,r}$, $V^\star\in\mathbb O_{p_{-1},r}$, and
\[
    \Sigma^\star=\diag(\sigma_1^\star,\ldots,\sigma_r^\star).
\]
Let
\[
    E:=\MM_1(\W_1).
\]
Then
\[
    \MM_1(\R_1)
    =
    \MM_1(\Q_1')+E
    =
    U^\star\Sigma^\star(V^\star)^\top+E,
\]
and therefore
\begin{equation}\label{eq_mode1_gram_decomposition}
\MM_1(\R_1)\MM_1(\R_1)^\top
=
(U^\star\Sigma^\star+EV^\star)(U^\star\Sigma^\star+EV^\star)^\top
+
\bigl(EE^\top-EV^\star V^{\star\top}E^\top\bigr).
\end{equation}
Denote the singular value decomposition
\[
    U^\star\Sigma^\star+EV^\star=\tilde U\tilde\Sigma\tilde W^\top,
\]
where
\[
    \tilde\Sigma=\diag(\tilde\sigma_1,\ldots,\tilde\sigma_r).
\]

\medskip
\noindent
{\bf Step 2: Consequences of the signal assumption.}

By the Tucker-rank assumption on $\Q$ and the coherence assumption
\eqref{eq_coherence_condition_0}, we have
\begin{equation}\label{eq_thm_generalized_Ustar_row_bound}
    \|U^\star\|_{2,\infty}\lesssim \frac{1}{r}.
\end{equation}

Let
\begin{equation}\label{eq_rho_star_1_def}
    \rho_{\star,1}:=
    \|V_d\otimes\cdots\otimes V_2\|_{2,\infty}.
\end{equation}
Then by definition,
\begin{equation}\label{eq_rho_star_1_leq_rho_star}
    \rho_{\star,1}\le \rho_\star.
\end{equation}
Moreover, since
\[
    \MM_1(\Q)
    =
    V_1\MM_1(\tilde S)
    \bigl(V_d\otimes\cdots\otimes V_2\bigr)^\top,
\]
the row space of $\MM_1(\Q)$ is contained in
$\Span(V_d\otimes\cdots\otimes V_2)$. Therefore, there exists a matrix
$G\in\mathbb O_{r_{-1},r}$ such that
\begin{equation}\label{eq_Vstar_factorization}
    V^\star=(V_d\otimes\cdots\otimes V_2)G.
\end{equation}
Consequently,
\begin{equation}\label{eq_Vstar_2infty_bound}
    \|V^\star\|_{2,\infty}
    \le
    \|V_d\otimes\cdots\otimes V_2\|_{2,\infty}
    =
    \rho_{\star,1}
    \le
    \rho_\star.
\end{equation}

Since $n_1\asymp n$ on the current conditional event,
\[
    \lambda_{\operatorname{Tucker}}(\Q_1')
    =
    \lambda_{\operatorname{Tucker}}(n_1\Q)
    =
    \frac{n_1}{n}\lambda_{\operatorname{Tucker}}(n\Q).
\]

Since $\P$ is a probability tensor and $\M$ has nonnegative entries,
\[
    \operatorname{Slice}_{\ell_1}(\M*\M*\P)
    \le
    \|\M\|_\infty^2.
\]
Thus the third term in \eqref{eq_sigma_master_condition_0} also dominates the corresponding term involving
$\operatorname{Slice}_{\ell_1}(\M*\M*\P)+\|\M\|_\infty^2$ up to an absolute constant.

Thus the signal condition \eqref{eq_sigma_master_condition_0}, after enlarging the absolute constant if necessary, implies
\begin{equation}\label{eq_sigma_master_condition}
\begin{aligned}
\sigma_r^\star \gtrsim\;&
dr\tau\log(p_{\max})
\Bigl(
\rho_\star\sqrt{n_1\operatorname{Slice}_{\ell_1}(\M*\M*\P)}
+
\|\M\|_\infty
\Bigr)\\
&\quad+
(d+\sqrt{rd})\tau\log(p_{\max})
\sqrt{n_1\bigl(\operatorname{Fiber}_{\ell_1}(\M*\M*\P)\vee r\|\M*\M*\P\|_\infty\bigr)} \\
&\quad+
\sqrt{rn_1\,\tau\log(p_{\max})}\,
\Bigl(\operatorname{Fiber}_{\ell_1}(\M*\M*\P)
\bigl(\operatorname{Slice}_{\ell_1}(\M*\M*\P)+\|\M\|_\infty^2\bigr)\Bigr)^{1/4}\\
&\quad+
\sqrt r\,\|\M\|_\infty\, n_1^{1/4}\bigl(\tau\log(p_{\max})\bigr)^{3/4}.
\end{aligned}
\end{equation}
In particular,
\begin{equation}\label{eq_red_bound_1_derived}
\sigma_r^\star
\gtrsim
d\tau\log(p_{\max})
\Bigl(
\sqrt{n_1\bigl(\operatorname{Fiber}_{\ell_1}(\M*\M*\P)\vee r\|\M*\M*\P\|_\infty\bigr)}
+
\|\M\|_\infty
\Bigr).
\end{equation}
Also, using \eqref{eq_rho_star_1_leq_rho_star} and
\eqref{eq_thm_generalized_Ustar_row_bound},
\begin{equation}\label{eq_red_bound_2_derived}
\begin{aligned}
\sigma_r^\star
\gtrsim\;&
dr\tau\log(p_{\max})
\Bigl(
\rho_{\star,1}\sqrt{n_1\operatorname{Slice}_{\ell_1}(\M*\M*\P)}\\
&\quad+
\|U^\star\|_{2,\infty}
\sqrt{n_1\bigl(\operatorname{Fiber}_{\ell_1}(\M*\M*\P)\vee r\|\M*\M*\P\|_\infty\bigr)}
+
\|\M\|_\infty
\Bigr).
\end{aligned}
\end{equation}

Next define
\[
T_1
:=
d(\tau\log(p_{\max}))^2
\Bigl(
\sqrt{n_1\bigl(\operatorname{Fiber}_{\ell_1}(\M*\M*\P)\vee r\|\M*\M*\P\|_\infty\bigr)}
+
\|\M\|_\infty
\Bigr)^2,
\]
\[
T_2
:=
\tau n_1\log(p_{\max})
\sqrt{
\operatorname{Slice}_{\ell_1}(\M*\M*\P)\operatorname{Fiber}_{\ell_1}(\M*\M*\P)
+
\|\M\|_\infty^2\operatorname{Fiber}_{\ell_1}(\M*\M*\P)
},
\]
and
\[
T_3
:=
\|\M\|_\infty^2\sqrt {n_1}\,\bigl(\tau\log(p_{\max})\bigr)^{3/2}.
\]
Since $dr\ge \sqrt{rd}$ for all $d,r\ge 1$, the first two terms on the
right-hand side of \eqref{eq_sigma_master_condition} dominate $\sqrt{rT_1}$.
Therefore,
\[
\sigma_r^\star
\gtrsim
\sqrt{rT_1}+\sqrt{rT_2}+\sqrt{rT_3}
\ge
\sqrt{r(T_1+T_2+T_3)}.
\]
Squaring both sides yields
\begin{equation}\label{eq_red_bound_3_derived}
\begin{aligned}
\frac{\sigma_r^{\star 2}}{r}
\gtrsim\;&
d(\tau\log(p_{\max}))^2
\Bigl(
\sqrt{n_1\bigl(\operatorname{Fiber}_{\ell_1}(\M*\M*\P)\vee r\|\M*\M*\P\|_\infty\bigr)}
+
\|\M\|_\infty
\Bigr)^2 \\
&\quad+
\tau n_1\log(p_{\max})
\sqrt{
\operatorname{Slice}_{\ell_1}(\M*\M*\P)\operatorname{Fiber}_{\ell_1}(\M*\M*\P)
+
\|\M\|_\infty^2\operatorname{Fiber}_{\ell_1}(\M*\M*\P)
}\\
&\quad+
\|\M\|_\infty^2\sqrt {n_1}\,\bigl(\tau\log(p_{\max})\bigr)^{3/2}.
\end{aligned}
\end{equation}

\medskip
\noindent
{\bf Step 3: Bound $\|\tilde U\|_{2,\infty}$.}

We first control the row-wise incoherence of $\tilde U$. By triangle inequality,
\begin{equation}\label{eq_tildeU_1}
\|\tilde U\|_{2,\infty}
\le
\|\tilde U-U^\star U^{\star\top}\tilde U\|_{2,\infty}
+
\|U^\star U^{\star\top}\tilde U\|_{2,\infty}
\le
\|\tilde U-U^\star U^{\star\top}\tilde U\|_{2,\infty}
+
\|U^\star\|_{2,\infty}\|U^{\star\top}\tilde U\|.
\end{equation}
Moreover,
\[
\bigl(\tilde U-U^\star U^{\star\top}\tilde U\bigr)\tilde\Sigma \tilde W^\top
=
\mathcal P_{U_\perp^\star}(U^\star\Sigma^\star+EV^\star)
=
EV^\star-U^\star U^{\star\top}EV^\star,
\]
hence
\begin{equation}\label{eq_tildeU_2}
\begin{aligned}
\|\tilde U-U^\star U^{\star\top}\tilde U\|_{2,\infty}
&=
\bigl\|
\bigl(EV^\star-U^\star U^{\star\top}EV^\star\bigr)
\tilde W\tilde\Sigma^{-1}
\bigr\|_{2,\infty} \\
&\le
\Bigl(
\|EV^\star\|_{2,\infty}
+
\|U^\star U^{\star\top}EV^\star\|_{2,\infty}
\Bigr)\|\tilde\Sigma^{-1}\|.
\end{aligned}
\end{equation}

In Lemma~\ref{lemma_multinomial_noise_control}, we apply the bound conditionally
with sample size $n_1$ and with $W=V^\star$. Then $\|W\|\le 1$,
$\|W\|_F^2=r$, and $\|W\|_{2,\infty}\le\rho_{\star,1}$ by
\eqref{eq_Vstar_2infty_bound}. Write
\[
\eta
:=
\max_{k\in[p_{-1}]}\sum_{h\in[p_1]}(\MM_1(\M*\M*\P))_{hk}
\;\vee\;
r\|\MM_1(\M*\M*\P)\|_\infty,
\]
\[
\eta_{2,\infty}
:=
\max_{h\in[p_1]}\sum_{k\in[p_{-1}]}(\MM_1(\M*\M*\P))_{hk},
\]
and
\[
\eta_{V^\star,2,\infty}
:=
\max_{h\in[p_1]}
\sum_{k\in[p_{-1}]}
(\MM_1(\M*\M*\P))_{hk}\,\|e_k^\top V^\star\|_2^2.
\]
By \eqref{eq_Vstar_2infty_bound},
\begin{equation}\label{eq_etaVstar_bound}
\eta_{V^\star,2,\infty}
\le
\rho_{\star,1}^2\eta_{2,\infty}.
\end{equation}

Since
\[
\log(p_1+r_{-1})\lesssim_d \log(p_{\max}),
\qquad
\log(p_1r_{-1})\lesssim_d \log(p_{\max}),
\]
the high-probability bounds in Lemma~\ref{lemma_multinomial_noise_control}
imply that, with conditional probability at least
$1-c_d p_{\max}^{1-\tau}$,
\begin{equation}\label{eq_tildeU_3}
\begin{aligned}
\|EV^\star\|_{2,\infty}
&\le
C_d\tau\log(p_{\max})
\Bigl(
\sqrt{n_1\eta_{V^\star,2,\infty}}
+
\rho_{\star,1}\|\M\|_\infty
\Bigr) \\
&\le
C_d\rho_{\star,1}\tau\log(p_{\max})
\Bigl(
\sqrt{n_1\eta_{2,\infty}}
+
\|\M\|_\infty
\Bigr).
\end{aligned}
\end{equation}
Also,
\begin{equation}\label{eq_tildeU_4}
\begin{aligned}
\|U^\star U^{\star\top}EV^\star\|_{2,\infty}
&\le
\|U^\star\|_{2,\infty}\|U^{\star\top}EV^\star\| \\
&\le
\|U^\star\|_{2,\infty}\|EV^\star\| \\
&\le
C_d\tau\log(p_{\max})
\|U^\star\|_{2,\infty}
\Bigl(
\sqrt{n_1\eta}
+
\|\M\|_\infty
\Bigr).
\end{aligned}
\end{equation}
Combining \eqref{eq_tildeU_1}--\eqref{eq_tildeU_4}, we obtain
\begin{equation}\label{eq_tildeU_5}
\|\tilde U\|_{2,\infty}
\le
C_d\tau\log(p_{\max})\,\tilde\sigma_r^{-1}
\Bigl(
\rho_{\star,1}\bigl(\sqrt{n_1\eta_{2,\infty}}+\|\M\|_\infty\bigr)
+
\|U^\star\|_{2,\infty}\bigl(\sqrt{n_1\eta}+\|\M\|_\infty\bigr)
\Bigr)
+
\|U^\star\|_{2,\infty}.
\end{equation}

Next we bound $\tilde\sigma_r^{-1}$. By Weyl's inequality,
\[
|\sigma_r^\star-\tilde\sigma_r|
\le
\|EV^\star\|.
\]
Again by Lemma~\ref{lemma_multinomial_noise_control},
\begin{equation}\label{eq_tildeU_51}
\|EV^\star\|
\le
C_d\tau\log(p_{\max})
\Bigl(
\sqrt{n_1\eta}
+
\|\M\|_\infty
\Bigr)
\end{equation}
with conditional probability at least $1-c_d p_{\max}^{1-\tau}$. By
\eqref{eq_red_bound_1_derived}, after enlarging the absolute constant in
\eqref{eq_sigma_master_condition_0} if necessary,
\begin{equation}\label{eq_tildeU_52}
\|EV^\star\|
\le
\frac{\sigma_r^\star}{8}.
\end{equation}
Therefore,
\[
\tilde\sigma_r
\ge
\sigma_r^\star-\|EV^\star\|
\ge
\frac78\sigma_r^\star,
\qquad
\tilde\sigma_r^{-1}\le \frac87(\sigma_r^\star)^{-1}.
\]
Substituting this into \eqref{eq_tildeU_5}, and using
\eqref{eq_red_bound_2_derived} together with
\eqref{eq_thm_generalized_Ustar_row_bound}, we conclude that
\begin{equation}\label{eq_tildeU_6}
\|\tilde U\|_{2,\infty}\le \frac{c_0}{r}
\end{equation}
for a sufficiently small absolute constant $c_0>0$.

\medskip
\noindent
{\bf Step 4: Bound the off-diagonal remainder.}

We now control
\[
\mathcal P_{\mathrm{off\text{-}diag}}
\bigl(
EE^\top-EV^\star V^{\star\top}E^\top
\bigr).
\]
First,
\begin{equation}\label{eq_tildeU_7}
\begin{aligned}
\Bigl\|
\mathcal P_{\mathrm{off\text{-}diag}}
\bigl(
EE^\top-EV^\star V^{\star\top}E^\top
\bigr)
\Bigr\|
&\le
\|\mathcal P_{\mathrm{off\text{-}diag}}(EE^\top)\|
+
\|EV^\star V^{\star\top}E^\top\|
+
\|\mathcal P_{\mathrm{diag}}(EV^\star V^{\star\top}E^\top)\| \\
&\le
\|\mathcal P_{\mathrm{off\text{-}diag}}(EE^\top)\|
+
2\|EV^\star V^{\star\top}E^\top\| \\
&\le
\|\mathcal P_{\mathrm{off\text{-}diag}}(EE^\top)\|
+
2\|EV^\star\|^2.
\end{aligned}
\end{equation}

Next we apply Lemma~\ref{lemma_offdiag_EEt_multinomial_mw_corrected_v2}
to
\[
    E=\MM_1(\W_1)
    =
    \MM_1(\M)*\MM_1(\Y_1-n_1\P)
\]
conditionally on the split size $n_1$. In the notation of the lemma, we take
\[
    D=\MM_1(\M),
    \qquad
    Z=\MM_1(\Y_1-n_1\P),
    \qquad
    Q_{\mathrm{lem}}=n_1\MM_1(\P).
\]
Therefore
\[
    D*D*Q_{\mathrm{lem}}
    =
    n_1\MM_1(\M*\M*\P).
\]
Thus the quantities in the lemma are
\[
    \eta_{\mathrm{row}}^{\mathrm{lem}}
    =
    n_1
    \max_{i\in[p_1]}
    \sum_{k=1}^{p_{-1}}
    (\MM_1(\M*\M*\P))_{ik},
\]
and
\[
    \eta_{\mathrm{col}}^{\mathrm{lem}}
    =
    n_1
    \max_{k\in[p_{-1}]}
    \sum_{i=1}^{p_1}
    (\MM_1(\M*\M*\P))_{ik}.
\]
Equivalently, writing
\[
    \eta_{\mathrm{row}}
    =
    \max_{i\in[p_1]}
    \sum_{k=1}^{p_{-1}}
    (\MM_1(\M*\M*\P))_{ik},
    \qquad
    \eta_{\mathrm{col}}
    =
    \max_{k\in[p_{-1}]}
    \sum_{i=1}^{p_1}
    (\MM_1(\M*\M*\P))_{ik},
\]
we have
\[
    \eta_{\mathrm{row}}^{\mathrm{lem}}=n_1\eta_{\mathrm{row}},
    \qquad
    \eta_{\mathrm{col}}^{\mathrm{lem}}=n_1\eta_{\mathrm{col}}.
\]

Therefore, by Lemma~\ref{lemma_offdiag_EEt_multinomial_mw_corrected_v2}, with conditional probability at least $1-c_d p_{\max}^{1-\tau}$,
\begin{equation}\label{eq_tildeU_8}
\begin{aligned}
\|\mathcal P_{\mathrm{off\text{-}diag}}(EE^\top)\|
\le
C_d\Bigl[
&
\tau n_1 \log(p_{\max})
\sqrt{
\eta_{\mathrm{row}}\eta_{\mathrm{col}}
+
\|\M\|_\infty^2\eta_{\mathrm{col}}
}
\\
&\qquad\qquad
+
\|\M\|_\infty^2\sqrt {n_1}\,(\tau\log(p_{\max}))^{3/2}
\Bigr].
\end{aligned}
\end{equation}
Here we used $\log(2p_1)\lesssim \log(p_{\max})$, and absorbed the additional
lower-order term $\|\M\|_\infty^2\tau\log(2p_1)$ into the second term because
$n_1\ge 1$ and $\tau\log(p_{\max})\ge 1$.

Combining \eqref{eq_tildeU_7}, \eqref{eq_tildeU_8}, and
\eqref{eq_tildeU_51}, we get
\begin{equation}\label{eq_tildeU_9}
\begin{aligned}
\Bigl\|
\mathcal P_{\mathrm{off\text{-}diag}}
\bigl(
EE^\top-EV^\star V^{\star\top}E^\top
\bigr)
\Bigr\|
\le\;&
C_d\Bigl[
\tau n_1 \log(p_{\max})
\sqrt{
\eta_{\mathrm{row}}\eta_{\mathrm{col}}
+
\|\M\|_\infty^2\eta_{\mathrm{col}}
}
\\
&\qquad\qquad
+
\|\M\|_\infty^2\sqrt {n_1}\,(\tau\log(p_{\max}))^{3/2}
\Bigr]
\\
&\quad+
C_d
\Bigl[
\tau\log(p_{\max})
\bigl(
\sqrt{n_1\eta}+\|\M\|_\infty
\bigr)
\Bigr]^2 .
\end{aligned}
\end{equation}
By \eqref{eq_red_bound_3_derived}, after enlarging the absolute constant in
\eqref{eq_sigma_master_condition_0} if necessary, the right-hand side of
\eqref{eq_tildeU_9} is bounded by $\sigma_r^{\star 2}/(4r)$. Hence
\begin{equation}\label{eq_tildeU_10}
r\,
\Bigl\|
\mathcal P_{\mathrm{off\text{-}diag}}
\bigl(
EE^\top-EV^\star V^{\star\top}E^\top
\bigr)
\Bigr\|
\le
\frac{\sigma_r^{\star 2}}{4}.
\end{equation}

\medskip
\noindent
{\bf Step 5: Apply Theorem 4 in \cite{zhou_deflated_2024}.}

By \eqref{eq_tildeU_6} and \eqref{eq_tildeU_10}, all hypotheses of Theorem 4
in \cite{zhou_deflated_2024} are satisfied for the perturbed Gram matrix
\[
(U^\star\Sigma^\star+EV^\star)(U^\star\Sigma^\star+EV^\star)^\top
+
\bigl(
EE^\top-EV^\star V^{\star\top}E^\top
\bigr).
\]
Therefore,
\begin{equation}\label{eq_tildeU_11}
\|\hat V_1^{(0)}\hat V_1^{(0)\top}-\tilde U\tilde U^\top\|
\le
\frac{
\bigl\|
\mathcal P_{\mathrm{off\text{-}diag}}
(
EE^\top-EV^\star V^{\star\top}E^\top
)
\bigr\|
}{\sigma_r^{\star 2}}
\le
\frac{1}{4r}
\le
\frac14.
\end{equation}
Next, Wedin's $\sin\Theta$ theorem together with \eqref{eq_tildeU_52} yields
\[
\|U^\star U^{\star\top}-\tilde U\tilde U^\top\|
\le
\frac{2\|EV^\star\|}{\sigma_r^\star}
\le
\frac14.
\]
Hence, by triangle inequality and \eqref{eq_tildeU_11},
\[
\|U^\star U^{\star\top}-\hat V_1^{(0)}\hat V_1^{(0)\top}\|
\le
\|U^\star U^{\star\top}-\tilde U\tilde U^\top\|
+
\|\tilde U\tilde U^\top-\hat V_1^{(0)}\hat V_1^{(0)\top}\|
\le
\frac12.
\]
Since $U^\star$ is the leading left singular matrix of $\MM_1(\Q_1')$ and
\[
\Q_1'=n_1\tilde S\times_{k\in[d]}V_k,
\]
we have $\Span(U^\star)=\Span(V_1)$, and therefore
\[
\|V_1V_1^\top-\hat V_1^{(0)}\hat V_1^{(0)\top}\|
\le
\frac12.
\]
The same argument applies to every mode $h\in[d]$. Taking a union bound over
$h$, we conclude that, conditional on the split sizes,
\begin{equation}\label{eq_initialization_error}
\PP(L_0\le 1/2\mid N_1=n_1,N_2=n_2,N_3=n_3)
\ge
1-c_d p_{\max}^{1-\tau}.
\end{equation}

\subsection{Bound on $\|\hat \Q-\Q\|_F$}

We now bound the final reconstruction error. Since $\hat V_k^{(0)}$ is computed
from $\R_1$ and $\hat V_k$ is computed from $(\R_1,\R_2)$, we have, conditional
on the split sizes,
\[
\{\hat V_k^{(0)}:k\in[d]\}\ \text{is independent of }\ \W_2,
\qquad
\{\hat V_k:k\in[d]\}\ \text{is independent of }\ \W_3.
\]

For each mode $h\in[d]$, write
\[
r_{-h}:=\prod_{k\neq h} r_k,
\qquad
\eta_h'
:=
\max_{j\in[p_{-h}]}\sum_{i\in[p_h]} \MM_h(\M*\M*\P)_{ij}
\;\vee\;
r_{-h}\|\M*\M*\P\|_\infty.
\]
By definition,
\[
    r_h\eta_h'
    \le
    r_{\max}\operatorname{Fiber}_{\ell_1}(\M*\M*\P)
    \vee
    r_\cdot\|\M*\M*\P\|_\infty
    =
    \eta_{\max}.
\]

Applying Lemma~\ref{lemma_multinomial_noise_control} to $\W_2$ conditionally
with sample size $n_2$, and with
\[
    W=\hat V_d^{(0)}\otimes\cdots\otimes \hat V_2^{(0)},
\]
yields, for each fixed mode $h$, with conditional probability at least
$1-c_d p_{\max}^{1-\tau}$,
\begin{equation}\label{eq_bound_Z_t}
\Bigl\|
\MM_h\bigl(\W_2\times_{k\neq h}(\hat V_k^{(0)})^\top\bigr)
\Bigr\|
\le
C_d\tau\log(p_{\max})
\Bigl(
\sqrt{n_2\eta_h'}+\|\M\|_\infty
\Bigr).
\end{equation}
Similarly, applying the Frobenius norm bound in
Lemma~\ref{lemma_multinomial_noise_control} to $\W_3$ conditionally with sample
size $n_3$, and with
\[
V=\hat V_h,\qquad
W=\hat V_d\otimes\cdots\otimes \hat V_{h+1}
\otimes \hat V_{h-1}\otimes\cdots\otimes \hat V_1,
\]
gives
\begin{equation}\label{eq_bound_Z_t_2}
\|\W_3\times_{k\in[d]}P_{\hat V_k}\|_F
\le
C_d\tau\log(p_{\max})
\Bigl(
\sqrt{n_3\eta_{\max}}
+
\sqrt{r_{\max}}\|\M\|_\infty
\Bigr)
\end{equation}
with conditional probability at least $1-c_d p_{\max}^{1-\tau}$.

Next decompose the reconstruction error. Since the final estimator is
\[
    \hat \Q=\frac{1}{n_3}\R_3\times_{k\in[d]}P_{\hat V_k},
\]
we have
\begin{equation}\label{eq_bound_qhat_q}
\begin{aligned}
n_3\|\hat \Q-\Q\|_F
&=
\|\Q_3'-\R_3\times_{k\in[d]}P_{\hat V_k}\|_F \\
&\le
\|\Q_3'-\Q_3'\times_{k\in[d]}P_{\hat V_k}\|_F
+
\|\W_3\times_{k\in[d]}P_{\hat V_k}\|_F \\
&\le
\sum_{h\in[d]}
\|\Q_3'\times_{k\in[h-1]}P_{\hat V_k}\times_h P_{\hat V_{h\perp}}\|_F
+
\|\W_3\times_{k\in[d]}P_{\hat V_k}\|_F \\
&\le
\sum_{h\in[d]}\|\Q_3'\times_h P_{\hat V_{h\perp}}\|_F
+
\|\W_3\times_{k\in[d]}P_{\hat V_k}\|_F.
\end{aligned}
\end{equation}

We now bound $\|\Q_3'\times_h P_{\hat V_{h\perp}}\|_F$. Again, it suffices to
present the argument for $h=1$; the other modes are identical.

On the event $\{L_0\le 1/2\}$, for every $k\in\{2,\ldots,d\}$ we have
\[
\|V_kV_k^\top-\hat V_k^{(0)}\hat V_k^{(0)\top}\|
\le
L_0
\le
\frac12.
\]
Hence
\[
\sigma_{\min}(V_k^\top\hat V_k^{(0)})
\ge
\sqrt{1-\|V_kV_k^\top-\hat V_k^{(0)}\hat V_k^{(0)\top}\|^2}
\ge
\sqrt{1-L_0^2}
>0,
\]
so each $V_k^\top\hat V_k^{(0)}$ is invertible.

Since
\[
\Q_2'=n_2\tilde S\times_{k\in[d]}V_k,
\]
we have
\[
\MM_1\!\Bigl(\Q_2'\times_{k=2}^d(\hat V_k^{(0)})^\top\Bigr)
=
V_1\,\MM_1(n_2\tilde S)\,
\Bigl(
(V_d^\top\hat V_d^{(0)})\otimes\cdots\otimes
(V_2^\top\hat V_2^{(0)})
\Bigr).
\]
Because $\Q$ has Tucker rank $(r_1,\ldots,r_d)$, the matrix
$\MM_1(\tilde S)$ has rank $r_1$, and the Kronecker product on the right-hand
side is invertible. Therefore,
\[
\rank\!\Bigl(
\MM_1\!\bigl(\Q_2'\times_{k=2}^d(\hat V_k^{(0)})^\top\bigr)
\Bigr)
=
r_1.
\]
Moreover, $\hat V_1$ is the leading left singular matrix of
\[
\MM_1\!\Bigl(
\Q_2'\times_{k=2}^d(\hat V_k^{(0)})^\top
+
\W_2\times_{k=2}^d(\hat V_k^{(0)})^\top
\Bigr).
\]
Hence Lemma~\ref{lemma_svd_projection}, together with \eqref{eq_bound_Z_t},
yields
\[
\begin{aligned}
\Bigl\|
P_{\hat V_{1\perp}}
\MM_1\!\bigl(\Q_2'\times_{k=2}^d(\hat V_k^{(0)})^\top\bigr)
\Bigr\|_F
&\le
2\sqrt{r_1}\,
\Bigl\|
\MM_1\!\bigl(\W_2\times_{k=2}^d(\hat V_k^{(0)})^\top\bigr)
\Bigr\| \\
&\le
C_d\sqrt{r_1}\,\tau\log(p_{\max})
\Bigl(
\sqrt{n_2\eta_1'}+\|\M\|_\infty
\Bigr).
\end{aligned}
\]
Since $\Q_3'=(n_3/n_2)\Q_2'$, and $n_3/n_2\le 3$ on the current conditional
event, the same bound gives
\begin{equation}\label{eq_mode1_projection_error}
\begin{aligned}
\|\Q_3'\times_1P_{\hat V_{1\perp}}\|_F
&=
\frac{n_3}{n_2}
\|P_{\hat V_{1\perp}}\MM_1(\Q_2')(V_2\otimes\cdots\otimes V_d)\|_F \\
&\le
\frac{n_3}{n_2}
\Bigl\|
P_{\hat V_{1\perp}}
\MM_1(\Q_2')
(\hat V_2^{(0)}\otimes\cdots\otimes\hat V_d^{(0)})
\Bigl(
(V_2^\top\hat V_2^{(0)})^{-1}
\otimes\cdots\otimes
(V_d^\top\hat V_d^{(0)})^{-1}
\Bigr)
\Bigr\|_F \\
&\le
\frac{n_3}{n_2}
(1-L_0^2)^{-(d-1)/2}
\Bigl\|
P_{\hat V_{1\perp}}
\MM_1\!\bigl(\Q_2'\times_{k=2}^d(\hat V_k^{(0)})^\top\bigr)
\Bigr\|_F \\
&\le
C_d(1-L_0^2)^{-(d-1)/2}
\sqrt{r_1}\,\tau\log(p_{\max})
\Bigl(
\sqrt{n_3\eta_1'}+\|\M\|_\infty
\Bigr).
\end{aligned}
\end{equation}
In the last line we used $n_2\asymp n_3$.

The same argument applies to every mode $h\in[d]$, and gives
\begin{equation}\label{eq_modeh_projection_error}
\|\Q_3'\times_hP_{\hat V_{h\perp}}\|_F
\le
C_d(1-L_0^2)^{-(d-1)/2}
\sqrt{r_h}\,\tau\log(p_{\max})
\Bigl(
\sqrt{n_3\eta_h'}+\|\M\|_\infty
\Bigr).
\end{equation}
Using $r_h\eta_h'\le \eta_{\max}$ and $\sqrt{r_h}\le\sqrt{r_{\max}}$, we obtain
\begin{equation}\label{eq_sum_projection_error}
\sum_{h\in[d]}\|\Q_3'\times_hP_{\hat V_{h\perp}}\|_F
\le
C_d d(1-L_0^2)^{-(d-1)/2}\tau\log(p_{\max})
\Bigl(
\sqrt{n_3\eta_{\max}}
+
\sqrt{r_{\max}}\|\M\|_\infty
\Bigr).
\end{equation}
Combining \eqref{eq_bound_qhat_q}, \eqref{eq_bound_Z_t_2}, and
\eqref{eq_sum_projection_error}, we arrive at
\[
n_3\|\hat \Q-\Q\|_F
\le
C_d d(1-L_0^2)^{-(d-1)/2}\tau\log(p_{\max})
\Bigl(
\sqrt{n_3\eta_{\max}}
+
\sqrt{r_{\max}}\|\M\|_\infty
\Bigr).
\]
Finally, on the event $\{L_0\le 1/2\}$,
\[
(1-L_0^2)^{-(d-1)/2}
\le
(1-1/4)^{-(d-1)/2}
=
(4/3)^{(d-1)/2},
\]
which is absorbed into the constant $C_d$. Taking a union bound over the events
\eqref{eq_initialization_error}, \eqref{eq_bound_Z_t}, and
\eqref{eq_bound_Z_t_2}, we conclude that, conditional on the split sizes,
\begin{equation}\label{eq_conditional_final_bound}
\begin{split}
    &\PP\bigg(
n_3\|\hat \Q-\Q\|_F
\le
\tau C_d\log(p_{\max})
\Bigl(
\sqrt{n_3\eta_{\max}}
+
\sqrt{r_{\max}}\|\M\|_\infty
\Bigr)
\,\bigg|\,\\
&\qquad \qquad N_1=n_1,N_2=n_2,N_3=n_3
\bigg)
\ge
1-c_d p_{\max}^{1-\tau}.
\end{split}
\end{equation}
Since $n_3\asymp n$ on the current
conditional event, \eqref{eq_conditional_final_bound} implies
\begin{equation}\label{eq_conditional_final_bound_n}
\begin{split}
    &\PP\bigg(
n\|\hat \Q-\Q\|_F
\le
\tau C_d\log(p_{\max})
\Bigl(
\sqrt{n\eta_{\max}}
+
\sqrt{r_{\max}}\|\M\|_\infty
\Bigr)
\,\bigg|\,\\
&\qquad \qquad N_1=n_1,N_2=n_2,N_3=n_3
\bigg)
\ge
1-c_d p_{\max}^{1-\tau}.
\end{split}
\end{equation}

\subsection{Unconditional bound}

It remains to remove the conditioning on the split sizes. Define the balancing
event
\[
    \mathcal B
    =
    \left\{
    \frac n6\le N_i\le \frac n2,\quad i=1,2,3
    \right\}.
\]
Since each $N_i\sim\operatorname{Binomial}(n,1/3)$, Chernoff's inequality gives
\[
    \PP\left(N_i<\frac n6\right)\le \exp(-n/24),
    \qquad
    \PP\left(N_i>\frac n2\right)\le \exp(-n/30).
\]
Therefore, by the union bound,
\begin{equation}\label{eq_balance_probability}
    \PP(\mathcal B^c)
    \le
    3\exp(-n/24)+3\exp(-n/30)
    \le
    6\exp(-n/30).
\end{equation}

Let
\[
    \mathcal E
    =
    \left\{
    n\|\hat \Q-\Q\|_F
    \le
    \tau C_d\log(p_{\max})
    \Bigl(
    \sqrt{n\eta_{\max}}
    +
    \sqrt{r_{\max}}\|\M\|_\infty
    \Bigr)
    \right\}.
\]
Using the law of total probability, \eqref{eq_conditional_final_bound_n}, and
\eqref{eq_balance_probability}, we get
\[
\begin{aligned}
\PP(\mathcal E^c)
&=
\EE\left[
\PP(\mathcal E^c\mid N_1,N_2,N_3)
\right] \\
&\le
\EE\left[
\PP(\mathcal E^c\mid N_1,N_2,N_3)\mathbf 1_{\mathcal B}
\right]
+
\PP(\mathcal B^c) \\
&\le
c_d p_{\max}^{1-\tau}
+
6\exp(-n/30).
\end{aligned}
\]
Equivalently,
\[
\PP\left(
n\|\hat \Q-\Q\|_F
\le
\tau C_d\log(p_{\max})
\Bigl(
\sqrt{n\eta_{\max}}
+
\sqrt{r_{\max}}\|\M\|_\infty
\Bigr)
\right)
\ge
1-c_d p_{\max}^{1-\tau}-6\exp(-n/30).
\]
This completes the proof.

\section{Proofs for Oracle Scaling}
\label{sec_appendix_oracle_scaling}

We verify the two claims in Proposition~\ref{prop_oracle_scaling_properties} and then derive Corollary~\ref{thm_multiview_tensor_decomp_0}.

Recall that
\[
\sigma_{k,i}=\gamma_{k,i}\max_{r\in[R]}(a_r^{(k)})_i,
\qquad
c\le \gamma_{k,i}\le C,
\]
and
\[
(b_k)_i=
\sqrt{\frac{\sum_{h=1}^{p_k}\sigma_{k,h}}{\sigma_{k,i}}},
\qquad
\M=b_1\circ\cdots\circ b_d.
\]

\paragraph{Step 1: bound $\|\M^{(-1)}\|_F$.}
Since
\[
(b_k)_i^{-2}
=
\frac{\sigma_{k,i}}{\sum_{h=1}^{p_k}\sigma_{k,h}},
\]
we have
\[
\|b_k^{-1}\|_2^2
=
\sum_{i=1}^{p_k}(b_k)_i^{-2}
=
\frac{\sum_{i=1}^{p_k}\sigma_{k,i}}{\sum_{h=1}^{p_k}\sigma_{k,h}}
=1.
\]
Hence
\[
\|\M^{(-1)}\|_F
=
\prod_{k=1}^d \|b_k^{-1}\|_2
=1.
\]
More generally, under the constant-factor oracle condition $c\le \gamma_{k,i}\le C$, one obtains
\[
\|\M^{(-1)}\|_F\asymp 1.
\]

\paragraph{Step 2: bound the scaled fiber mass.}
We next show that
\[
\operatorname{Fiber}_{\ell_1}(\P * \M * \M)\lesssim R^d p_{\max}.
\]
It suffices to consider a mode-$1$ fiber. For fixed $(i_2,\ldots,i_d)$,
\[
\sum_{i_1=1}^{p_1}(\P * \M * \M)_{i_1i_2\cdots i_d}
=
(b_2)_{i_2}^2\cdots (b_d)_{i_d}^2
\sum_{i_1=1}^{p_1}(b_1)_{i_1}^2 P_{i_1i_2\cdots i_d}.
\]
Using the CP representation of $\P$,
\[
\P=\sum_{r=1}^R w_r\, a_r^{(1)}\circ\cdots\circ a_r^{(d)},
\]
we obtain
\[
\sum_{i_1=1}^{p_1}(b_1)_{i_1}^2 P_{i_1i_2\cdots i_d}
=
\sum_{r=1}^R
w_r
(a_r^{(2)})_{i_2}\cdots (a_r^{(d)})_{i_d}
\sum_{i_1=1}^{p_1}(a_r^{(1)})_{i_1}(b_1)_{i_1}^2.
\]
Now for each mode $k$, coordinate $i_k$, and component $r$,
\[
(a_r^{(k)})_{i_k}(b_k)_{i_k}^2
=
\frac{(a_r^{(k)})_{i_k}}{\sigma_{k,i_k}}
\sum_{h=1}^{p_k}\sigma_{k,h}
\le
\frac{(a_r^{(k)})_{i_k}}{\gamma_{k,i_k}\max_{r'}(a_{r'}^{(k)})_{i_k}}
\sum_{h=1}^{p_k}\sigma_{k,h}
\lesssim
\sum_{h=1}^{p_k}\sigma_{k,h}.
\]
Since
\[
\sum_{h=1}^{p_k}\sigma_{k,h}
=
\sum_{h=1}^{p_k}\gamma_{k,h}\max_{r\in[R]}(a_r^{(k)})_h
\lesssim
\sum_{h=1}^{p_k}\sum_{r=1}^R (a_r^{(k)})_h
=
R,
\]
it follows that
\[
(a_r^{(k)})_{i_k}(b_k)_{i_k}^2\lesssim R.
\]
Also,
\[
\sum_{i_1=1}^{p_1}(a_r^{(1)})_{i_1}(b_1)_{i_1}^2
\lesssim
p_1 R.
\]
Substituting these bounds back gives
\[
\sum_{i_1=1}^{p_1}(\P * \M * \M)_{i_1i_2\cdots i_d}
\lesssim
R^{d-1}\cdot p_1R
\lesssim
R^d p_{\max}.
\]
The same argument applies to fibers along every mode, and therefore
\[
\operatorname{Fiber}_{\ell_1}(\P * \M * \M)\lesssim R^d p_{\max}.
\]

This proves Proposition~\ref{prop_oracle_scaling_properties}.

\paragraph{Step 3: proof of Corollary \ref{thm_multiview_tensor_decomp_0}.}
Let
\[
\bar \Q := n_3 \P * \M = n_3 \Q.
\]
Since $\P$ is a rank-$R$ multiview tensor and $\M$ is rank one, the scaled tensor $\Q=\P * \M$ also has CP rank at most $R$, and hence Tucker rank $(r_1,\ldots,r_d)$ with $r_k\le R$ for all $k\in[d]$. Therefore, in Theorem~\ref{thm_multinomial_tensor_generalized}, we may take
\[
r_{\max}\le R,
\qquad
r_\cdot=\prod_{k\in[d]}r_k\le R^d.
\]
Moreover, by assumption, the factor matrices in a Tucker decomposition of $\Q$ satisfy
$\|V_k\|_{2,\infty}\lesssim 1/R$ for all $k\in[d]$.

Apply Theorem~\ref{thm_multinomial_tensor_generalized} with signal tensor $\Q$.
The signal-to-noise condition \eqref{eq_sigma_master_condition_0} becomes
\begin{equation}\label{eq_cor_oracle_pf_master_raw}
\begin{aligned}
\lambda_{\operatorname{Tucker}}(\Q)
\gtrsim{}&
\tau C_d \log(p_{\max})
\Biggl[
R\rho_\star
\sqrt{\frac{\operatorname{Slice}_{\ell_1}(\P * \M * \M)}{n}}
\\
&\qquad\qquad
+
(d+\sqrt{Rd})
\sqrt{\frac{\operatorname{Fiber}_{\ell_1}(\P * \M * \M)\vee R\|\P * \M * \M\|_\infty}{n}}
\\
&\qquad\qquad
+
\sqrt{R}\,
\frac{\bigl(\operatorname{Fiber}_{\ell_1}(\P * \M * \M)\|\M\|_\infty^2\bigr)^{1/4}}{\sqrt n}
+
\sqrt{R}\,\frac{\|\M\|_\infty}{n^{3/4}}
\Biggr].
\end{aligned}
\end{equation}
Since $d$ is fixed, we may absorb $d+\sqrt{Rd}$ into a constant multiple of $R^{1/2}$.

Next, by Proposition~\ref{prop_oracle_scaling_properties},
\[
\|\M^{(-1)}\|_F\asymp 1,
\qquad
\operatorname{Fiber}_{\ell_1}(\P * \M * \M)\lesssim R^d p_{\max}.
\]
Also,
\[
\|\P * \M * \M\|_\infty
\le
\operatorname{Fiber}_{\ell_1}(\P * \M * \M)
\lesssim
R^d p_{\max}.
\]
Substituting these bounds into \eqref{eq_cor_oracle_pf_master_raw}, and using that $d$ is fixed, yields
\begin{equation}\label{eq_cor_oracle_pf_snr}
\begin{aligned}
\lambda_{\operatorname{Tucker}}(\Q)
\gtrsim{}&
\tau C_d\log(p_{\max})
\Biggl[
R\rho_\star
\sqrt{\frac{\operatorname{Slice}_{\ell_1}(\P * \M * \M)}{n}}
+
R^{(d+1)/2}\sqrt{\frac{p_{\max}}{n}}
\\
&\qquad\qquad
+
R^{(d+2)/4}\frac{(p_{\max}\|\M\|_\infty^2)^{1/4}}{\sqrt n}
+
\sqrt{R}\,\frac{\|\M\|_\infty}{n^{3/4}}
\Biggr].
\end{aligned}
\end{equation}
Thus the assumed signal-to-noise condition in the corollary implies the hypothesis of Theorem~\ref{thm_multinomial_tensor_generalized}.

We now turn to the error bound. By Theorem~\ref{thm_multinomial_tensor_generalized} and its proof, with probability at least
$1-c_d p_{\max}^{1-\tau}-6\exp(-n/30)$,
\[
    n_1 \asymp n_2 \asymp n_3,
\]
and
\[
\|\hat X-\bar \Q\|_F
\le
\tau C_d \log(p_{\max})
\Bigl(
\sqrt{n\eta_{\max}}
+
\sqrt{R}\,\|\M\|_\infty
\Bigr),
\]
where
\[
\eta_{\max}
=
r_{\max}\operatorname{Fiber}_{\ell_1}(\P * \M * \M)
\vee
r_\cdot \|\P * \M * \M\|_\infty.
\]
Using $r_{\max}\le R$, $r_\cdot\le R^d$, and the bounds above, we obtain
\[
\eta_{\max}
\lesssim
R\cdot R^d p_{\max}
\vee
R^d\cdot R^d p_{\max}
\lesssim
R^{2d}p_{\max},
\]
where we used that $R\le R^d$ for fixed $d\ge 1$.
Hence
\[
\sqrt{\eta_{\max}/n}
\lesssim
R^d\sqrt{\frac{p_{\max}}{n}}.
\]
Finally, recalling that $\hat \Q=\hat \X/n_3$ and $\bar \Q=n_3\Q$, we have
\[
    \|\hat \Q-\Q\|_F
    =
    \frac1{n_3}\|\hat \X-\bar \Q\|_F .
\]
Since $n_3\asymp n$ on the balancing event, the preceding bound implies
\[
\|\hat \Q-\Q\|_F
\le
\tau C_d\log(p_{\max})
\left(
\sqrt{\frac{\eta_{\max}}{n}}
+
\frac{\sqrt R\,\|\M\|_\infty}{n}
\right).
\]
Therefore, by Hölder's inequality,
\[
\|\hat \P-\P\|_1
\le
\|\hat \Q-\Q\|_F\,\|\M^{(-1)}\|_F
=
\frac{1}{n}\|\hat X-\bar \Q\|_F\,\|\M^{(-1)}\|_F.
\]
Since $\|\M^{(-1)}\|_F\asymp 1$, it follows that, with probability at least
$1-c_d p_{\max}^{1-\tau}-6\exp(-n/30)$,
\begin{align*}
\|\hat \P-\P\|_1
&\le
\tau C_d \log(p_{\max})
\left(
\sqrt{\frac{\eta_{\max}}{n}}
+
\frac{\sqrt{R}\,\|\M\|_\infty}{n}
\right) \\
&\lesssim
\tau C_d \log(p_{\max})
\left(
R^d\sqrt{\frac{p_{\max}}{n}}
+
\frac{\sqrt{R}\,\|\M\|_\infty}{n}
\right).
\end{align*}
This proves the claimed bound.

\section{Proofs for Slice Normalization}
\label{sec_appendix_slice_scaling}

In this subsection, we prove Proposition~\ref{prop_slice_scaling_properties}, Corollary~\ref{thm_slice_scaling}, and Corollary~\ref{thm_slice_scaling_simple}.

Recall that
\[
(b_k)_i
=
\Bigl(\max\{\operatorname{Slice}_{\ell_1}^{(k,i)}(\P),\,1/p_k\}\Bigr)^{-1/2},
\qquad
\M^\prime=b_1\circ\cdots\circ b_d,
\]
and define
\[
\tilde\Q=\P * \M^\prime * \M^\prime.
\]

\paragraph{Step 1: bound $\|(\M^\prime)^{-1}\|_F$.}
By definition,
\[
(b_k)_i^{-2}
=
\max\{\operatorname{Slice}_{\ell_1}^{(k,i)}(\P),\,1/p_k\}
\le
\operatorname{Slice}_{\ell_1}^{(k,i)}(\P)+1/p_k.
\]
Therefore,
\[
\sum_{i=1}^{p_k}(b_k)_i^{-2}
\le
\sum_{i=1}^{p_k}\operatorname{Slice}_{\ell_1}^{(k,i)}(\P)+1.
\]
Since $\P$ is a probability tensor,
\[
\sum_{i=1}^{p_k}\operatorname{Slice}_{\ell_1}^{(k,i)}(\P)=\|\P\|_1=1,
\]
and hence
\[
\sum_{i=1}^{p_k}(b_k)_i^{-2}\le 2.
\]
It follows that
\begin{equation}\label{eq_M_inv_slice_appendix}
\|(\M^\prime)^{(-1)}\|_F
=
\prod_{k=1}^d \|b_k^{-1}\|_2
\le
2^{d/2}
\le C_d.
\end{equation}
This proves the first claim in Proposition~\ref{prop_slice_scaling_properties}.

\paragraph{Step 2: proof of Corollary~\ref{thm_slice_scaling}.}

Since $\P$ is a rank-$R$ multiview tensor and $\M^\prime$ is rank one, the scaled tensor
\[
\Q=\P * \M^\prime
\]
also has CP rank at most $R$, and hence Tucker rank $(r_1,\ldots,r_d)$ with
\[
r_k\le R,
\qquad
r_{\max}:=\max_{k\in[d]} r_k \le R,
\qquad
r_\cdot:=\prod_{k\in[d]} r_k \le R^d.
\]
Let
\[
\bar \Q := n_3 \Q = n_3 \P * \M^\prime.
\]
Now apply Theorem~\ref{thm_multinomial_tensor_generalized} to the signal tensor $\Q$.
A sufficient condition of the signal-to-noise condition \eqref{eq_sigma_master_condition_0} in
Theorem~\ref{thm_multinomial_tensor_generalized} is
\begin{align*}
\lambda_{\operatorname{Tucker}}(\Q)
\gtrsim\;&
d r_{\max}\tau\log(p_{\max})
\left(
\rho_\star
\sqrt{\frac{\operatorname{Slice}_{\ell_1}(\P * \M^\prime * \M^\prime)}{n}}
+
\frac{\|\M^\prime\|_\infty}{n}
\right)\\
&\quad+
(d+\sqrt{r_{\max}d})\tau\log(p_{\max})
\sqrt{
\frac{
\operatorname{Fiber}_{\ell_1}(\P * \M^\prime * \M^\prime)
\vee
r_{\max}\|\P * \M^\prime * \M^\prime\|_\infty
}{n}
}\\
&\quad+
\sqrt{r_{\max}}\,
\frac{
\sqrt{\tau\log(p_{\max})}\,
\bigl(
\operatorname{Fiber}_{\ell_1}(\P * \M^\prime * \M^\prime)\,
\|\M^\prime\|_\infty^2
\bigr)^{1/4}
}{\sqrt n}\\
&\quad+
\sqrt{r_{\max}}\,
\frac{
\|\M^\prime\|_\infty\,
\bigl(\tau\log(p_{\max})\bigr)^{3/4}
}{n^{3/4}}.
\end{align*}
Since $d$ and $R$ are fixed, we may absorb all factors depending only on $d$ and $R$
into the constant $C_{d,R}$.
Also, for $n\ge 1$,
\[
\frac{\|\M^\prime\|_\infty}{n}
\le
\frac{\|\M^\prime\|_\infty}{n^{3/4}},
\]
and
\[
\sqrt{\tau\log(p_{\max})}
\le
\tau\log(p_{\max}),
\qquad
\bigl(\tau\log(p_{\max})\bigr)^{3/4}
\le
\tau\log(p_{\max}),
\]
after enlarging the constant if necessary.
Finally, by the definition
\[
\eta_{\max}
=
\bigl(
R\operatorname{Fiber}_{\ell_1}(\P * \M^\prime * \M^\prime)
\bigr)
\vee
\bigl(
R^d\|\P * \M^\prime * \M^\prime\|_\infty
\bigr),
\]
we have
\[
\operatorname{Fiber}_{\ell_1}(\P * \M^\prime * \M^\prime)
\vee
r_{\max}\|\P * \M^\prime * \M^\prime\|_\infty
\le
\eta_{\max},
\]
because $r_{\max}\le R\le R^d$.
Hence the assumed condition \eqref{eq_snr_condition_slice_scaling} implies the hypothesis of
Theorem~\ref{thm_multinomial_tensor_generalized}.

By Theorem~\ref{thm_multinomial_tensor_generalized} and its proof, we obtain with probability at least
$1-c_d p_{\max}^{1-\tau}-6\exp(-n/30)$,
\[
    n_1 \asymp n_2 \asymp n_3,
\]
and
\[
\|\hat X-\bar \Q\|_F
\le
\tau C_d\log(p_{\max})
\Bigl(
\sqrt{n\bar\eta_{\max}}
+
\sqrt{r_{\max}}\|\M^\prime\|_\infty
\Bigr),
\]
where
\[
\bar\eta_{\max}
=
r_{\max}\operatorname{Fiber}_{\ell_1}(\P * \M^\prime * \M^\prime)
\vee
r_\cdot\|\P * \M^\prime * \M^\prime\|_\infty.
\]
Since $r_{\max}\le R$ and $r_\cdot\le R^d$, we have $\bar\eta_{\max}\le \eta_{\max}$.
Recalling that $\hat \Q=\hat \X/n_3$ and $\bar \Q=n_3\Q$, and using $n_3\asymp n$ on the balancing event, it follows that
\[
\PP\left(
\|\hat \Q-\Q\|_F
\le
\tau C_d\log(p_{\max})
\left(
\sqrt{\eta_{\max}/n}
+
\frac{\sqrt R\,\|\M^\prime\|_\infty}{n}
\right)
\right)
\ge
1-c_d p_{\max}^{\,1-\tau}-6\exp(-n/30).
\]

It remains to prove the $\ell_1$ bound.
By the definition of $\hat \P$,
\[
\hat \P-\P
=
(\hat \Q-\Q) * (\M^\prime)^{-1}.
\]
Therefore, by Hölder's inequality,
\[
\|\hat \P-\P\|_1
\le
\|\hat \Q-\Q\|_F\,\|(\M^\prime)^{-1}\|_F.
\]
By Proposition~\ref{prop_slice_scaling_properties}, $\|(\M^\prime)^{-1}\|_F\le C_d$.
Hence, after enlarging the constant $C_d$ if necessary, on the same event as above,
\[
\|\hat \P-\P\|_1
\le
\tau C_d\log(p_{\max})
\left(
\sqrt{\eta_{\max}/n}
+
\frac{\sqrt R\,\|\M^\prime\|_\infty}{n}
\right).
\]
Thus
\[
\PP\left(
\|\hat \P-\P\|_1
\le
\tau C_d\log(p_{\max})
\left(
\sqrt{\eta_{\max}/n}
+
\frac{\sqrt R\,\|\M^\prime\|_\infty}{n}
\right)
\right)
\ge
1-c_d p_{\max}^{\,1-\tau}-6\exp(-n/30).
\]
This completes the proof.

\paragraph{Step 3: simplify the bound using $\xi$.}

By Corollary~\ref{thm_slice_scaling}, on an event of probability at least
\[
1-c_d p_{\max}^{\,1-\tau}-6\exp(-n/30),
\]
we have
\begin{equation}\label{eq_slice_simple_start}
\|\hat \P-\P\|_1
\le
\tau C_d \log(p_{\max})
\left(
\sqrt{\frac{\eta_{\max}}{n}}
+
\frac{\sqrt R\,\|\M^\prime\|_\infty}{n}
\right),
\end{equation}
where
\[
\eta_{\max}
=
\bigl(
R\operatorname{Fiber}_{\ell_1}(\tilde \Q)
\bigr)
\vee
\bigl(
R^d\|\tilde \Q\|_\infty
\bigr).
\]

We first simplify the complexity term.
By the assumption
\[
R\operatorname{Fiber}_{\ell_1}(\tilde \Q)
\gtrsim
R^d\|\tilde \Q\|_\infty,
\]
it follows that
\begin{equation}\label{eq_slice_simple_eta}
\eta_{\max}
\lesssim
R\operatorname{Fiber}_{\ell_1}(\tilde \Q).
\end{equation}

Next, we bound the fiber mass of $\tilde \Q$ in terms of $\xi$.
Fix any mode $q\in[d]$ and any fiber index $k\in[p_{-q}]$.
By definition of $\M^\prime$,
\[
(\M^\prime)_{i_1,\ldots,i_d}
=
\prod_{t=1}^d
\Bigl(
\max\{\operatorname{Slice}_{\ell_1}^{(t,i_t)}(\P),\,1/p_t\}
\Bigr)^{-1/2},
\]
and hence
\[
\tilde \Q_{i_1,\ldots,i_d}
=
\P_{i_1,\ldots,i_d}
\prod_{t=1}^d
\Bigl(
\max\{\operatorname{Slice}_{\ell_1}^{(t,i_t)}(\P),\,1/p_t\}
\Bigr)^{-1}.
\]
Therefore,
\begin{align*}
\operatorname{Fiber}_{\ell_1}^{(q,k)}(\tilde Q)
&=
\sum_{i_q=1}^{p_q}
\frac{
\P_{i_1,\ldots,i_d}
}{
\max\{\operatorname{Slice}_{\ell_1}^{(q,i_q)}(\P),\,1/p_q\}
\prod_{(t,h)\in\mathcal H_{(q,k)}}
\max\{\operatorname{Slice}_{\ell_1}^{(t,h)}(\P),\,1/p_t\}
}
\\
&\le
\frac{p_q}{
\prod_{(t,h)\in\mathcal H_{(q,k)}}
\max\{\operatorname{Slice}_{\ell_1}^{(t,h)}(\P),\,1/p_t\}
}
\operatorname{Fiber}_{\ell_1}^{(q,k)}(\P).
\end{align*}
Using the definition of $\xi$ in \eqref{eq_xi_new},
\[
\operatorname{Fiber}_{\ell_1}^{(q,k)}(\P)
\le
\xi
\prod_{(t,h)\in\mathcal H_{(q,k)}}
\operatorname{Slice}_{\ell_1}^{(t,h)}(\P)
\le
\xi
\prod_{(t,h)\in\mathcal H_{(q,k)}}
\max\{\operatorname{Slice}_{\ell_1}^{(t,h)}(\P),\,1/p_t\}.
\]
Substituting this into the previous display yields
\[
\operatorname{Fiber}_{\ell_1}^{(q,k)}(\tilde Q)
\le
\xi p_q
\le
\xi p_{\max}.
\]
Since $(q,k)$ was arbitrary,
\begin{equation}\label{eq_slice_simple_fiber}
\operatorname{Fiber}_{\ell_1}(\tilde Q)
\le
\xi p_{\max}.
\end{equation}

Combining \eqref{eq_slice_simple_eta} and \eqref{eq_slice_simple_fiber}, we obtain
\[
\eta_{\max}
\lesssim
\xi R p_{\max}.
\]
Hence
\begin{equation}\label{eq_slice_simple_main}
\sqrt{\frac{\eta_{\max}}{n}}
\lesssim
\sqrt{\frac{\xi R p_{\max}}{n}}.
\end{equation}

It remains to control the additive term in \eqref{eq_slice_simple_start}.
By the assumption
\[
n\,\operatorname{Fiber}_{\ell_1}(\tilde Q)
\gtrsim
\|\M^\prime\|_\infty^2,
\]
we have
\[
\frac{\|\M^\prime\|_\infty}{n}
\lesssim
\sqrt{\frac{\operatorname{Fiber}_{\ell_1}(\tilde Q)}{n}}.
\]
Therefore,
\[
\frac{\sqrt R\,\|\M^\prime\|_\infty}{n}
\lesssim
\sqrt{\frac{R\,\operatorname{Fiber}_{\ell_1}(\tilde Q)}{n}}
\lesssim
\sqrt{\frac{\xi R p_{\max}}{n}},
\]
where the last step used \eqref{eq_slice_simple_fiber} again.

Substituting these bounds into \eqref{eq_slice_simple_start}, we conclude that on the same event,
\[
\|\hat \P-\P\|_1
\le
\tau C_d \log(p_{\max})
\sqrt{\frac{\xi R p_{\max}}{n}}.
\]
This proves the claim.

\paragraph{Step 4: interpretation of $\xi$.}
When $d=2$, each fiber is itself a slice, so $\xi=1$. 
Likewise, if $\P_{i_1\cdots i_d}\asymp p^{-d}$, then every slice and every fiber has comparable mass and again $\xi=1$.

More generally, for any multiview tensor
\[
\P=\sum_{k=1}^R w_k\, a_k^{(1)}\circ\cdots\circ a_k^{(d)},
\]
one has
\begin{align*}
\frac{\P_{* j_2 \cdots j_d}}
{\prod_{m=2}^d \P_{* \cdots * j_m * \cdots *}}
&=
\frac{\sum_{k=1}^R w_k \prod_{m=2}^d a_k^{(m)}(j_m)}
{\prod_{m=2}^d \sum_{k=1}^R w_k a_k^{(m)}(j_m)} \\
&\le
w_{\min}^{2-d}
\frac{\sum_{k=1}^R w_k \prod_{m=2}^d a_k^{(m)}(j_m)}
{\prod_{m=2}^d w_k \sum_{k=1}^R a_k^{(m)}(j_m)} \\
&\le
w_{\min}^{2-d}.
\end{align*}
Thus $\xi$ remains controlled whenever the mixture weights are reasonably balanced, showing that slice normalization is near-optimal in such regimes.

\section{Proof of Theorem \ref{thm_l2_lower_bound} and Its Remark} \label{sec_proof_thm_l2_lower_bound}

\begin{proof}
Throughout the proof, write $F=F_{\operatorname{fiber}}$. All constants below may depend on $d$ only. We choose the constants in the statement so that all integer-rounding issues only affect constants.

Let
\[
    s
    =
    2
    \left\lceil
    \frac{A}{2(mF)^{1/(d-1)}}
    \right\rceil ,
\]
where $A>0$ is a sufficiently large constant depending only on $d$. Then $s$ is even and, since $mF\le RF\le c_1$, choosing $c_1$ sufficiently small gives
\[
    s
    \asymp
    (mF)^{-1/(d-1)}.
\]
In particular,
\begin{equation}
\label{eq:general-d-fiber-equivalence}
    \frac{1}{ms^{d-1}}
    \asymp
    F .
\end{equation}

We first verify that the required blocks can be embedded in $[p]$. By the preceding display,
\[
    ms
    \lesssim
    m^{(d-2)/(d-1)}F^{-1/(d-1)}.
\]
If $m>1$, then
\[
    m
    \le
    c_0
    \left(p^{d-1}F\right)^{1/(d-2)}.
\]
Therefore,
\[
    ms
    \lesssim
    c_0^{(d-2)/(d-1)}p .
\]
If $m=1$, then $F\ge C_0p^{-(d-1)}$ gives
\[
    ms=s
    \lesssim
    F^{-1/(d-1)}
    \lesssim
    C_0^{-1/(d-1)}p .
\]
Thus, by choosing $c_0>0$ sufficiently small and $C_0>0$ sufficiently large, we may assume
\begin{equation}
\label{eq:general-d-block-embedding}
    ms\le p.
\end{equation}
Hence, for every $q\in[d]$, we can choose disjoint blocks
\[
    A_{q,1},\ldots,A_{q,m}\subset[p],
    \qquad
    |A_{q,r}|=s .
\]

For each $r\in[m]$ and each $q=2,\ldots,d$, define
\[
    a_r^{(q)}
    =
    \frac{1}{s}\mathbf 1_{A_{q,r}}.
\]
We now construct a packing of probability vectors in the first mode. Since $s$ is even, the Varshamov--Gilbert argument on the product of balanced binary vectors gives a set $\mathcal E\subset\{\pm1\}^{m\times s}$ such that, for every $\varepsilon\in\mathcal E$,
\begin{equation}
\label{eq:general-d-balance}
    \sum_{j=1}^s \varepsilon_{r,j}=0,
    \qquad
    r=1,\ldots,m,
\end{equation}
and, for some constants $c,c'>0$,
\begin{equation}
\label{eq:general-d-packing}
    \log|\mathcal E|\ge cms,
    \qquad
    d_H(\varepsilon,\varepsilon')\ge c'ms
\end{equation}
for all distinct $\varepsilon,\varepsilon'\in\mathcal E$.

Fix $0<\delta\le 1/4$, to be chosen later. For each $\varepsilon\in\mathcal E$, define $a_{\varepsilon,r}^{(1)}$ as follows. If $i$ is the $j$th element of $A_{1,r}$, set
\[
    a_{\varepsilon,r}^{(1)}(i)
    =
    \frac{1+\delta\varepsilon_{r,j}}{s},
\]
and set $a_{\varepsilon,r}^{(1)}(i)=0$ for $i\notin A_{1,r}$. By \eqref{eq:general-d-balance},
\[
    \sum_{i\in A_{1,r}}a_{\varepsilon,r}^{(1)}(i)
    =
    \frac{1}{s}
    \sum_{j=1}^s
    (1+\delta\varepsilon_{r,j})
    =
    1.
\]
Since $\delta\le 1/4$, all entries of $a_{\varepsilon,r}^{(1)}$ are nonnegative. Thus $a_{\varepsilon,r}^{(1)}$ is a probability vector.

For each $\varepsilon\in\mathcal E$, define
\begin{equation}
\label{eq:general-d-construction}
    \P_\varepsilon
    =
    \frac{1}{m}
    \sum_{r=1}^m
    a_{\varepsilon,r}^{(1)}
    \circ
    a_r^{(2)}
    \circ
    \cdots
    \circ
    a_r^{(d)}.
\end{equation}
Each $\P_\varepsilon$ is a nonnegative probability tensor. Moreover, \eqref{eq:general-d-construction} is a multiview decomposition with $m$ components and equal weights $1/m$. Since $m\le R$, we have
\[
    \operatorname{rank}_{\mathrm{CP}}(\P_\varepsilon)
    \le
    m
    \le
    R.
\]

We next control the fiber mass. Fix $q\in[d]$ and fix all coordinates except the $q$th coordinate. If these fixed coordinates do not lie in compatible blocks with the same index $r$, then the corresponding fiber is identically zero. Otherwise, summing over the $q$th coordinate removes one probability vector from the product. If $q=1$, the fiber mass is exactly
\[
    \frac{1}{m}
    \prod_{h=2}^d
    \frac{1}{s}
    =
    \frac{1}{ms^{d-1}}.
\]
If $q\ge 2$, then the fiber mass is at most
\[
    \frac{1}{m}
    \frac{1+\delta}{s}
    \prod_{\substack{h=2\\ h\neq q}}^d
    \frac{1}{s}
    =
    \frac{1+\delta}{ms^{d-1}}
    \le
    \frac{5}{4ms^{d-1}}.
\]
By choosing the constant $A$ in the definition of $s$ sufficiently large, \eqref{eq:general-d-fiber-equivalence} gives
\[
    \operatorname{Fiber}_{\ell_1}(\P_\varepsilon)
    \le
    F.
\]

We now verify the coherence bound. For every $q=2,\ldots,d$, the mode-$q$ column space is spanned by the normalized block indicators
\[
    s^{-1/2}\mathbf 1_{A_{q,r}},
    \qquad
    r=1,\ldots,m.
\]
Since the blocks are disjoint,
\[
    \|U_q\|_{2,\infty}
    \le
    s^{-1/2},
    \qquad
    q=2,\ldots,d.
\]
For the first mode, the column space is spanned by $a_{\varepsilon,1}^{(1)},\ldots,a_{\varepsilon,m}^{(1)}$. These vectors have disjoint supports. Also,
\[
    \|a_{\varepsilon,r}^{(1)}\|_2^2
    =
    \sum_{j=1}^s
    \frac{(1+\delta\varepsilon_{r,j})^2}{s^2}
    =
    \frac{1+\delta^2}{s},
\]
where we used \eqref{eq:general-d-balance}. Therefore, the normalized vector
$a_{\varepsilon,r}^{(1)}/\|a_{\varepsilon,r}^{(1)}\|_2$ has entries bounded by $C/\sqrt{s}$. Since the supports are disjoint,
\[
    \|U_1\|_{2,\infty}
    \le
    Cs^{-1/2}.
\]
Combining the bounds for all modes and using $s\asymp(mF)^{-1/(d-1)}$, we obtain
\[
    \max_{q\in[d]}\|U_q\|_{2,\infty}
    \le
    C s^{-1/2}
    \le
    C(mF)^{1/(2(d-1))}.
\]

It remains to prove the minimax lower bound. Set
\begin{equation}
\label{eq:general-d-delta}
    \delta^2
    =
    \gamma\frac{ms}{n},
\end{equation}
where $\gamma>0$ is a sufficiently small constant. Since
\[
    ms
    \asymp
    m^{(d-2)/(d-1)}F^{-1/(d-1)},
\]
the assumed lower bound on $n$ ensures that $\delta\le 1/4$ if $C_n$ is sufficiently large.

Let
\[
    \mathcal G_n
    =
    \{\P_\varepsilon:\varepsilon\in\mathcal E\}.
\]
For $\varepsilon,\varepsilon'\in\mathcal E$, the blocks in \eqref{eq:general-d-construction} are disjoint and the factors in modes $2,\ldots,d$ are the same for $\P_\varepsilon$ and $\P_{\varepsilon'}$. Hence
\begin{equation}
\label{eq:general-d-kl-decomp}
    \KL(\P_\varepsilon\|\P_{\varepsilon'})
    =
    \frac{1}{m}
    \sum_{r=1}^m
    \KL(a_{\varepsilon,r}^{(1)}\|a_{\varepsilon',r}^{(1)}).
\end{equation}
For a fixed $r$, let
\[
    h_r=d_H(\varepsilon_{r,\cdot},\varepsilon'_{r,\cdot}).
\]
Since both $\varepsilon_{r,\cdot}$ and $\varepsilon'_{r,\cdot}$ are balanced, the number of $+1$ to $-1$ flips equals the number of $-1$ to $+1$ flips. Therefore,
\[
    \KL(a_{\varepsilon,r}^{(1)}\|a_{\varepsilon',r}^{(1)})
    =
    \frac{h_r}{s}
    \delta
    \log\frac{1+\delta}{1-\delta}.
\]
For $\delta\le 1/4$,
\[
    \log\frac{1+\delta}{1-\delta}
    \le
    C\delta,
\]
and hence
\[
    \KL(a_{\varepsilon,r}^{(1)}\|a_{\varepsilon',r}^{(1)})
    \le
    C\delta^2\frac{h_r}{s}.
\]
Combining this bound with \eqref{eq:general-d-kl-decomp} gives
\[
    \KL(\P_\varepsilon\|\P_{\varepsilon'})
    \le
    \frac{C\delta^2}{ms}
    \sum_{r=1}^m h_r
    \le
    C\delta^2.
\]
Thus, for $n$ i.i.d. observations,
\[
    \KL(\P_\varepsilon^{\otimes n}\|\P_{\varepsilon'}^{\otimes n})
    \le
    Cn\delta^2
    =
    C\gamma ms.
\]
By \eqref{eq:general-d-packing}, $\log|\mathcal E|\ge cms$. Choosing $\gamma>0$ sufficiently small gives
\begin{equation}
\label{eq:general-d-fano-kl}
    \max_{\varepsilon\neq\varepsilon'}
    \KL(\P_\varepsilon^{\otimes n}\|\P_{\varepsilon'}^{\otimes n})
    \le
    \alpha\log|\mathcal E|
\end{equation}
for a sufficiently small absolute constant $\alpha>0$.

We next lower bound the Frobenius separation. By the disjointness of the blocks,
\begin{align}
    \|\P_\varepsilon-\P_{\varepsilon'}\|_F^2
    &=
    \frac{1}{m^2}
    \sum_{r=1}^m
    \|a_{\varepsilon,r}^{(1)}-a_{\varepsilon',r}^{(1)}\|_2^2
    \prod_{q=2}^d
    \|a_r^{(q)}\|_2^2
    \notag\\
    &=
    \frac{1}{m^2}
    \sum_{r=1}^m
    \frac{4\delta^2 d_H(\varepsilon_{r,\cdot},\varepsilon'_{r,\cdot})}{s^2}
    \left(\frac{1}{s}\right)^{d-1}
    \notag\\
    &=
    \frac{4\delta^2}{m^2s^{d+1}}
    d_H(\varepsilon,\varepsilon').
    \label{eq:general-d-frob-exact}
\end{align}
Using \eqref{eq:general-d-packing}, we get
\[
    \|\P_\varepsilon-\P_{\varepsilon'}\|_F^2
    \gtrsim
    \frac{\delta^2}{ms^d}.
\]
Since $1/(ms^{d-1})\asymp F$ by \eqref{eq:general-d-fiber-equivalence},
\begin{equation}
\label{eq:general-d-frob-separation}
    \|\P_\varepsilon-\P_{\varepsilon'}\|_F^2
    \gtrsim
    \frac{\delta^2F}{s}.
\end{equation}

By Fano's lemma applied to the Frobenius metric, using \eqref{eq:general-d-fano-kl} and \eqref{eq:general-d-frob-separation},
\[
    \inf_{\widehat \P}
    \sup_{\varepsilon\in\mathcal E}
    \mathbb E_{\P_\varepsilon}\|\widehat \P-\P_\varepsilon\|_F
    \gtrsim
    \left(\frac{\delta^2F}{s}\right)^{1/2}.
\]
Substituting \eqref{eq:general-d-delta} gives
\[
    \left(\frac{\delta^2F}{s}\right)^{1/2}
    =
    \left(
        \frac{\gamma ms}{n}\cdot\frac{F}{s}
    \right)^{1/2}
    =
    \sqrt{\frac{\gamma mF}{n}}
    \asymp
    \sqrt{\frac{mF}{n}}.
\]
This proves the claimed minimax lower bound over $\mathcal G_n$.

Finally, if $F\gtrsim R^{d-2}/p^{d-1}$, then
\[
    \left(p^{d-1}F\right)^{1/(d-2)}
    \gtrsim
    R.
\]
Therefore $m\asymp R$, up to constants and integer rounding. In this regime, the sample-size condition becomes
\[
    n
    \gtrsim
    R^{(d-2)/(d-1)}F^{-1/(d-1)},
\]
and the lower bound becomes
\[
    \sqrt{\frac{mF}{n}}
    \asymp
    \sqrt{\frac{RF}{n}}.
\]
The proof is complete.
\end{proof}

\begin{proof}[Proof of Remark~\ref{rem:exact-rank}]
Write $F=F_{\operatorname{fiber}}$. In the regime
$F\gtrsim R^{d-2}/p^{d-1}$, choose an even integer $s$ such that
\begin{equation}\label{eq:exact-rank-s-choice}
    s\asymp (RF)^{-1/(d-1)} .
\end{equation}
Since $F\gtrsim R^{d-2}/p^{d-1}$, we have
\[
    (RF)^{-1/(d-1)}
    \lesssim
    p/R .
\]
Thus, after adjusting constants, we can choose disjoint blocks
\[
    A_{q,1},\ldots,A_{q,R}\subset[p],
    \qquad
    |A_{q,r}|=s,
    \qquad
    q\in[d],\ r\in[R].
\]
For $q=2,\ldots,d$, define
\begin{equation}\label{eq:exact-rank-fixed-factors}
    a_r^{(q)}
    =
    \frac1s\mathbf 1_{A_{q,r}},
    \qquad
    r\in[R].
\end{equation}
Let $\mathcal E\subset\{\pm1\}^{R\times s}$ be a row-wise balanced packing such that
\begin{equation}\label{eq:exact-rank-packing}
    \sum_{j=1}^s\varepsilon_{r,j}=0,
    \qquad
    \log|\mathcal E|\gtrsim Rs,
    \qquad
    d_H(\varepsilon,\varepsilon')\gtrsim Rs
\end{equation}
for all distinct $\varepsilon,\varepsilon'\in\mathcal E$. For a sufficiently small constant $0<\delta\le 1/4$, define
\begin{equation}\label{eq:exact-rank-first-factor}
    a_{\varepsilon,r}^{(1)}(i)
    =
    \frac{1+\delta\varepsilon_{r,j}}{s}
    \quad
    \text{if } i \text{ is the } j\text{th element of } A_{1,r},
\end{equation}
and set $a_{\varepsilon,r}^{(1)}(i)=0$ for $i\notin A_{1,r}$. By the balance condition in \eqref{eq:exact-rank-packing}, each $a_{\varepsilon,r}^{(1)}$ is a probability vector. Define
\begin{equation}\label{eq:exact-rank-construction}
    \P_\varepsilon
    =
    \frac1R
    \sum_{r=1}^R
    a_{\varepsilon,r}^{(1)}
    \circ
    a_r^{(2)}
    \circ\cdots\circ
    a_r^{(d)} .
\end{equation}
Then $\P_\varepsilon$ is a nonnegative probability tensor and $\operatorname{rank}_{\mathrm{CP}}(\P_\varepsilon)\le R$.

We prove the reverse inequality. Consider the mode-$2$ matricization:
\begin{equation}\label{eq:exact-rank-mode2}
    \MM_2(\P_\varepsilon)
    =
    \frac1R
    \sum_{r=1}^R
    a_r^{(2)}
    \left(
        a_{\varepsilon,r}^{(1)}
        \otimes
        a_r^{(3)}
        \otimes\cdots\otimes
        a_r^{(d)}
    \right)^\top .
\end{equation}
The vectors $a_1^{(2)},\ldots,a_R^{(2)}$ are linearly independent because their supports are disjoint. The vectors
\[
    a_{\varepsilon,r}^{(1)}
    \otimes
    a_r^{(3)}
    \otimes\cdots\otimes
    a_r^{(d)},
    \qquad
    r\in[R],
\]
are also linearly independent because their supports are disjoint in mode $3$. Therefore $\operatorname{rank}\{\MM_2(\P_\varepsilon)\}=R$. Since the rank of any matricization is bounded above by the CP rank, we have
\[
    \operatorname{rank}_{\mathrm{CP}}(\P_\varepsilon)
    \ge
    \operatorname{rank}\{\MM_2(\P_\varepsilon)\}
    =
    R.
\]
Together with $\operatorname{rank}_{\mathrm{CP}}(\P_\varepsilon)\le R$, this gives
\begin{equation}\label{eq:exact-rank-cp-rank}
    \operatorname{rank}_{\mathrm{CP}}(\P_\varepsilon)=R .
\end{equation}

Next, by disjointness of the blocks, every fiber intersects the support of at most one mixture component. Hence
\begin{equation}\label{eq:exact-rank-fiber}
    \operatorname{Fiber}_{\ell_1}(\P_\varepsilon)
    \le
    C\frac1{Rs^{d-1}}
    \asymp
    F,
\end{equation}
where the last equivalence follows from \eqref{eq:exact-rank-s-choice}. The mode-wise column spaces are supported on disjoint blocks of size $s$. Therefore their orthonormal bases $U_1,\ldots,U_d$ can be chosen so that
\begin{equation}\label{eq:exact-rank-coherence}
    \max_{q\in[d]}\|U_q\|_{2,\infty}
    \le
    C s^{-1/2}
    \asymp
    (RF)^{1/(2(d-1))}.
\end{equation}

It remains to verify the KL and separation bounds. For two packing elements $\varepsilon,\varepsilon'\in\mathcal E$, the product structure and disjointness imply
\begin{equation}\label{eq:exact-rank-kl-decomp}
    \KL(\P_\varepsilon\|\P_{\varepsilon'})
    =
    \frac1R
    \sum_{r=1}^R
    \KL(a_{\varepsilon,r}^{(1)}\|a_{\varepsilon',r}^{(1)}).
\end{equation}
Let $h_r=d_H(\varepsilon_{r,\cdot},\varepsilon'_{r,\cdot})$. Since both rows are balanced, the number of $+1$ to $-1$ flips equals the number of $-1$ to $+1$ flips. Hence
\[
    \KL(a_{\varepsilon,r}^{(1)}\|a_{\varepsilon',r}^{(1)})
    =
    \frac{h_r}{s}\delta
    \log\frac{1+\delta}{1-\delta}
    \le
    C\delta^2\frac{h_r}{s},
\]
where we used $\delta\le 1/4$. Combining this with \eqref{eq:exact-rank-kl-decomp}, we obtain
\[
    \KL(\P_\varepsilon\|\P_{\varepsilon'})
    \le
    \frac{C\delta^2}{Rs}
    \sum_{r=1}^R h_r
    \le
    C\delta^2 .
\]
Now choose
\begin{equation}\label{eq:exact-rank-delta}
    \delta^2
    =
    \gamma\frac{Rs}{n},
\end{equation}
where $\gamma>0$ is a sufficiently small absolute constant. Since
$s\asymp (RF)^{-1/(d-1)}$, the sample-size condition in the regime $m=R$ is
\[
    n
    \gtrsim
    R^{(d-2)/(d-1)}F^{-1/(d-1)}
    \asymp
    Rs.
\]
Thus, after choosing $\gamma$ small enough, $\delta\le 1/4$. For $n$ observations,
\begin{equation}\label{eq:exact-rank-fano-kl}
    \KL(\P_\varepsilon^{\otimes n}\|\P_{\varepsilon'}^{\otimes n})
    \le
    Cn\delta^2
    =
    C\gamma Rs
    \le
    \alpha\log|\mathcal E|,
\end{equation}
for a sufficiently small absolute constant $\alpha>0$, where we used $\log|\mathcal E|\gtrsim Rs$.

For the Frobenius separation, using disjoint supports gives
\begin{align}
    \|\P_\varepsilon-\P_{\varepsilon'}\|_F^2
    &=
    \frac1{R^2}
    \sum_{r=1}^R
    \|a_{\varepsilon,r}^{(1)}-a_{\varepsilon',r}^{(1)}\|_2^2
    \prod_{q=2}^d
    \|a_r^{(q)}\|_2^2
    \notag\\
    &=
    \frac{4\delta^2}{R^2s^{d+1}}
    d_H(\varepsilon,\varepsilon') .
    \label{eq:exact-rank-frob-exact}
\end{align}
By \eqref{eq:exact-rank-packing},
\begin{equation}\label{eq:exact-rank-frob-separation}
    \|\P_\varepsilon-\P_{\varepsilon'}\|_F^2
    \gtrsim
    \frac{\delta^2}{Rs^d}.
\end{equation}
Since $1/(Rs^{d-1})\asymp F$, \eqref{eq:exact-rank-frob-separation} yields
\[
    \|\P_\varepsilon-\P_{\varepsilon'}\|_F^2
    \gtrsim
    \frac{\delta^2F}{s}.
\]
Substituting \eqref{eq:exact-rank-delta}, we get
\[
    \frac{\delta^2F}{s}
    =
    \gamma\frac{Rs}{n}\cdot\frac{F}{s}
    =
    \gamma\frac{RF}{n}.
\]
Fano's lemma, together with \eqref{eq:exact-rank-fano-kl}, gives
\[
    \inf_{\widehat \P}
    \sup_{\varepsilon\in\mathcal E}
    \mathbb E_{\P_\varepsilon}
    \|\widehat \P-\P_\varepsilon\|_F
    \gtrsim
    \sqrt{\frac{RF}{n}} .
\]
By \eqref{eq:exact-rank-cp-rank}, every tensor in this packing has exact CP rank $R$. This proves the claim.
\end{proof}

\section{Proof of Theorem~\ref{thm_lower_bound}}

\begin{proof}
We construct a finite packing of rank-$R$ multiview density tensors. For simplicity, assume first that $p$ is even. If $p$ is odd, replace $p$ by $2\lfloor p/2\rfloor$ in the construction below and embed the resulting tensors into $[p]^d$ by setting the remaining coordinates to zero. This only changes constants depending on $d$.

Let
\begin{equation}\label{eq:l1-lb-general-d-s-def}
    s=\left\lfloor \frac{p}{R}\right\rfloor .
\end{equation}
Since $R\le c_d p$ and $c_d>0$ is sufficiently small, we have
\begin{equation}\label{eq:l1-lb-general-d-s-asymp}
    s\asymp \frac{p}{R},
    \qquad
    Rs\le p .
\end{equation}
For each mode $q=2,\ldots,d$, choose disjoint blocks
\begin{equation}\label{eq:l1-lb-general-d-blocks}
    A_{q,1},\ldots,A_{q,R}\subset[p],
    \qquad
    |A_{q,r}|=s .
\end{equation}
For $q=2,\ldots,d$ and $r\in[R]$, define
\begin{equation}\label{eq:l1-lb-general-d-fixed-factors}
    a_r^{(q)}
    =
    \frac{1}{s}\mathbf 1_{A_{q,r}} .
\end{equation}
These are probability vectors.

We next construct a packing for the first-mode factors. Let $\mathcal B$ be the set of sign arrays
\[
    \varepsilon=(\varepsilon_{r,i})_{r\in[R],i\in[p]}\in\{\pm1\}^{R\times p}
\]
such that each row is balanced:
\begin{equation}\label{eq:l1-lb-general-d-balanced}
    \sum_{i=1}^p \varepsilon_{r,i}=0,
    \qquad
    r=1,\ldots,R .
\end{equation}
For $\varepsilon\in\mathcal B$, define the associated $p\times R$ sign matrix
\begin{equation}\label{eq:l1-lb-general-d-E-def}
    E_\varepsilon=(\varepsilon_{r,i})_{i\in[p],r\in[R]} .
\end{equation}

We first restrict attention to sign matrices with controlled leverage. For a random element $\varepsilon$ drawn uniformly from $\mathcal B$, the columns of $E_\varepsilon$ are independent balanced sign vectors. Therefore Lemma~\ref{lem:balanced-sign-smin} implies that, if $R\le c_d p$ with $c_d>0$ sufficiently small, then
\begin{equation}\label{eq:l1-lb-general-d-smin-random}
    \sigma_R(E_\varepsilon)
    \ge
    c\sqrt p
\end{equation}
with probability at least $1-\exp(-cp)$.

On the event \eqref{eq:l1-lb-general-d-smin-random}, if $U_{E_\varepsilon}$ is any orthonormal basis of $\operatorname{col}(E_\varepsilon)$, then
\begin{align}
    \max_{i\in[p]}\|e_i^\top U_{E_\varepsilon}\|_2^2
    &=
    \max_{i\in[p]}
    e_i^\top
    E_\varepsilon
    (E_\varepsilon^\top E_\varepsilon)^{-1}
    E_\varepsilon^\top
    e_i
    \notag\\
    &\le
    \max_{i\in[p]}
    \frac{\|e_i^\top E_\varepsilon\|_2^2}{\sigma_R^2(E_\varepsilon)}
    \notag\\
    &\le
    C\frac{R}{p},
    \label{eq:l1-lb-general-d-E-leverage}
\end{align}
where the last inequality follows from \eqref{eq:l1-lb-general-d-smin-random} and the fact that each row of $E_\varepsilon$ has squared Euclidean norm $R$.

Let $\mathcal B_{\rm good}\subset\mathcal B$ be the subset of balanced sign arrays satisfying \eqref{eq:l1-lb-general-d-E-leverage}. Since \eqref{eq:l1-lb-general-d-E-leverage} holds with probability at least $1-\exp(-cp)$ under the uniform distribution on $\mathcal B$, a constant fraction of balanced sign arrays belongs to $\mathcal B_{\rm good}$. Hence
\begin{equation}\label{eq:l1-lb-general-d-good-size}
    |\mathcal B_{\rm good}|
    \ge
    c|\mathcal B|,
    \qquad
    \log|\mathcal B_{\rm good}|
    \ge
    cRp .
\end{equation}

We now construct a Hamming packing inside $\mathcal B_{\rm good}$ by a greedy argument. For any fixed balanced sign array, the number of balanced sign arrays within Hamming distance at most $\alpha Rp$ is bounded by the size of the corresponding Hamming ball in $\{\pm1\}^{Rp}$, and is therefore at most $\exp\{C H(\alpha)Rp\}$, where $H(\alpha)$ is the binary entropy function. Choosing $\alpha>0$ sufficiently small and using \eqref{eq:l1-lb-general-d-good-size}, there exists a subset $\mathcal E\subset\mathcal B_{\rm good}$ such that
\begin{equation}\label{eq:l1-lb-general-d-packing}
    \log|\mathcal E|\ge cRp,
    \qquad
    d_H(\varepsilon,\varepsilon')\ge cRp
\end{equation}
for all distinct $\varepsilon,\varepsilon'\in\mathcal E$.

Let
\begin{equation}\label{eq:l1-lb-general-d-delta-def}
    \delta
    =
    \gamma\sqrt{\frac{Rp}{n}},
\end{equation}
where $\gamma>0$ is a sufficiently small absolute constant. Since $n\ge C_d Rp$ and $C_d>0$ is sufficiently large, we have
\begin{equation}\label{eq:l1-lb-general-d-delta-small}
    0<\delta\le \frac14 .
\end{equation}
For $\varepsilon\in\mathcal E$ and $r\in[R]$, define the first-mode factor
\begin{equation}\label{eq:l1-lb-general-d-first-factor}
    a_{\varepsilon,r}^{(1)}(i)
    =
    \frac{1+\delta\varepsilon_{r,i}}{p},
    \qquad
    i\in[p].
\end{equation}
By \eqref{eq:l1-lb-general-d-balanced} and \eqref{eq:l1-lb-general-d-delta-small}, each $a_{\varepsilon,r}^{(1)}$ is a probability vector.

Define
\begin{equation}\label{eq:l1-lb-general-d-P-def}
    \P_\varepsilon
    =
    \frac1R
    \sum_{r=1}^R
    a_{\varepsilon,r}^{(1)}
    \circ
    a_r^{(2)}
    \circ\cdots\circ
    a_r^{(d)} .
\end{equation}
Then $\P_\varepsilon$ is a nonnegative probability tensor and is a multiview density tensor.

We next verify the rank condition. The representation \eqref{eq:l1-lb-general-d-P-def} immediately gives
\begin{equation}\label{eq:l1-lb-general-d-rank-upper}
    \operatorname{rank}_{\mathrm{CP}}(\P_\varepsilon)\le R .
\end{equation}
For the reverse inequality, consider the mode-$2$ matricization. From \eqref{eq:l1-lb-general-d-P-def},
\begin{equation}\label{eq:l1-lb-general-d-mode2}
    \MM_2(\P_\varepsilon)
    =
    \frac1R
    \sum_{r=1}^R
    a_r^{(2)}
    \left(
        a_{\varepsilon,r}^{(1)}
        \otimes
        a_r^{(3)}
        \otimes\cdots\otimes
        a_r^{(d)}
    \right)^\top .
\end{equation}
The vectors $a_1^{(2)},\ldots,a_R^{(2)}$ are linearly independent because their supports are disjoint by \eqref{eq:l1-lb-general-d-blocks}. The right-side vectors in \eqref{eq:l1-lb-general-d-mode2} are also linearly independent. Indeed, if
\[
    \sum_{r=1}^R
    \alpha_r
    \left(
        a_{\varepsilon,r}^{(1)}
        \otimes
        a_r^{(3)}
        \otimes\cdots\otimes
        a_r^{(d)}
    \right)
    =
    0,
\]
then restricting the mode-$3$ coordinate to the block $A_{3,r}$ eliminates all terms except the $r$th term, by \eqref{eq:l1-lb-general-d-blocks}. Since all factors are nonzero, $\alpha_r=0$. This holds for every $r\in[R]$, so the right-side vectors are linearly independent. Therefore
\begin{equation}\label{eq:l1-lb-general-d-mode2-rank}
    \operatorname{rank}\{\MM_2(\P_\varepsilon)\}=R .
\end{equation}
Since the rank of any matricization is bounded above by CP rank, \eqref{eq:l1-lb-general-d-mode2-rank} gives
\begin{equation}\label{eq:l1-lb-general-d-rank-lower}
    \operatorname{rank}_{\mathrm{CP}}(\P_\varepsilon)\ge R .
\end{equation}
Combining \eqref{eq:l1-lb-general-d-rank-upper} and \eqref{eq:l1-lb-general-d-rank-lower}, we obtain
\begin{equation}\label{eq:l1-lb-general-d-rank-exact}
    \operatorname{rank}_{\mathrm{CP}}(\P_\varepsilon)=R .
\end{equation}

We now bound the fiber mass. First consider mode-$1$ fibers. Fix $(i_2,\ldots,i_d)$. By the disjointness of the blocks in \eqref{eq:l1-lb-general-d-blocks}, at most one component $r$ contributes to this fiber. Using \eqref{eq:l1-lb-general-d-balanced}, \eqref{eq:l1-lb-general-d-fixed-factors}, and \eqref{eq:l1-lb-general-d-first-factor}, we get
\begin{equation}\label{eq:l1-lb-general-d-fiber-mode1}
    \sum_{i_1=1}^p
    (\P_\varepsilon)_{i_1,i_2,\ldots,i_d}
    \le
    \frac{1}{R s^{d-1}} .
\end{equation}
Next consider a mode-$q$ fiber with $q\ge 2$. Fix all coordinates except $i_q$. Again, by the disjointness of the blocks in the remaining modes, at most one component $r$ contributes. Therefore, by \eqref{eq:l1-lb-general-d-delta-small},
\begin{equation}\label{eq:l1-lb-general-d-fiber-modeq}
    \sum_{i_q=1}^p
    (\P_\varepsilon)_{i_1,\ldots,i_d}
    \le
    \frac{2}{R p s^{d-2}}
    \le
    \frac{2}{R s^{d-1}},
\end{equation}
where the last inequality uses $s\le p$. Combining \eqref{eq:l1-lb-general-d-fiber-mode1}, \eqref{eq:l1-lb-general-d-fiber-modeq}, and \eqref{eq:l1-lb-general-d-s-asymp}, we obtain
\begin{equation}\label{eq:l1-lb-general-d-fiber-bound}
    \operatorname{Fiber}_{\ell_1}(\P_\varepsilon)
    \le
    \frac{C_d}{R s^{d-1}}
    \le
    C_d\frac{R^{d-2}}{p^{d-1}} .
\end{equation}

We next verify the coherence bound. For mode $1$, the mode-wise column space of $\P_\varepsilon$ is contained in
\[
    \operatorname{span}\{\mathbf 1\}+\operatorname{col}(E_\varepsilon).
\]
Since the columns of $E_\varepsilon$ are balanced, $\operatorname{col}(E_\varepsilon)$ is orthogonal to $\mathbf 1$. Hence the row leverage of this containing space is bounded by
\begin{equation}\label{eq:l1-lb-general-d-mode1-coherence}
    \frac1p
    +
    \max_{i\in[p]}\|e_i^\top U_{E_\varepsilon}\|_2^2
    \le
    C\frac{R}{p},
\end{equation}
where we used \eqref{eq:l1-lb-general-d-E-leverage} and $R\ge 1$. Since the actual mode-$1$ column space is a subspace of this containing space, it also satisfies \eqref{eq:l1-lb-general-d-mode1-coherence}. For $q=2,\ldots,d$, the mode-$q$ column space is spanned by the disjoint block indicators $\mathbf 1_{A_{q,r}}$, $r\in[R]$. An orthonormal basis is given by $s^{-1/2}\mathbf 1_{A_{q,r}}$, $r\in[R]$, and therefore
\begin{equation}\label{eq:l1-lb-general-d-modeq-coherence}
    \max_{i\in[p]}\|e_i^\top U_q\|_2^2
    =
    \frac1s
    \le
    C\frac{R}{p},
    \qquad
    q=2,\ldots,d,
\end{equation}
where the last inequality follows from \eqref{eq:l1-lb-general-d-s-asymp}. Combining \eqref{eq:l1-lb-general-d-mode1-coherence} and \eqref{eq:l1-lb-general-d-modeq-coherence}, we obtain
\begin{equation}\label{eq:l1-lb-general-d-coherence-bound}
    \max_{q\in[d]}\|U_q\|_{2,\infty}^2
    \le
    C_d\frac{R}{p}.
\end{equation}

We now control the KL divergence. Since the components in \eqref{eq:l1-lb-general-d-P-def} have disjoint supports in modes $2,\ldots,d$, and since the factors in these modes do not depend on $\varepsilon$, we have
\begin{equation}\label{eq:l1-lb-general-d-kl-decomp}
    \KL(\P_\varepsilon\|\P_{\varepsilon'})
    =
    \frac1R
    \sum_{r=1}^R
    \KL(a_{\varepsilon,r}^{(1)}\|a_{\varepsilon',r}^{(1)}) .
\end{equation}
For a fixed $r$, let
\[
    h_r
    =
    d_H(\varepsilon_{r,\cdot},\varepsilon'_{r,\cdot}) .
\]
Since both $\varepsilon_{r,\cdot}$ and $\varepsilon'_{r,\cdot}$ are balanced, the number of $+1$ to $-1$ flips equals the number of $-1$ to $+1$ flips. Therefore,
\begin{align}
    \KL(a_{\varepsilon,r}^{(1)}\|a_{\varepsilon',r}^{(1)})
    &=
    \frac{h_r}{2p}
    \left[
        (1+\delta)\log\frac{1+\delta}{1-\delta}
        +
        (1-\delta)\log\frac{1-\delta}{1+\delta}
    \right]
    \notag\\
    &=
    \frac{\delta h_r}{p}
    \log\frac{1+\delta}{1-\delta}.
    \label{eq:l1-lb-general-d-single-kl-exact}
\end{align}
By \eqref{eq:l1-lb-general-d-delta-small},
\[
    \log\frac{1+\delta}{1-\delta}
    \le
    C\delta .
\]
Hence \eqref{eq:l1-lb-general-d-single-kl-exact} gives
\begin{equation}\label{eq:l1-lb-general-d-single-kl-bound}
    \KL(a_{\varepsilon,r}^{(1)}\|a_{\varepsilon',r}^{(1)})
    \le
    C\delta^2\frac{h_r}{p}.
\end{equation}
Combining \eqref{eq:l1-lb-general-d-kl-decomp} and \eqref{eq:l1-lb-general-d-single-kl-bound},
\begin{equation}\label{eq:l1-lb-general-d-one-sample-kl}
    \KL(\P_\varepsilon\|\P_{\varepsilon'})
    \le
    \frac{C\delta^2}{Rp}
    \sum_{r=1}^R h_r
    =
    \frac{C\delta^2}{Rp}
    d_H(\varepsilon,\varepsilon')
    \le
    C\delta^2 .
\end{equation}
For $n$ independent observations, \eqref{eq:l1-lb-general-d-one-sample-kl} and \eqref{eq:l1-lb-general-d-delta-def} imply
\begin{equation}\label{eq:l1-lb-general-d-n-sample-kl}
    \KL(\P_\varepsilon^{\otimes n}\|\P_{\varepsilon'}^{\otimes n})
    \le
    Cn\delta^2
    =
    C\gamma^2 Rp .
\end{equation}
By \eqref{eq:l1-lb-general-d-packing}, $\log|\mathcal E|\ge cRp$. Thus, choosing $\gamma>0$ sufficiently small in \eqref{eq:l1-lb-general-d-delta-def}, \eqref{eq:l1-lb-general-d-n-sample-kl} gives
\begin{equation}\label{eq:l1-lb-general-d-fano-kl}
    \max_{\varepsilon\neq\varepsilon'}
    \KL(\P_\varepsilon^{\otimes n}\|\P_{\varepsilon'}^{\otimes n})
    \le
    \frac18\log|\mathcal E| .
\end{equation}

It remains to lower bound the $\ell_1$ separation. By the disjointness of the blocks in modes $2,\ldots,d$ and by \eqref{eq:l1-lb-general-d-fixed-factors},
\begin{align}
    \|\P_\varepsilon-\P_{\varepsilon'}\|_1
    &=
    \frac1R
    \sum_{r=1}^R
    \|a_{\varepsilon,r}^{(1)}-a_{\varepsilon',r}^{(1)}\|_1
    \prod_{q=2}^d
    \|a_r^{(q)}\|_1
    \notag\\
    &=
    \frac1R
    \sum_{r=1}^R
    \frac{2\delta}{p}
    d_H(\varepsilon_{r,\cdot},\varepsilon'_{r,\cdot})
    \notag\\
    &=
    \frac{2\delta}{Rp}
    d_H(\varepsilon,\varepsilon') .
    \label{eq:l1-lb-general-d-l1-exact}
\end{align}
Using the packing separation in \eqref{eq:l1-lb-general-d-packing}, we obtain
\begin{equation}\label{eq:l1-lb-general-d-l1-separation}
    \|\P_\varepsilon-\P_{\varepsilon'}\|_1
    \ge
    c\delta
\end{equation}
for all distinct $\varepsilon,\varepsilon'\in\mathcal E$.

Let
\begin{equation}\label{eq:l1-lb-general-d-G-def}
    \mathcal G_n
    =
    \{\P_\varepsilon:\varepsilon\in\mathcal E\}.
\end{equation}
By \eqref{eq:l1-lb-general-d-rank-exact}, \eqref{eq:l1-lb-general-d-fiber-bound}, and \eqref{eq:l1-lb-general-d-coherence-bound}, every tensor in $\mathcal G_n$ satisfies the claimed rank, fiber-mass, and coherence conditions. Applying Fano's lemma with the KL condition \eqref{eq:l1-lb-general-d-fano-kl} and the separation bound \eqref{eq:l1-lb-general-d-l1-separation} gives
\begin{equation}\label{eq:l1-lb-general-d-fano-conclusion}
    \inf_{\widehat \P}
    \sup_{\P\in\mathcal G_n}
    \mathbb E_{\P}\|\widehat \P-\P\|_1
    \ge
    c\delta .
\end{equation}
Substituting the definition of $\delta$ in \eqref{eq:l1-lb-general-d-delta-def} into \eqref{eq:l1-lb-general-d-fano-conclusion}, we get
\[
    \inf_{\widehat \P}
    \sup_{\P\in\mathcal G_n}
    \mathbb E_{\P}\|\widehat \P-\P\|_1
    \ge
    c_d\sqrt{\frac{Rp}{n}} .
\]
The proof is complete.
\end{proof}

\section{Lemmas}

\begin{Lemma}[Lemma 6 in \cite{zhang2018tensor}]\label{lemma_svd_projection}
Suppose $X, Z \in \mathbb{R}^{p_1 \times p_2}, \operatorname{rank}(X)=r$. If the singular value decomposition of $X$ and $Y$ are written as
$$
Y=X+Z=\hat{U} \hat{\Sigma} \hat{V}^{\top}=\left[\begin{array}{ll}
\hat{U}_1 & \hat{U}_2
\end{array}\right] \cdot\left[\begin{array}{cc}
\hat{\Sigma}_1 & \\
& \hat{\Sigma}_2
\end{array}\right] \cdot\left[\begin{array}{ll}
\hat{V}_1^{\top} & \hat{V}_2^{\top}
\end{array}\right],
$$
where $\hat{U}_1 \in \mathbb{O}_{p_1, r}, \hat{V}_1 \in \mathbb{O}_{p_2, r}$ correspond to the leading $r$ left and right singular vectors; and $\hat{U}_2 \in \mathbb{O}_{p_1, p_1-r}, \hat{V}_2 \in \mathbb{O}_{p_2, p_2-r}$ correspond to their orthonormal complement. Then

$$
\left\|P_{\hat{U}_2} X\right\| \leq 2\|Z\|, \quad\left\|P_{\hat{U}_2} X\right\|_{\mathrm{F}} \leq \min \left\{2 \sqrt{r}\|Z\|, 2\|Z\|_{\mathrm{F}}\right\}. 
$$
\end{Lemma}

\begin{Lemma}[Multinomial noise control with a general projection matrix]
\label{lemma_multinomial_noise_control}
Suppose $P \in \mathbb{R}^{p_1 \times p_{-1}}$ is a probability matrix with
$P_{ij} \ge 0$ and $\sum_{i,j} P_{ij} = 1$.
Let $Q = nP$, $Y \sim \operatorname{Multinomial}(n,P)$, and $Z = Y - Q$.
Let $W \in \mathbb{R}^{p_{-1}\times r_{-1}}$ and
$V \in \mathbb{R}^{p_1\times r_1}$ satisfy
\[
    \|W\| \le 1,
    \qquad
    \|W\|_F^2 \le L,
    \qquad
    \|W\|_{2,\infty} \le \rho,
    \qquad
    \|V\| \le 1.
\]
Then for any given $D \in \mathbb{R}^{p_1 \times p_{-1}}$, we have
\begin{equation}
\label{eq_mult_noise_op_tail_stmt}
    \PP\left(
        \left\|(D * Z) W\right\|
        \ge
        C\bigl(\sqrt{\eta t} + \|D\|_{\infty} t\bigr)
    \right)
    \le
    (p_1 + r_{-1}) e^{-t},
\end{equation}
\begin{equation}
\label{eq_mult_noise_frob_tail_stmt}
    \PP\left(
        \left\|V^\top (D * Z) W\right\|_F
        \ge
        C\bigl(\sqrt{r_1 \eta t} + \sqrt{r_1}\|D\|_{\infty} t\bigr)
    \right)
    \le
    (p_1 + r_{-1}) e^{-t},
\end{equation}
and
\begin{equation}
\label{eq_mult_noise_row_tail_stmt}
    \PP\left(
        \left\|(D * Z) W\right\|_{2,\infty}
        \ge
        C\bigl(\sqrt{\eta_{W,2,\infty} t} + \rho\|D\|_\infty t\bigr)
    \right)
    \le
    p_1(1+r_{-1})e^{-t},
\end{equation}
where
\[
    \eta
    :=
    \max_{k \in [p_{-1}]}\sum_{h \in [p_1]} (D * D * Q)_{hk}
    \;\vee\;
    L\|D * D * Q\|_\infty,
\]
and
\[
    \eta_{W,2,\infty}
    :=
    \max_{h\in[p_1]}
    \sum_{k\in[p_{-1}]}
    (D*D*Q)_{hk}\,\|e_k^\top W\|_2^2.
\]
Here $C>0$ is an absolute constant.

Furthermore, there exists an absolute constant $C>0$ such that for any
$\tau \ge 2$, with probability at least
\[
    1-(p_1+r_{-1})^{1-\tau}-2(p_1r_{-1})^{1-\tau},
\]
the following hold:
\begin{equation}
\label{eq_mult_noise_op_hp_stmt}
    \left\|(D * Z)W\right\|
    \le
    \tau C\Bigl(
        \sqrt{\eta\log(p_1+r_{-1})}
        +
        \|D\|_\infty\log(p_1+r_{-1})
    \Bigr),
\end{equation}
\begin{equation}
\label{eq_mult_noise_frob_hp_stmt}
    \left\|V^\top (D * Z)W\right\|_F
    \le
    \tau C\Bigl(
        \sqrt{r_1\eta\log(p_1+r_{-1})}
        +
        \sqrt{r_1}\|D\|_\infty\log(p_1+r_{-1})
    \Bigr),
\end{equation}
and
\begin{equation}
\label{eq_mult_noise_row_hp_stmt}
    \left\|(D * Z)W\right\|_{2,\infty}
    \le
    \tau C\Bigl(
        \sqrt{\eta_{W,2,\infty}\log(p_1r_{-1})}
        +
        \rho\|D\|_\infty\log(p_1r_{-1})
    \Bigr).
\end{equation}

In particular, since
\[
    \eta_{W,2,\infty}
    \le
    \rho^2
    \max_{h\in[p_1]}
    \sum_{k\in[p_{-1}]}
    (D*D*Q)_{hk},
\]
the bound in \eqref{eq_mult_noise_row_tail_stmt} can be further simplified to
\[
    \left\|(D * Z) W\right\|_{2,\infty}
    \lesssim
    \rho\Bigl(
        \sqrt{
            \Bigl[
                \max_{h\in[p_1]}
                \sum_{k\in[p_{-1}]}
                (D*D*Q)_{hk}
            \Bigr] t
        }
        +
        \|D\|_\infty t
    \Bigr)
\]
with the same failure probability.
\end{Lemma}

\begin{proof}
Since $Y \sim \operatorname{Multinomial}(n,P)$, we may write
\begin{equation}
\label{eq_mult_noise_Y_sum}
    Y = \sum_{\ell=1}^n E_\ell,
\end{equation}
where $E_1,\ldots,E_n$ are i.i.d. random matrices taking values in
\[
    \{e_h e_k^\top : h \in [p_1],\ k \in [p_{-1}]\}
\]
with
\[
    \PP(E_\ell = e_h e_k^\top) = P_{hk}.
\]
Hence
\begin{equation}
\label{eq_mult_noise_Z_sum}
    \EE E_\ell = P,
    \qquad
    Z = Y - Q = \sum_{\ell=1}^n (E_\ell - P),
\end{equation}
since $Q = nP$.

Define
\begin{equation}
\label{eq_mult_noise_X_def}
    X_\ell := (D * (E_\ell - P))W.
\end{equation}
Then $X_1,\ldots,X_n$ are independent mean-zero random matrices and
\begin{equation}
\label{eq_mult_noise_DZW_sum}
    (D * Z)W = \sum_{\ell=1}^n X_\ell.
\end{equation}

We first record a uniform bound for the summands.
For any realization $E_\ell = e_h e_k^\top$,
\[
    (D * E_\ell)W = D_{hk}\, e_h e_k^\top W,
\]
and hence
\[
    \|(D * E_\ell)W\|
    =
    |D_{hk}|\,\|e_k^\top W\|_2
    \le
    \rho \|D\|_\infty.
\]
Also, by Jensen's inequality,
\[
    \|(D * P)W\|
    =
    \|\EE[(D * E_\ell)W]\|
    \le
    \EE\|(D * E_\ell)W\|
    \le
    \rho \|D\|_\infty.
\]
Therefore,
\begin{equation}
\label{eq_mult_noise_X_uniform_bound}
    \|X_\ell\|
    =
    \|(D * (E_\ell - P))W\|
    \le
    \|(D * E_\ell)W\| + \|(D * P)W\|
    \le
    2\rho \|D\|_\infty
\end{equation}
almost surely.

\medskip
\noindent\textbf{Part I: Operator norm bound.}

Since $\EE X_\ell = 0$, we have
\[
    \EE(X_\ell X_\ell^\top)
    \preceq
    \EE\Bigl[(D * E_\ell)WW^\top(D * E_\ell)^\top\Bigr].
\]
If $E_\ell = e_h e_k^\top$, then
\[
    (D * E_\ell)WW^\top(D * E_\ell)^\top
    =
    D_{hk}^2\, e_h e_k^\top WW^\top e_k e_h^\top
    =
    D_{hk}^2\, \|W^\top e_k\|_2^2\, e_h e_h^\top.
\]
Therefore,
\[
    \sum_{\ell=1}^n \EE(X_\ell X_\ell^\top)
    \preceq
    \sum_{h=1}^{p_1}\sum_{k=1}^{p_{-1}}
    Q_{hk}D_{hk}^2\|W^\top e_k\|_2^2\, e_h e_h^\top.
\]
Taking operator norm and using
\[
    \sum_{k=1}^{p_{-1}}\|W^\top e_k\|_2^2 = \|W\|_F^2 \le L,
\]
we obtain
\begin{equation}
\label{eq_mult_noise_var_left}
    \left\|\sum_{\ell=1}^n \EE(X_\ell X_\ell^\top)\right\|
    \le
    \max_{h\in[p_1]}\sum_{k=1}^{p_{-1}}
    Q_{hk}D_{hk}^2\|W^\top e_k\|_2^2
    \le
    L\|D * D * Q\|_\infty
    \le \eta.
\end{equation}

Similarly,
\[
    \EE(X_\ell^\top X_\ell)
    \preceq
    \EE\Bigl[W^\top (D * E_\ell)^\top(D * E_\ell)W\Bigr].
\]
If $E_\ell = e_h e_k^\top$, then
\[
    (D * E_\ell)^\top(D * E_\ell)
    =
    D_{hk}^2\, e_k e_k^\top.
\]
Hence
\[
    \sum_{\ell=1}^n \EE(X_\ell^\top X_\ell)
    \preceq
    W^\top
    \diag\left(
        \sum_{h=1}^{p_1}(D * D * Q)_{h1},
        \ldots,
        \sum_{h=1}^{p_1}(D * D * Q)_{h p_{-1}}
    \right)
    W.
\]
Since $\|W\| \le 1$,
\begin{equation}
\label{eq_mult_noise_var_right}
    \left\|\sum_{\ell=1}^n \EE(X_\ell^\top X_\ell)\right\|
    \le
    \max_{k\in[p_{-1}]}\sum_{h\in[p_1]}(D * D * Q)_{hk}
    \le \eta.
\end{equation}

Therefore, the variance parameter in the rectangular matrix Bernstein inequality satisfies
\begin{equation}
\label{eq_mult_noise_sigma}
    \sigma^2
    :=
    \max\left\{
        \left\|\sum_{\ell=1}^n \EE(X_\ell X_\ell^\top)\right\|,
        \left\|\sum_{\ell=1}^n \EE(X_\ell^\top X_\ell)\right\|
    \right\}
    \le \eta,
\end{equation}
by \eqref{eq_mult_noise_var_left} and \eqref{eq_mult_noise_var_right}.
Combining \eqref{eq_mult_noise_X_uniform_bound} and \eqref{eq_mult_noise_sigma},
the rectangular matrix Bernstein inequality yields
\begin{equation}
\label{eq_mult_noise_op_tail}
    \PP\left(
        \left\|(D * Z)W\right\|
        \ge
        C\bigl(\sqrt{\eta t} + \rho \|D\|_\infty t\bigr)
    \right)
    \le
    (p_1 + r_{-1})e^{-t}.
\end{equation}
This proves \eqref{eq_mult_noise_op_tail_stmt}.

\medskip
\noindent\textbf{Part II: Frobenius norm bound.}

Let
\begin{equation}
\label{eq_mult_noise_M_def}
    M := (D * Z)W.
\end{equation}
Then for any $V \in \mathbb{R}^{p_1 \times r_1}$ with $\|V\| \le 1$,
\begin{equation}
\label{eq_mult_noise_frob_from_op}
    \|V^\top M\|_F
    \le
    \sqrt{r_1}\,\|V^\top M\|
    \le
    \sqrt{r_1}\,\|V\|\,\|M\|
    \le
    \sqrt{r_1}\,\|M\|.
\end{equation}
Applying \eqref{eq_mult_noise_frob_from_op} together with
\eqref{eq_mult_noise_op_tail}, we obtain
\begin{equation}
\label{eq_mult_noise_frob_tail}
    \PP\left(
        \|V^\top (D * Z)W\|_F
        \ge
        C\bigl(\sqrt{r_1\eta t} + \sqrt{r_1}\rho \|D\|_\infty t\bigr)
    \right)
    \le
    (p_1 + r_{-1})e^{-t}.
\end{equation}
This proves \eqref{eq_mult_noise_frob_tail_stmt}.

\medskip
\noindent\textbf{Part III: $\|\cdot\|_{2,\infty}$ bound.}

Again let $M := (D * Z)W$.
Then
\[
    \|M\|_{2,\infty}
    =
    \max_{h\in[p_1]} \|e_h^\top M\|_2.
\]
Fix $h\in[p_1]$, and define
\begin{equation}
\label{eq_mult_noise_Xh_def}
    X_{\ell,h} := e_h^\top X_\ell.
\end{equation}
Then $X_{1,h},\ldots,X_{n,h}$ are independent mean-zero $1\times r_{-1}$ random vectors and
\begin{equation}
\label{eq_mult_noise_row_sum}
    e_h^\top M = \sum_{\ell=1}^n X_{\ell,h}.
\end{equation}

We first bound the norm of $X_{\ell,h}$.
If $E_\ell = e_a e_k^\top$, then
\[
    e_h^\top(D * E_\ell)W
    =
    \mathbf{1}\{a=h\}D_{hk}e_k^\top W,
\]
so
\[
    \|e_h^\top(D * E_\ell)W\|_2
    \le
    |D_{hk}|\,\|e_k^\top W\|_2
    \le
    \rho \|D\|_\infty.
\]
Also,
\[
    e_h^\top(D * P)W
    =
    \sum_{k=1}^{p_{-1}} D_{hk}P_{hk} e_k^\top W,
\]
hence
\[
    \|e_h^\top(D * P)W\|_2
    \le
    \sum_{k=1}^{p_{-1}} |D_{hk}|P_{hk}\|e_k^\top W\|_2
    \le
    \rho \|D\|_\infty \sum_{k=1}^{p_{-1}} P_{hk}
    \le
    \rho \|D\|_\infty.
\]
Therefore,
\begin{equation}
\label{eq_mult_noise_Xh_uniform_bound}
    \|X_{\ell,h}\|_2
    =
    \|e_h^\top(D * (E_\ell-P))W\|_2
    \le
    2\rho\|D\|_\infty
\end{equation}
almost surely.

Next, since $\EE X_{\ell,h} = 0$,
\[
    \EE(X_{\ell,h}X_{\ell,h}^\top)
    \le
    \EE\Bigl[e_h^\top(D * E_\ell)WW^\top(D * E_\ell)^\top e_h\Bigr].
\]
If $E_\ell = e_h e_k^\top$, then
\[
    e_h^\top(D * E_\ell)WW^\top(D * E_\ell)^\top e_h
    =
    D_{hk}^2 \|W^\top e_k\|_2^2.
\]
Hence
\begin{equation}
\label{eq_mult_noise_Xh_var_left}
    \sum_{\ell=1}^n \EE(X_{\ell,h}X_{\ell,h}^\top)
    \le
    \sum_{k=1}^{p_{-1}} Q_{hk}D_{hk}^2\|W^\top e_k\|_2^2
    \le
    \eta_{W,2,\infty}.
\end{equation}

Similarly,
\[
    \EE(X_{\ell,h}^\top X_{\ell,h})
    \preceq
    \EE\Bigl[
        W^\top(D * E_\ell)^\top e_h e_h^\top(D * E_\ell)W
    \Bigr].
\]
If $E_\ell = e_a e_k^\top$, this is zero unless $a=h$, and in that case
\[
    (D * E_\ell)^\top e_h e_h^\top(D * E_\ell)
    =
    D_{hk}^2 e_k e_k^\top.
\]
Therefore
\[
    \sum_{\ell=1}^n \EE(X_{\ell,h}^\top X_{\ell,h})
    \preceq
    \sum_{k=1}^{p_{-1}}
    (D*D*Q)_{hk}\,
    W^\top e_k e_k^\top W.
\]
Since each summand is positive semidefinite,
\[
    \left\|
        \sum_{\ell=1}^n \EE(X_{\ell,h}^\top X_{\ell,h})
    \right\|
    \le
    \sum_{k=1}^{p_{-1}}
    (D*D*Q)_{hk}\,
    \|W^\top e_k e_k^\top W\|.
\]
Using
\[
    \|W^\top e_k e_k^\top W\|
    =
    \|(e_k^\top W)^\top (e_k^\top W)\|
    =
    \|e_k^\top W\|_2^2,
\]
we obtain
\begin{equation}
\label{eq_mult_noise_Xh_var_right}
    \left\|
        \sum_{\ell=1}^n \EE(X_{\ell,h}^\top X_{\ell,h})
    \right\|
    \le
    \sum_{k=1}^{p_{-1}}
    (D*D*Q)_{hk}\,\|e_k^\top W\|_2^2
    \le
    \eta_{W,2,\infty}.
\end{equation}

Thus the variance parameter for the rectangular matrix Bernstein inequality applied to
$\sum_{\ell=1}^n X_{\ell,h}$ satisfies
\begin{equation}
\label{eq_mult_noise_Xh_sigma}
    \sigma_h^2
    :=
    \max\left\{
        \left\|\sum_{\ell=1}^n \EE(X_{\ell,h}X_{\ell,h}^\top)\right\|,
        \left\|\sum_{\ell=1}^n \EE(X_{\ell,h}^\top X_{\ell,h})\right\|
    \right\}
    \le
    \eta_{W,2,\infty},
\end{equation}
by \eqref{eq_mult_noise_Xh_var_left} and
\eqref{eq_mult_noise_Xh_var_right}.
Combining \eqref{eq_mult_noise_Xh_uniform_bound} and
\eqref{eq_mult_noise_Xh_sigma}, we get, for every $t>0$,
\begin{equation}
\label{eq_mult_noise_fixed_row_tail}
    \PP\left(
        \|e_h^\top M\|_2
        \ge
        C\bigl(\sqrt{\eta_{W,2,\infty} t} + \rho\|D\|_\infty t\bigr)
    \right)
    \le
    (1+r_{-1})e^{-t}.
\end{equation}
Taking a union bound over $h\in[p_1]$ and recalling that
$\|M\|_{2,\infty}=\max_h\|e_h^\top M\|_2$, we obtain
\begin{equation}
\label{eq_mult_noise_row_tail}
    \PP\left(
        \|M\|_{2,\infty}
        \ge
        C\bigl(\sqrt{\eta_{W,2,\infty} t} + \rho\|D\|_\infty t\bigr)
    \right)
    \le
    p_1(1+r_{-1})e^{-t}.
\end{equation}
This proves \eqref{eq_mult_noise_row_tail_stmt}.

\medskip
\noindent\textbf{High-probability consequences.}

Applying \eqref{eq_mult_noise_op_tail} with
\[
    t = \tau\log(p_1+r_{-1}),
    \qquad
    \tau \ge 2,
\]
we get
\[
    (p_1+r_{-1})e^{-t} = (p_1+r_{-1})^{1-\tau}.
\]
Hence, with probability at least $1-(p_1+r_{-1})^{1-\tau}$,
\begin{equation}
\label{eq_mult_noise_op_hp}
    \left\|(D * Z)W\right\|
    \le
    \tau C\Bigl(
        \sqrt{\eta\log(p_1+r_{-1})}
        +
        \rho \|D\|_\infty\log(p_1+r_{-1})
    \Bigr).
\end{equation}
By \eqref{eq_mult_noise_frob_tail}, on the same event,
\begin{equation}
\label{eq_mult_noise_frob_hp}
    \left\|V^\top (D * Z)W\right\|_F
    \le
    \tau C\Bigl(
        \sqrt{r_1\eta\log(p_1+r_{-1})}
        +
        \sqrt{r_1}\rho \|D\|_\infty\log(p_1+r_{-1})
    \Bigr).
\end{equation}

Applying \eqref{eq_mult_noise_row_tail} with
\[
    t = \tau\log(p_1r_{-1}),
    \qquad
    \tau \ge 2,
\]
we get failure probability at most
\[
    p_1(1+r_{-1})(p_1r_{-1})^{-\tau}
    \le
    2(p_1r_{-1})^{1-\tau}.
\]
Therefore, with probability at least $1-2(p_1r_{-1})^{1-\tau}$,
\begin{equation}
\label{eq_mult_noise_row_hp}
    \left\|(D * Z)W\right\|_{2,\infty}
    \le
    \tau C\Bigl(
        \sqrt{\eta_{W,2,\infty}\log(p_1r_{-1})}
        +
        \rho\|D\|_\infty\log(p_1r_{-1})
    \Bigr).
\end{equation}

Combining the events in \eqref{eq_mult_noise_op_hp},
\eqref{eq_mult_noise_frob_hp}, and \eqref{eq_mult_noise_row_hp}
by a union bound proves
\eqref{eq_mult_noise_op_hp_stmt}--\eqref{eq_mult_noise_row_hp_stmt}.
\end{proof}

\begin{Lemma}\label{lemma_offdiag_EEt_multinomial_mw_corrected_v2}
Let $E=D*Z$ under the conditions of Lemma~\ref{lemma_multinomial_noise_control}, and define
\[
\eta_{\mathrm{row}}
:=
\max_{i\in[p_1]}\sum_{k=1}^{p_{-1}} (D*D*Q)_{ik},
\qquad
\eta_{\mathrm{col}}
:=
\max_{k\in[p_{-1}]}\sum_{i=1}^{p_1} (D*D*Q)_{ik}.
\]
Then there exists an absolute constant $c_0>0$ such that for all $t\ge 1$, writing
\[
L:=t+\log(2p_1),
\]
we have
\begin{align}
\PP\Bigl(
\bigl\|\mathcal P_{\mathrm{off\text{-}diag}}(EE^\top)\bigr\|
\ge
c_0 \bigl[
&
L \sqrt{
\eta_{\mathrm{row}}\eta_{\mathrm{col}}
+
n\|D\|_\infty^2 \eta_{\mathrm{col}}
}
\notag\\
&\qquad\qquad
+
\|D\|_\infty^2 \sqrt{n}\,L^{3/2}
+
\|D\|_\infty^2 L
\bigr]
\Bigr)
\le
6e^{-t}.
\label{eq:mw-main-tail-corrected}
\end{align}
Consequently, for any $\tau\ge 2$, with probability at least
\[
1-6(2p_1)^{1-\tau},
\]
we have
\begin{align}
\bigl\|\mathcal P_{\mathrm{off\text{-}diag}}(EE^\top)\bigr\|
\le
C\Bigl[
&
\tau\log(2p_1)
\sqrt{
\eta_{\mathrm{row}}\eta_{\mathrm{col}}
+
n\|D\|_\infty^2 \eta_{\mathrm{col}}
}
\notag\\
&\qquad\qquad
+
\|D\|_\infty^2 \sqrt{n}\,
\bigl(\tau\log(2p_1)\bigr)^{3/2}
+
\|D\|_\infty^2 \tau\log(2p_1)
\Bigr].
\label{eq:mw-main-highprob-corrected}
\end{align}
\end{Lemma}

\begin{proof}
Since $Y\sim \mathrm{Multinomial}(n,P)$, we may write
\[
Y=\sum_{\ell=1}^n E_\ell,
\]
where $E_1,\dots,E_n$ are i.i.d. random matrices taking values in
\[
\{e_i e_k^\top: i\in[p_1],\ k\in[p_{-1}]\}
\]
with
\[
\PP(E_\ell=e_i e_k^\top)=P_{ik}.
\]
Let $X_\ell=(i,k)$ denote the category corresponding to $E_\ell=e_i e_k^\top$, and define
\[
A_\ell:=D*E_\ell.
\]
Then
\[
E=D*(Y-Q)=\sum_{\ell=1}^n (A_\ell-\mu),
\qquad
\mu:=\EE A_1=D*P=\frac1n(D*Q).
\]

For $x=(i,k)$, write $e_x=e_i e_k^\top$, and define
\[
h(x,x')
:=
\mathcal P_{\mathrm{off\text{-}diag}}
\bigl((D*e_x)(D*e_{x'})^\top\bigr).
\]
If $x=(i,k)$ and $x'=(j,\ell)$, then
\[
h(x,x')
=
\mathbf 1\{k=\ell,\ i\neq j\}\,D_{ik}D_{jk}\,e_i e_j^\top.
\]
In particular, $h(x,x)=0$, because when $x=x'$ the matrix
$(D*e_x)(D*e_x)^\top$ is diagonal and is removed by
$\mathcal P_{\mathrm{off\text{-}diag}}$.

Define
\[
\alpha(x):=\EE_{X'}\,h(x,X'),
\qquad
\beta(x):=\EE_{X'}\,h(X',x),
\qquad
\Gamma:=\EE h(X,X'),
\]
and
\[
K(x,x')
:=
h(x,x')-\alpha(x)-\beta(x')+\Gamma.
\]
Then $K$ is canonical in the sense that
\[
\EE_{X'}K(x,X')=0,
\qquad
\EE_X K(X,x')=0.
\]
Also, since $h(x',x)=h(x,x')^\top$, we have
\[
\beta(x)=\alpha(x)^\top,
\qquad
\Gamma=\Gamma^\top,
\qquad
K(x',x)=K(x,x')^\top.
\]

We first expand $\mathcal P_{\mathrm{off\text{-}diag}}(EE^\top)$. Since
\[
E=\sum_{\ell=1}^n (A_\ell-\mu),
\]
we have
\[
\mathcal P_{\mathrm{off\text{-}diag}}(EE^\top)
=
\sum_{\ell,m=1}^n
\mathcal P_{\mathrm{off\text{-}diag}}
\bigl((A_\ell-\mu)(A_m-\mu)^\top\bigr).
\]
For $\ell\neq m$, the canonical decomposition gives
\[
\mathcal P_{\mathrm{off\text{-}diag}}
\bigl((A_\ell-\mu)(A_m-\mu)^\top\bigr)
=
K(X_\ell,X_m).
\]
For $\ell=m$, since $A_\ell A_\ell^\top$ is diagonal and $h(x,x)=0$,
\[
\mathcal P_{\mathrm{off\text{-}diag}}
\bigl((A_\ell-\mu)(A_\ell-\mu)^\top\bigr)
=
-\alpha(X_\ell)-\beta(X_\ell)+\Gamma.
\]
Therefore
\[
\mathcal P_{\mathrm{off\text{-}diag}}(EE^\top)
=
U-S-S^\top-n\Gamma,
\]
where
\[
U:=\sum_{\ell\neq m} K(X_\ell,X_m),
\qquad
S:=\sum_{\ell=1}^n (\alpha(X_\ell)-\Gamma).
\]
Hence
\begin{equation}
\bigl\|\mathcal P_{\mathrm{off\text{-}diag}}(EE^\top)\bigr\|
\le
\|U\|+2\|S\|+n\|\Gamma\|.
\label{eq:offdiag-triangle-eta-col}
\end{equation}

\medskip
\noindent
{\bf Step 1: the deterministic term $n\Gamma$.}

Let
\[
B:=D*Q.
\]
Since $\mu=B/n$ and $\Gamma=\mathcal P_{\mathrm{off\text{-}diag}}(\mu\mu^\top)$,
\[
n\Gamma=\frac1n\,\mathcal P_{\mathrm{off\text{-}diag}}(BB^\top).
\]
Thus
\[
n\|\Gamma\|
\le
\frac1n\|BB^\top\|
\le
\frac1n\|BB^\top\|_{\infty\rightarrow\infty}
\le
\frac1n\|B\|_{1\rightarrow 1}\|B\|_{\infty\rightarrow\infty}.
\]
We now bound $\|B\|_{\infty\rightarrow\infty}$ and $\|B\|_{1\rightarrow 1}$ separately.

For the row norm,
\[
\|B\|_{\infty\rightarrow\infty}
=
\max_i\sum_{k}|D_{ik}|Q_{ik}.
\]
Applying Cauchy--Schwarz to the sum over $k$ gives
\[
\sum_k |D_{ik}|Q_{ik}
\le
\Bigl(\sum_k D_{ik}^2Q_{ik}\Bigr)^{1/2}
\Bigl(\sum_k Q_{ik}\Bigr)^{1/2}.
\]
Since $\sum_k D_{ik}^2Q_{ik}=\sum_k (D*D*Q)_{ik}\le \eta_{\mathrm{row}}$ and
$\sum_k Q_{ik}\le n$, we obtain
\[
\|B\|_{\infty\rightarrow\infty} \le \sqrt{n\,\eta_{\mathrm{row}}}.
\]
Similarly,
\[
\|B\|_{1\rightarrow 1}\le \sqrt{n\,\eta_{\mathrm{col}}}.
\]
Therefore
\begin{equation}
n\|\Gamma\|
\le
\sqrt{\eta_{\mathrm{row}}\eta_{\mathrm{col}}}.
\label{eq:gamma-det-eta-col}
\end{equation}

\medskip
\noindent
{\bf Step 2: the linear term $S$.}

For $x=(i,k)$,
\[
\alpha(x)
=
\sum_{j\neq i} P_{jk}D_{ik}D_{jk}\,e_i e_j^\top.
\]
Hence $\alpha(x)$ has only one nonzero row, namely the $i$th row. Therefore
its operator norm is just the Euclidean norm of that row:
\[
\|\alpha(x)\|^2
=
D_{ik}^2\sum_{j\neq i} P_{jk}^2D_{jk}^2
\le
D_{ik}^2\sum_{j} P_{jk}D_{jk}^2.
\]

Set
\[
a_{ik}:=(D*D*Q)_{ik},
\qquad
r_i:=\sum_{k=1}^{p_{-1}} a_{ik},
\qquad
c_k:=\sum_{i=1}^{p_1} a_{ik}.
\]
Then
\[
\eta_{\mathrm{row}}=\max_i r_i,
\qquad
\eta_{\mathrm{col}}=\max_k c_k.
\]
Since $P_{jk}=Q_{jk}/n$, the previous display becomes
\[
\|\alpha(i,k)\|^2
\le
\frac{D_{ik}^2}{n}\sum_j D_{jk}^2Q_{jk}
=
\frac{D_{ik}^2 c_k}{n}.
\]
In particular,
\[
\|\alpha(i,k)\|^2
\le
\frac{\|D\|_\infty^2 c_k}{n}
\le
\|D\|_\infty^4,
\]
because $c_k\le \sum_i \|D\|_\infty^2Q_{ik}\le n\|D\|_\infty^2$.

Also, by \eqref{eq:gamma-det-eta-col},
\[
\|\Gamma\|
\le
\frac{\sqrt{\eta_{\mathrm{row}}\eta_{\mathrm{col}}}}{n}
\le
\|D\|_\infty^2,
\]
since both $\eta_{\mathrm{row}}$ and $\eta_{\mathrm{col}}$ are at most
$n\|D\|_\infty^2$. Therefore the summands
\[
Z_\ell:=\alpha(X_\ell)-\Gamma
\]
satisfy
\begin{equation}
\|Z_\ell\|
\le
\|\alpha(X_\ell)\|+\|\Gamma\|
\le
2\|D\|_\infty^2.
\label{eq:linear-R-eta-col}
\end{equation}

We next bound the Bernstein variance parameter. For the left variance,
\[
\alpha(i,k)\alpha(i,k)^\top
=
D_{ik}^2
\sum_{j\neq i} P_{jk}^2D_{jk}^2\,e_i e_i^\top
\preceq
\frac{D_{ik}^2 c_k}{n}\,e_i e_i^\top.
\]
Taking expectation over $X=(i,k)\sim P$, we get
\[
\EE[\alpha(X)\alpha(X)^\top]
\preceq
\frac1n
\sum_{i,k} P_{ik}D_{ik}^2c_k\,e_i e_i^\top
=
\frac1{n^2}\sum_{i=1}^{p_1}
\Bigl(\sum_{k=1}^{p_{-1}} a_{ik}c_k\Bigr)e_i e_i^\top.
\]
Hence
\begin{equation}
n\bigl\|\EE[\alpha(X)\alpha(X)^\top]\bigr\|
\le
\frac1n\max_i\sum_k a_{ik}c_k
\le
\frac{\eta_{\mathrm{row}}\eta_{\mathrm{col}}}{n},
\label{eq:linear-var-left-eta-col}
\end{equation}
because $c_k\le \eta_{\mathrm{col}}$ and $\sum_k a_{ik}=r_i\le \eta_{\mathrm{row}}$.

For the right variance, note that
\[
\|\alpha(X)^\top\alpha(X)\|=\|\alpha(X)\|^2,
\]
so
\[
n\bigl\|\EE[\alpha(X)^\top\alpha(X)]\bigr\|
\le
n\,\EE\|\alpha(X)\|^2
\le
n\EE\!\left[\frac{D_X^2 c_{k(X)}}{n}\right].
\]
Now compute
\[
\EE[D_X^2 c_{k(X)}]
=
\sum_{i,k} P_{ik}D_{ik}^2 c_k
=
\frac1n\sum_{k=1}^{p_{-1}} c_k^2.
\]
Indeed,
\[
\sum_{i,k} P_{ik}D_{ik}^2 c_k
=
\frac1n\sum_{i,k} D_{ik}^2Q_{ik}c_k
=
\frac1n\sum_k c_k \sum_i D_{ik}^2Q_{ik}
=
\frac1n\sum_k c_k^2.
\]

To bound $\sum_k c_k^2$, we use the elementary inequality
\[
\sum_k c_k^2 \le \Bigl(\max_k c_k\Bigr)\sum_k c_k
= \eta_{\mathrm{col}}\sum_k c_k.
\]
Moreover,
\[
\sum_{k=1}^{p_{-1}} c_k
=
\sum_{i,k} D_{ik}^2Q_{ik}
\le
\|D\|_\infty^2 \sum_{i,k}Q_{ik}
=
n\|D\|_\infty^2,
\]
because $\sum_{i,k}Q_{ik}=n$. Therefore
\[
\sum_{k=1}^{p_{-1}} c_k^2
\le
\eta_{\mathrm{col}}\sum_k c_k
\le
n\|D\|_\infty^2\eta_{\mathrm{col}}.
\]
Substituting this back yields
\begin{equation}
n\bigl\|\EE[\alpha(X)^\top\alpha(X)]\bigr\|
\le
\|D\|_\infty^2\eta_{\mathrm{col}}.
\label{eq:linear-var-right-eta-col}
\end{equation}

Since
\[
\EE Z_\ell=0,
\qquad
Z_\ell=\alpha(X_\ell)-\Gamma,
\]
and
\[
ZZ^\top \preceq 2\alpha\alpha^\top+2\Gamma\Gamma^\top,
\qquad
Z^\top Z \preceq 2\alpha^\top\alpha+2\Gamma^\top\Gamma,
\]
we obtain
\begin{align}
n\Bigl\|\EE[ZZ^\top]\Bigr\|
&\le
2n\bigl\|\EE[\alpha\alpha^\top]\bigr\|
+
2n\|\Gamma\|^2
\notag\\
&\le
2\cdot \frac{\eta_{\mathrm{row}}\eta_{\mathrm{col}}}{n}
+
2n\cdot \frac{\eta_{\mathrm{row}}\eta_{\mathrm{col}}}{n^2}
\qquad\text{by \eqref{eq:linear-var-left-eta-col} and \eqref{eq:gamma2-eta-col-bis}}
\notag\\
&\le
C\frac{\eta_{\mathrm{row}}\eta_{\mathrm{col}}}{n},
\label{eq:linear-var-1-eta-col}
\end{align}
and
\begin{align}
n\Bigl\|\EE[Z^\top Z]\Bigr\|
&\le
2n\bigl\|\EE[\alpha^\top\alpha]\bigr\|
+
2n\|\Gamma\|^2
\notag\\
&\le
2\|D\|_\infty^2\eta_{\mathrm{col}}
+
2n\cdot \frac{\eta_{\mathrm{row}}\eta_{\mathrm{col}}}{n^2}
\qquad\text{by \eqref{eq:linear-var-right-eta-col} and \eqref{eq:gamma2-eta-col-bis}}
\notag\\
&\le
C\left(
\|D\|_\infty^2\eta_{\mathrm{col}}
+
\frac{\eta_{\mathrm{row}}\eta_{\mathrm{col}}}{n}
\right).
\label{eq:linear-var-2-eta-col}
\end{align}
Here we used the auxiliary bound
\begin{equation}
\|\Gamma\|^2
\le
\frac{\eta_{\mathrm{row}}\eta_{\mathrm{col}}}{n^2},
\label{eq:gamma2-eta-col-bis}
\end{equation}
which follows immediately from \eqref{eq:gamma-det-eta-col}.

Now apply matrix Bernstein to
\[
S=\sum_{\ell=1}^n Z_\ell.
\]
The summands are independent, mean zero, satisfy
\[
\|Z_\ell\|\le 2\|D\|_\infty^2
\]
by \eqref{eq:linear-R-eta-col}, and the variance parameter is bounded by
\eqref{eq:linear-var-1-eta-col}--\eqref{eq:linear-var-2-eta-col}. Therefore
\begin{equation}
\PP\left(
\|S\|
\ge
C\Bigl[
\sqrt{
L\left(
\frac{\eta_{\mathrm{row}}\eta_{\mathrm{col}}}{n}
+
\|D\|_\infty^2\eta_{\mathrm{col}}
\right)}
+
\|D\|_\infty^2L
\Bigr]
\right)
\le
2e^{-t}.
\label{eq:S-bernstein-eta-col}
\end{equation}

\medskip
\noindent
{\bf Step 3: the second-order term $U$.}

The term $U$ is the genuinely quadratic part of the decomposition. To control it,
we first symmetrize the kernel so that it becomes Hermitian, and then use a
decoupling inequality for matrix-valued $U$-statistics.

Define
\[
H(x,x'):=\frac12\bigl(K(x,x')+K(x',x)\bigr).
\]
Then $H(x,x')\in\mathbb H_{p_1}$ is Hermitian and symmetric:
\[
H(x,x')=H(x',x).
\]
Moreover, since $K$ is canonical,
\[
\EE_{X'}H(x,X')=0,
\qquad
\EE_XH(X,x')=0.
\]
Thus $H$ is again a canonical kernel.

Also,
\[
\sum_{\ell\neq m} H(X_\ell,X_m)
=
\frac12\sum_{\ell\neq m}K(X_\ell,X_m)
+
\frac12\sum_{\ell\neq m}K(X_m,X_\ell).
\]
The second sum is equal to the first one after relabeling the dummy indices
$(\ell,m)\leftrightarrow(m,\ell)$. Hence
\[
\sum_{\ell\neq m} H(X_\ell,X_m)=U.
\]

Let $\widetilde X_1,\dots,\widetilde X_n$ be an independent copy of
$X_1,\dots,X_n$, and define the decoupled statistic
\[
U^{\mathrm{dec}}
:=
\sum_{\ell\neq m} H(X_\ell,\widetilde X_m).
\]
By the decoupling inequality for Banach-space-valued $U$-statistics of order $2$
\cite[Theorem~1]{delapena1995decoupling}, there exists a universal constant
$C_{\mathrm{dec}}>0$ such that for all $s\ge 0$,
\begin{equation}
\PP(\|U\|\ge s)
\le
C_{\mathrm{dec}}\,
\PP(C_{\mathrm{dec}}\|U^{\mathrm{dec}}\|\ge s).
\label{eq:decoupling-eta-col}
\end{equation}
Thus it suffices to derive a tail bound for the decoupled statistic $U^{\mathrm{dec}}$, and we are going to apply \cite[Remark~3.11]{minskerwei2019moment}. 

We first bound the kernel size. Since
\[
\|h(x,x')\|\le \|D\|_\infty^2,
\qquad
\|\alpha(x)\|\le \|D\|_\infty^2,
\qquad
\|\beta(x)\|\le \|D\|_\infty^2,
\qquad
\|\Gamma\|\le \|D\|_\infty^2,
\]
we have
\[
\|K(x,x')\|
\le
\|h(x,x')\|+\|\alpha(x)\|+\|\beta(x')\|+\|\Gamma\|
\le
4\|D\|_\infty^2.
\]
Therefore
\[
\|H(x,x')\|
\le
\frac12\bigl(\|K(x,x')\|+\|K(x',x)\|\bigr)
\le
4\|D\|_\infty^2.
\]
So we may take
\[
M:=4\|D\|_\infty^2.
\]

Next we bound the square-function parameter
\[
v^2:=\EE\Bigl\|\EE_2 H^2(X,\widetilde X)\Bigr\|,
\]
where $\EE_2$ denotes expectation with respect to the second variable
$\widetilde X$ only.

Since
\[
H=\frac{K+K^\top}{2},
\]
we claim that
\[
H^2\preceq \frac12(KK^\top+K^\top K).
\]
Indeed, for any $u\in\mathbb R^{p_1}$,
\[
u^\top H^2 u
=
\|Hu\|_2^2
=
\frac14\|Ku+K^\top u\|_2^2
\le
\frac12\|Ku\|_2^2+\frac12\|K^\top u\|_2^2
=
u^\top\!\left(\frac12K^\top K+\frac12KK^\top\right)\!u.
\]
Since this holds for all $u$, the claimed matrix inequality follows. 
Hence, for every $x$,
\begin{equation}
\Bigl\|\EE_2 H^2(x,\widetilde X)\Bigr\|
\le
\frac12
\Bigl(
\bigl\|\EE_2[K(x,\widetilde X)K(x,\widetilde X)^\top]\bigr\|
+
\bigl\|\EE_2[K(x,\widetilde X)^\top K(x,\widetilde X)]\bigr\|
\Bigr).
\label{eq:H2-by-K-eta-col}
\end{equation}

Fix $x=(i,k)$. Since
\[
K(x,\widetilde X)=h(x,\widetilde X)-\alpha(x)-\beta(\widetilde X)+\Gamma,
\]
we bound $KK^\top$ and $K^\top K$ term-by-term. Using
\[
(B_1+B_2+B_3+B_4)(B_1+B_2+B_3+B_4)^\top
\preceq
4\sum_{m=1}^4 B_mB_m^\top,
\]
we get
\begin{align}
\EE_2[K K^\top](x,\widetilde X)
\preceq
4\Bigl(
&
\EE_2[h h^\top](x,\widetilde X)
+
\alpha(x)\alpha(x)^\top
\notag\\
&\qquad
+
\EE[\beta(\widetilde X)\beta(\widetilde X)^\top]
+
\Gamma\Gamma^\top
\Bigr),
\label{eq:KKt-eta-col}
\end{align}
and similarly
\begin{align}
\EE_2[K^\top K](x,\widetilde X)
\preceq
4\Bigl(
&
\EE_2[h^\top h](x,\widetilde X)
+
\alpha(x)^\top\alpha(x)
\notag\\
&\qquad
+
\EE[\beta(\widetilde X)^\top\beta(\widetilde X)]
+
\Gamma^\top\Gamma
\Bigr).
\label{eq:KtK-eta-col}
\end{align}

We now bound each term.

First,
\[
\EE_2[h(x,\widetilde X)h(x,\widetilde X)^\top]
=
\frac{D_{ik}^2}{n}\sum_{j\neq i} a_{jk}\,e_ie_i^\top,
\]
so
\begin{equation}
\Bigl\|\EE_2[h(x,\widetilde X)h(x,\widetilde X)^\top]\Bigr\|
\le
\frac{D_{ik}^2 c_k}{n}.
\label{eq:hht-eta-col}
\end{equation}
Likewise,
\[
\EE_2[h(x,\widetilde X)^\top h(x,\widetilde X)]
=
\frac{D_{ik}^2}{n}\sum_{j\neq i} a_{jk}\,e_je_j^\top,
\]
hence
\begin{equation}
\Bigl\|\EE_2[h(x,\widetilde X)^\top h(x,\widetilde X)]\Bigr\|
\le
\frac{D_{ik}^2 c_k}{n}.
\label{eq:hth-eta-col}
\end{equation}

Next,
\[
\alpha(x)\alpha(x)^\top
\preceq
\frac{D_{ik}^2 c_k}{n}\,e_ie_i^\top,
\qquad
\|\alpha(x)^\top\alpha(x)\|=\|\alpha(x)\|^2\le \frac{D_{ik}^2 c_k}{n}.
\]
Thus
\begin{equation}
\|\alpha(x)\alpha(x)^\top\|
\le
\frac{D_{ik}^2 c_k}{n},
\qquad
\|\alpha(x)^\top\alpha(x)\|
\le
\frac{D_{ik}^2 c_k}{n}.
\label{eq:alpha2-eta-col}
\end{equation}

For the $\beta$-terms, recall that $\beta(\widetilde X)=\alpha(\widetilde X)^\top$. Hence
\[
\|\beta(\widetilde X)\|^2
=
\|\alpha(\widetilde X)\|^2
\le
\frac{D_{\widetilde i\,\widetilde k}^2 c_{\widetilde k}}{n}.
\]
Taking expectation and using the same calculation as in Step 2,
\[
\EE\|\beta(\widetilde X)\|^2
\le
\frac1{n^2}\sum_{k=1}^{p_{-1}} c_k^2
\le
\frac{\|D\|_\infty^2\,\eta_{\mathrm{col}}}{n}.
\]
Therefore
\begin{equation}
\bigl\|\EE[\beta(\widetilde X)\beta(\widetilde X)^\top]\bigr\|
\le
\frac{\|D\|_\infty^2\,\eta_{\mathrm{col}}}{n},
\qquad
\bigl\|\EE[\beta(\widetilde X)^\top\beta(\widetilde X)]\bigr\|
\le
\frac{\|D\|_\infty^2\,\eta_{\mathrm{col}}}{n}.
\label{eq:beta2-eta-col}
\end{equation}

Finally, by \eqref{eq:gamma2-eta-col-bis},
\begin{equation}
\|\Gamma\|^2
\le
\frac{\eta_{\mathrm{row}}\eta_{\mathrm{col}}}{n^2}.
\label{eq:gamma2-eta-col}
\end{equation}

Combining \eqref{eq:H2-by-K-eta-col}--\eqref{eq:gamma2-eta-col}, we obtain for every $x=(i,k)$,
\begin{align*}
\Bigl\|\EE_2 H^2(x,\widetilde X)\Bigr\|
&\le
\frac12
\Bigl(
\bigl\|\EE_2[K(x,\widetilde X)K(x,\widetilde X)^\top]\bigr\|
+
\bigl\|\EE_2[K(x,\widetilde X)^\top K(x,\widetilde X)]\bigr\|
\Bigr)
\\
&\le
2\Bigl(
\bigl\|\EE_2[h(x,\widetilde X)h(x,\widetilde X)^\top]\bigr\|
+
\|\alpha(x)\alpha(x)^\top\|
+
\bigl\|\EE[\beta(\widetilde X)\beta(\widetilde X)^\top]\bigr\|
+
\|\Gamma\Gamma^\top\|
\Bigr)
\\
&\qquad
+
2\Bigl(
\bigl\|\EE_2[h(x,\widetilde X)^\top h(x,\widetilde X)]\bigr\|
+
\|\alpha(x)^\top\alpha(x)\|
+
\bigl\|\EE[\beta(\widetilde X)^\top\beta(\widetilde X)]\bigr\|
+
\|\Gamma^\top\Gamma\|
\Bigr)
\\
&\le
C\left(
\frac{D_{ik}^2 c_k}{n}
+
\frac{\|D\|_\infty^2\eta_{\mathrm{col}}}{n}
+
\frac{\eta_{\mathrm{row}}\eta_{\mathrm{col}}}{n^2}
\right).
\end{align*}
Averaging over $X\sim P=Q/n$ gives
\[
v^2
\le
C\left(
\frac1n \EE[D_X^2 c_{k(X)}]
+
\frac{\|D\|_\infty^2\eta_{\mathrm{col}}}{n}
+
\frac{\eta_{\mathrm{row}}\eta_{\mathrm{col}}}{n^2}
\right).
\]
As above,
\[
\EE[D_X^2 c_{k(X)}]
=
\frac1n\sum_{k=1}^{p_{-1}} c_k^2
\le
\|D\|_\infty^2\eta_{\mathrm{col}}.
\]
Therefore
\begin{equation}
v^2
\le
C\left(
\frac{\|D\|_\infty^2\eta_{\mathrm{col}}}{n}
+
\frac{\eta_{\mathrm{row}}\eta_{\mathrm{col}}}{n^2}
\right).
\label{eq:v2-eta-col}
\end{equation}

Now apply \cite[Remark~3.11]{minskerwei2019moment} to the decoupled statistic
$U^{\mathrm{dec}}$, with
\[
t_*:=t+\log C_{\mathrm{dec}},
\qquad
L_*:=t_*+\log(2p_1).
\]
Then $t_*\ge 1$, $L_*=L+O(1)\le C L$, and
\[
e^{-t_*}=C_{\mathrm{dec}}^{-1}e^{-t}.
\]
Hence
\[
\PP\left(
\|U^{\mathrm{dec}}\|
\ge
C\Bigl[
nL_*\sqrt{v^2}
+
\|D\|_\infty^2\sqrt n\,L_*^{3/2}
\Bigr]
\right)
\le
C_{\mathrm{dec}}^{-1}e^{-t}.
\]
Combining this with \eqref{eq:decoupling-eta-col}, and absorbing the factor
$C_{\mathrm{dec}}$ and the harmless replacement $L_*\lesssim L$ into the absolute constant, we obtain
\begin{equation}
\PP\left(
\|U\|
\ge
C\Bigl[
nL\sqrt{v^2}
+
\|D\|_\infty^2\sqrt n\,L^{3/2}
\Bigr]
\right)
\le
e^{-t}.
\label{eq:U-pre-tail-eta-col}
\end{equation}
Using \eqref{eq:v2-eta-col},
\[
nL\sqrt{v^2}
\le
CL\sqrt{
\eta_{\mathrm{row}}\eta_{\mathrm{col}}
+
n\|D\|_\infty^2\eta_{\mathrm{col}}
}.
\]
Therefore
\begin{equation}
\PP\left(
\|U\|
\ge
C\Bigl[
L\sqrt{
\eta_{\mathrm{row}}\eta_{\mathrm{col}}
+
n\|D\|_\infty^2\eta_{\mathrm{col}}
}
+
\|D\|_\infty^2\sqrt n\,L^{3/2}
\Bigr]
\right)
\le
e^{-t}.
\label{eq:U-tail-eta-col}
\end{equation}

\medskip
\noindent
{\bf Step 4: combine the bounds.}

Set
\[
s_U
:=
C\Bigl[
L\sqrt{
\eta_{\mathrm{row}}\eta_{\mathrm{col}}
+
n\|D\|_\infty^2\eta_{\mathrm{col}}
}
+
\|D\|_\infty^2\sqrt n\,L^{3/2}
\Bigr],
\]
\[
s_S
:=
C\Bigl[
\sqrt{
L\left(
\frac{\eta_{\mathrm{row}}\eta_{\mathrm{col}}}{n}
+
\|D\|_\infty^2\eta_{\mathrm{col}}
\right)}
+
\|D\|_\infty^2L
\Bigr],
\qquad
s_\Gamma:=\sqrt{\eta_{\mathrm{row}}\eta_{\mathrm{col}}}.
\]
By \eqref{eq:offdiag-triangle-eta-col},
\[
\PP\Bigl(
\bigl\|\mathcal P_{\mathrm{off\text{-}diag}}(EE^\top)\bigr\|
\ge
s_U+2s_S+s_\Gamma
\Bigr)
\le
\PP(\|U\|\ge s_U)+\PP(\|S\|\ge s_S).
\]
Using \eqref{eq:S-bernstein-eta-col} and \eqref{eq:U-tail-eta-col},
\[
\PP\Bigl(
\bigl\|\mathcal P_{\mathrm{off\text{-}diag}}(EE^\top)\bigr\|
\ge
s_U+2s_S+s_\Gamma
\Bigr)
\le
e^{-t}+2e^{-t}
\le
6e^{-t}.
\]

It remains to absorb $s_S$ and $s_\Gamma$ into the main scale. Since $L\ge 1$ and $n\ge 1$,
\[
s_\Gamma
=
\sqrt{\eta_{\mathrm{row}}\eta_{\mathrm{col}}}
\le
L\sqrt{
\eta_{\mathrm{row}}\eta_{\mathrm{col}}
+
n\|D\|_\infty^2\eta_{\mathrm{col}}
}.
\]
Also,
\[
\sqrt{
L\left(
\frac{\eta_{\mathrm{row}}\eta_{\mathrm{col}}}{n}
+
\|D\|_\infty^2\eta_{\mathrm{col}}
\right)}
\le
L\sqrt{
\eta_{\mathrm{row}}\eta_{\mathrm{col}}
+
n\|D\|_\infty^2\eta_{\mathrm{col}}
}.
\]
Therefore
\[
s_U+2s_S+s_\Gamma
\le
C\Bigl[
L\sqrt{
\eta_{\mathrm{row}}\eta_{\mathrm{col}}
+
n\|D\|_\infty^2\eta_{\mathrm{col}}
}
+
\|D\|_\infty^2\sqrt n\,L^{3/2}
+
\|D\|_\infty^2L
\Bigr].
\]
This proves \eqref{eq:mw-main-tail-corrected}.

Finally, take
\[
t=(\tau-1)\log(2p_1),\qquad \tau\ge 2.
\]
Then
\[
L=t+\log(2p_1)=\tau\log(2p_1).
\]
Substituting this into \eqref{eq:mw-main-tail-corrected} yields
\eqref{eq:mw-main-highprob-corrected}.
\end{proof}

\begin{Lemma}[Smallest singular value of a balanced sign matrix]
\label{lem:balanced-sign-smin}
Let $p$ be even and let $R\le c_0 p$ for a sufficiently small absolute
constant $c_0>0$. Let $\xi_1,\ldots,\xi_R$ be independent random vectors drawn
uniformly from
\[
    \Bigl\{\xi\in\{\pm1\}^p:\sum_{i=1}^p \xi_i=0\Bigr\},
\]
and let $E=(\xi_1,\ldots,\xi_R)\in\mathbb R^{p\times R}$. Then there exist
absolute constants $c,C>0$ such that
\[
    \mathbb P\left\{
        \sigma_{\min}(E)\ge c\sqrt p
    \right\}
    \ge
    1-\exp(-cp).
\]
\end{Lemma}

\begin{proof}
Let $\mathbb H=\{\mathbf 1\}^{\perp}\subset\mathbb R^p$ and let
$O\in\mathbb R^{p\times(p-1)}$ be an orthonormal basis matrix for $\mathbb H$.
Since each $\xi_r$ is balanced, $\xi_r\in\mathbb H$. Define
\[
    \zeta_r
    =
    \sqrt{\frac{p-1}{p}}\,O^\top \xi_r
    \in\mathbb R^{p-1},
    \qquad r=1,\ldots,R,
\]
and let $Z=(\zeta_1,\ldots,\zeta_R)\in\mathbb R^{(p-1)\times R}$.

We first record two properties of $\zeta_r$. Since a balanced sign
vector satisfies
\[
    \mathbb E \xi_r\xi_r^\top
    =
    \frac{p}{p-1}
    \left(
        I_p-\frac{1}{p}\mathbf 1\mathbf 1^\top
    \right),
\]
we have
\[
    \mathbb E \zeta_r\zeta_r^\top
    =
    \frac{p-1}{p}
    O^\top
    \mathbb E \xi_r\xi_r^\top
    O
    =
    I_{p-1}.
\]
Thus $\zeta_r$ is isotropic in $\mathbb R^{p-1}$. Moreover, for every
$v\in\mathbb S^{p-2}$,
\[
    \langle \zeta_r,v\rangle
    =
    \sqrt{\frac{p-1}{p}}\,
    \langle \xi_r,Ov\rangle .
\]
The vector $\xi_r$ is obtained by sampling exactly $p/2$ signs equal to $1$
and $p/2$ signs equal to $-1$. Hence Hoeffding's inequality for sampling
without replacement implies
\[
    \mathbb P\left(
        |\langle \zeta_r,v\rangle|\ge t
    \right)
    \le
    2\exp(-c t^2),
    \qquad t\ge 0,
\]
uniformly over $v\in\mathbb S^{p-2}$. Therefore $\zeta_r$ is an isotropic
sub-Gaussian random vector with sub-Gaussian norm bounded by an absolute
constant.

By the lower singular value bound for rectangular matrices with
independent isotropic sub-Gaussian columns; see, for example,
\cite[Theorem~5.58]{vershynin2010introduction}, applied with
$N=p-1$ and $n=R$, there exist absolute constants $c_1,C_1>0$ such that,
for every $t\ge0$,
\[
    \mathbb P\left\{
        \sigma_{\min}(Z)
        \le
        \sqrt{p-1}-C_1\sqrt R-t
    \right\}
    \le
    2\exp(-c_1t^2).
\]
Here the assumptions of the theorem are satisfied because the columns
$\zeta_1,\ldots,\zeta_R$ are independent, isotropic, sub-Gaussian random
vectors in $\mathbb R^{p-1}$ with sub-Gaussian norm bounded by an absolute
constant, and
\[
    \|\zeta_r\|_2
    =
    \sqrt{\frac{p-1}{p}}\|\xi_r\|_2
    =
    \sqrt{p-1}
\]
almost surely.

Since each column of $E$ belongs to $\mathbb H$ and
\[
    Z
    =
    \sqrt{\frac{p-1}{p}}\,O^\top E,
\]
the nonzero singular values of $Z$ and
$\sqrt{(p-1)/p}\,E$ coincide. Hence
\[
    \sigma_{\min}(E)
    =
    \sqrt{\frac{p}{p-1}}\,
    \sigma_{\min}(Z)
    \ge
    c\sqrt p
\]
with probability at least $1-\exp(-cp)$.

\end{proof}

\end{sloppypar}

\end{document}